\documentclass[a4paper]{article}
\usepackage{enumitem}
\usepackage{calc}

\usepackage{lmodern} 

\usepackage{graphicx} 

\usepackage{amsmath, accents}
\usepackage{amssymb,tikz}
\usepackage{pict2e}
\usepackage{amsthm}
\usepackage{amscd}
\usepackage{mathrsfs}
\usepackage{todonotes}

\usetikzlibrary{calc} 

\usepackage{yfonts}

\usepackage{marvosym}

\newtheorem{theorem}			     {Theorem} [section]
	
\newtheorem{corollary}	  [theorem]	 {Corollary}	
		
\theoremstyle{definition}

\newtheorem{remark} {Remark}

\newcommand{\C}{\mathbb{C}}
\newcommand{\R}{\mathbb{R}}
\newcommand{\N}{\mathbb{N}}

\usepackage{tikz}
\usetikzlibrary{arrows}
\usetikzlibrary{decorations.pathmorphing}
\usetikzlibrary{decorations.markings}
\usetikzlibrary{patterns}
\usetikzlibrary{automata}
\usetikzlibrary{positioning}
\usepackage{tikz-cd}

\usepackage{pgfplots}

\numberwithin{equation}{section}

\def\ds{\displaystyle}
\def\bigO{{\cal O}}
\begin{document}
\title{Exponential moments and piecewise thinning \\for the Bessel point process}
\author{Christophe Charlier\footnote{Department of Mathematics, KTH Royal Institute of Technology, Lindstedtsv\"{a}gen 25, SE-114 28 Stockholm, Sweden. e-mail: cchar@kth.se}}

\maketitle

\begin{abstract}
We obtain exponential moment asymptotics for the Bessel point process. As a direct consequence, we improve on the asymptotics for the expectation and variance of the associated counting function, and establish several central limit theorems. We show that exponential moment asymptotics can also be interpreted as large gap asymptotics, in the case where we apply the operation of a piecewise constant thinning on several consecutive intervals. We believe our results also provide important estimates for later studies of the global rigidity of the Bessel point process.



\end{abstract}

\section{Introduction and statement of results}
A point process is a model for a collection of points randomly located on some underlying space $\Lambda$. Determinantal point processes \cite{Soshnikov,Borodin,Johansson} are characterized by the fact that their correlation functions have a determinantal structure. More precisely, if $\rho_{k}$ denotes the $k$-th correlation function of a given point process $X$, then $X$ is \textit{determinantal} if there exists a function $K:\Lambda^{2} \to \mathbb{R}$ such that for any $k > 0$, we have
\begin{align*}
\rho_{k}(x_{1},\ldots,x_{k}) = \det \big( K(x_{i},x_{j}) \big)_{i,j=1}^{k}, \qquad \mbox{for all } x_{1},\ldots,x_{k} \in \Lambda.
\end{align*}
$K$ encodes all the information about the process and is called the \textit{correlation kernel} (or simply the \textit{kernel}). Determinantal point processes appear in various domains of mathematics, such as random matrix theory, orthogonal polynomials and tiling models \cite{Deift,BG}. These processes exhibit a rich analytic and algebraic structure, and several tools have been developed to study them.   

\paragraph{Bessel point process.} In this paper, we focus on the Bessel point process. This is a determinantal point process on $\R^+:=[0,+\infty)$ whose kernel is given by
\begin{equation}\label{Bessel kernel}
K_{\alpha}^{\mathrm{Be}}(x,y) = \frac{J_{\alpha}(\sqrt{x})\sqrt{y}J_{\alpha}^{\prime}(\sqrt{y})-\sqrt{x}J_{\alpha}^{\prime}(\sqrt{x})J_{\alpha}(\sqrt{y})}{2(x-y)}, \qquad \alpha > -1,
\end{equation}
where $\alpha$ is a parameter of the process which quantifies the attraction (if $\alpha < 0$) or repulsion (if $\alpha > 0$) between the particles and the origin, and $J_\alpha$ is the Bessel function of the first kind of order $\alpha$. 

\vspace{0.2cm}The Bessel point process appears typically in repulsive particle systems, when the particles are accumulating along a natural boundary (called ``hard edge"). This is one of the canonical point processes from the theory of random matrices. It encodes the behavior of the eigenvalues of certain positive definite matrices near the origin \cite{For1,ForNag}. This point process appears also in, among other applications, non-intersecting squared Bessel paths \cite{KuijMartFinWie}.

\vspace{0.2cm}Given a Borel set $B \subseteq \R^+$, the \textit{occupancy number} $N_B$ is the random variable defined as the number of points that fall into $B$. Determinantal point processes are always locally finite, i.e. $N_B$ is finite with probability $1$ for $B$ bounded. Moreover, all particles are distinct with probability $1$. 

\paragraph{Exponential moments and piecewise constant thinning.}  Let us introduce the parameters
\begin{equation}
m \in \N_{>0}, \quad \vec{s}=(s_1,...,s_m) \in \mathbb{C}^{m} \quad \mbox{ and } \quad \vec{x} = (x_{1},...,x_{m}) \in (\mathbb{R}^{+})^{m},
\end{equation}
where $\vec{x}$ is such that $0 =: x_{0} < x_1 < x_2 < ... < x_m < +\infty$. We are interested in exponential moments of the form
\begin{equation}\label{def of F}
F_{\alpha}(\vec{x},\vec{s}) := \mathbb{E}_{\alpha}\Bigg[ \prod_{j=1}^{m} s_{j}^{N_{(x_{j-1},x_{j})}} \Bigg] = \sum_{k_{1},...,k_{m} \geq 0} \mathbb{P}_{\alpha}\Bigg(\bigcap_{j=1}^{m}N_{(x_{j-1},x_{j})}=k_{j}\Bigg)\prod_{j=1}^{m} s_{j}^{k_{j}}.
\end{equation}
If $s_{j}=0$ and $k_{j} = 0$ in \eqref{def of F} for a certain $j \in \{1,\ldots,m\}$, then $s_{j}^{k_{j}}$ should be interpreted as being equal to $1$. The function $F_{\alpha}(\vec{x},\vec{s})$ is an entire function in $s_{1},\ldots,s_{m}$, and can be written as a Fredholm determinant with $m$ discontinuities (by \cite[Theorem 2]{Soshnikov}). This function is also known as the joint probability generating function of occupancy numbers on consecutive intervals, and contains a lot of information about the Bessel process: several probabilistic quantities can be deduced from it, such as the expectation and variance of $N_{(x_{j-1},x_{j})}$ (see also \cite{ChDoe} for further applications). Here, we show that if $s_{j}\in [0,1]$ for all $j=1,\ldots,m$, $F_{\alpha}(\vec{x},\vec{s})$ admits another application related to thinning.

\medskip

 The operation of thinning is well-known in the theory of point processes, see e.g. \cite{IPSS}. It has been first studied in the context of random matrices by Bohigas and Pato \cite{BohigasPato1,BohigasPato2}, and first set out rigorously in \cite{BotDeiItsKra1} (see also \cite{BDIK2017,ChCl2,BothnerBuckingham,Charlier, BDIK2019,CharlierGharakhloo} for further results). A constant thinning consists of removing each particle independently with the same probability $s \in  [0,1]$. As $s$ increases, the level of correlation decreases, so the thinned point process interpolates between the original point process (when $s=0$) and an uncorrelated process (when $s \to 1$ at a certain speed) \cite{Kallenberg}. Here, we consider a more general situation and apply a piecewise constant thinning as follows: for each $j \in \{1,2,...,m\}$, each particle on the interval $[x_{j-1},x_{j})$ is removed with probability $s_{j}$. The probability to observe a gap on $(0,x_{m})$ in this thinned point process is given by
\begin{multline}\label{gap probability}
\sum_{k_{1},...,k_{m} \geq 0} \mathbb{P}_{\alpha}\Bigg(\bigcap_{j=1}^{m}N_{[x_{j-1},x_{j})}=k_{j}\Bigg)\prod_{j=1}^{m} \mathbb{P}\left( \begin{array}{c}
\mbox{the } k_{j} \mbox{ points on } [x_{j-1},x_{j}) \mbox{ have} \\
\mbox{been removed by the thinning}
\end{array} \right) \\
=\sum_{k_{1},...,k_{m} \geq 0} \mathbb{P}_{\alpha}\Bigg(\bigcap_{j=1}^{m}N_{(x_{j-1},x_{j})}=k_{j}\Bigg)\prod_{j=1}^{m} s_{j}^{k_{j}},
\end{multline}
which is precisely $F_{\alpha}(\vec{x},\vec{s})$.\footnote{To be precise, we have used the fact that the boundary of the successive intervals $(x_{j-1},x_{j})$ carry particles with probability $0$.}

\vspace{0.2cm} In \cite{TraWidLUE}, Tracy and Widom have studied $F_{\alpha}(x_{1},s_{1})$, i.e. the case $m = 1$. If $s_{1} \in [0,1]$, $F_{\alpha}(x_{1},s_{1})$ is the gap probability on $(0,x_{1})$ in the constant thinned Bessel point process, with the thinning parameter given by $s_{1}$. They expressed $F_{\alpha}(x_{1},s_{1})$ in terms of the solution of a Painlev\'{e} V equation. The generalization of this result for an arbitrary $m\geq 1$ was recently obtained in \cite{ChDoe}, where it is shown that $F_{\alpha}(\vec{x},\vec{s})$ can be expressed identically in terms of a solution to a system of $m$ coupled Painlev\'{e} V equations.

\paragraph{Exponential moment asymptotics.} Let us now scale the size of the intervals with a new parameter $r>0$, that is, we consider $F_{\alpha}(r\vec{x},\vec{s})$. As $r$ decreases and tends to $0$, the intervals $(rx_{j-1},rx_{j})$ shrink and asymptotics for $F_{\alpha}(r\vec{x},\vec{s})$ can be obtained rather straightforwardly from its representation as a Fredholm series. In this paper, we address the harder problem of finding the so-called \textit{exponential moment asymptotics}, which are asymptotics for $F_{\alpha}(r\vec{x},\vec{s})$ as $r \to + \infty$. If $s_{j} \in [0,1]$ for all $j=1,\ldots,m$, it follows from \eqref{gap probability} that exponential moment asymptotics can be interpreted as \textit{large gap asymptotics} in the piecewise thinned point process. 

\medskip


We also elaborate here on another main motivation to the study of exponential moment asymptotics.  For a large class of point processes taken from the theory of random matrices, the fluctuations of an \textit{individual} point  around its typical position are well-understood \cite{Johansson98, Gustavsson}. An important problem in recent years has been to estimate the \textit{maximal} fluctuation of the points. A bound for the maximal fluctuation that holds true with high probability is called a \textit{global rigidity bound}. The first results of this nature can be found in \cite{ErdosSchleinYau} (for an upper bound) and in \cite{Gustavsson} (for a lower bound). Sharp bounds have now been obtained for several point processes \cite{ArguinBeliusBourgade, ChhaibiMadauleNajnudel, HolcombPaquette, PaquetteZeitouni, LambertCircular, LambertGinibre, CFLW}. As it turns out, some of the techniques that have been used successfully to derive \textit{optimal} rigidity bounds rely heavily on exponential moment asymptotics.

\medskip

By studying asymptotics for $F_{\alpha}(r\vec{x},\vec{s})$ as $r \to + \infty$, another main objective of this paper is precisely to provide important estimates for later studies of the global rigidity of the Bessel point process. We comment more on that in Remark \ref{remark: rigidity} below. We should mention that, while the interpretation of $F_{\alpha}(r\vec{x},\vec{s})$ as a gap probability only makes sense for $s_{1},...,s_{m} \in [0,1]$, the estimates needed to study the global rigidity require to consider the more general situation $s_{1},...,s_{m} \in [0,+\infty)$.

\paragraph{Known results.} For $m=1$, the large $r$ asymptotics of $F_{\alpha}(r\vec{x},\vec{s})=F_{\alpha}(rx_{1},s_{1})$ are already known. There are two distinct regimes: 1) $s_{1}=0$ and 2) $s_{1}>0$. 

\medskip \textit{1) The case $s_{1}=0$.} Using a connection between $F_{\alpha}(rx_{1},0)$ and a solution to the Painlev\'{e} V equation, Tracy and Widom \cite{TraWidLUE} gave an heuristic derivation of the following
\begin{equation}\label{asymp m=1 with s1 = 0}
F_{\alpha}(rx_{1},0) = \tau_{\alpha} \, (rx_{1})^{-\frac{\alpha^{2}}{4}}e^{-\frac{rx_{1}}{4}+\alpha \sqrt{rx_{1}}}\Big( 1+\bigO(r^{-1/2}) \Big),  \quad r \to +\infty,
\end{equation}
for some constant $\tau_{\alpha}$. They also noted that for $\alpha = \mp \frac{1}{2}$, $K_{\alpha}^{\mathrm{Be}}$ reduces to sine-kernels appearing in orthogonal and symplectic ensembles for which large gap asymptotics are known from the work of Dyson \cite{Dyson}. Using this observation and supported with numerical calculations, they conjectured that $\tau_{\alpha} = G(1+\alpha)/(2\pi)^{\frac{\alpha}{2}}$, where $G$ is Barnes' $G$-function. A proof of the asymptotics \eqref{asymp m=1 with s1 = 0} (including the constant) was first given by Ehrhardt in \cite{Ehr} for $\alpha \in (-1,1)$ using operator theory methods and finally for all values of $\alpha>-1$ by Deift, Krasovsky and Vasilevska in \cite{DeiftKrasVasi} by performing a Deift/Zhou \cite{DeiftZhou1992} steepest descent on a Riemann-Hilbert (RH) problem.

\medskip \textit{2) The case $s_{1}=e^{u_{1}}>0$.} Bothner, Its and Prokhorov in \cite[Eq (1.35)]{BIP2019} recently proved that
\begin{equation}\label{asymp m=1 with s1 > 0}
F_{\alpha}(rx_{1},e^{u_{1}}) = G( 1+\tfrac{u_{1}}{2\pi i} )G( 1-\tfrac{u_{1}}{2\pi i} )e^{- \frac{\alpha}{2}u_{1}} (4\sqrt{rx_{1}})^{\frac{u_{1}^{2}}{4\pi^{2}}}e^{\frac{u_{1}}{\pi}\sqrt{rx_{1}}}\Big( 1+\bigO(r^{-1/2}) \Big),  \quad r \to + \infty,
\end{equation}
uniformly for $s_{1}=e^{u_{1}}$ in compact subsets of $(0,+\infty)$.



\newpage 

The main contribution of this paper is to obtain large $r$ asymptotics for $F_{\alpha}(r \vec{x},\vec{s})$ up to and including the constant term in two different situations: in Theorem \ref{thm:s1 neq 0} below, we assume $s_{1},...,s_{m} \in (0,+\infty)$ and in Theorem \ref{thm:s1=0}, we assume $s_{1}=0$ and $s_{2},...,s_{m} \in (0,+\infty)$. The cases $m=2$ of Theorem \ref{thm:s1 neq 0} and $m=2$ of Theorem \ref{thm:s1=0} are already new results. These theorems contain a lot of information about the Bessel process, and we illustrate this by providing some direct consequences: in Corollaries \ref{coro: s1>0 consequence of thm} and \ref{coro: exp and var s1 = 0}, we obtain new asymptotics for the expectation, variance and covariance of the counting function, and in Corollaries \ref{coro: CLTs} and \ref{coro: CLTs 3}, we establish several central limit theorems.



\vspace{0.2cm}We mention that large gap asymptotics in the two cases of Theorems \ref{thm:s1 neq 0} and \ref{thm:s1=0} cannot be treated both at once. In fact, a critical transition occurs as $s_{1} \to 0$ and simultaneously $r \to + \infty$. This transition is expected to be described in terms of elliptic $\theta$-functions and is not addressed in this paper. For a heuristic discussion, see Section \ref{subsection: heuristic discussion}. 

\paragraph{Main results for $s_{1} > 0$.} Asymptotics for $F_{\alpha}(r\vec{x},\vec{s})$ as $r \to + \infty$ with $\vec{s}=(s_1,...,s_m) \in (0,+\infty)^{m}$ are more elegantly described after making the following change of variables
\begin{equation}\label{def tj thm s1 neq 0}
u_{j} = \left\{ \begin{array}{l l}
\ds \log \frac{s_{j}}{s_{j+1}} & \mbox{for } j = 1,...,m-1, \\[0.3cm]
\ds \log s_{m} & \mbox{for } j=m.
\end{array} \right.
\end{equation}
Note that $u_{j} \in (-\infty,+\infty)$. Then, with the notation $\vec{u}=(u_{1},\ldots,u_{m})$, we define
\begin{equation}\label{def of E alpha}
E_{\alpha}(\vec{x},\vec{u}) := \mathbb{E}_{\alpha}\Bigg[ \prod_{j=1}^{m} e^{u_{j}N_{(0,x_{j})}} \Bigg] =
F_{\alpha}(\vec{x},\vec{s}).
\end{equation}
\begin{theorem}\label{thm:s1 neq 0}
Let $\alpha > -1$, $m \in \mathbb{N}_{>0}$,
\begin{align*}
\vec{u}=(u_1,...,u_m) \in \mathbb{R}^{m} \qquad \mbox{ and } \qquad \vec{x} = (x_{1},...,x_{m}) \in (\mathbb{R}^{+})^{m}
\end{align*}
be such that $0 < x_1 < x_2 < ... < x_m < +\infty$.  As $r \to + \infty$, we have
\begin{multline}\label{F asymptotics thm s1 neq 0}
E_{\alpha}(r\vec{x},\vec{u}) = \exp \Bigg( \sum_{j=1}^{m}  u_{j} \mu_{\alpha}(rx_{j}) + \sum_{j=1}^{m} \frac{u_{j}^{2}}{2} \sigma^{2}(rx_{j}) + \sum_{1 \leq j < k \leq m} u_{j} u_{k} \Sigma(x_{k},x_{j})  \\ + \sum_{j=1}^{m} \log G(1+\tfrac{u_{j}}{2\pi i})G(1-\tfrac{u_{j}}{2\pi i}) + \bigO \bigg( \frac{\log r}{\sqrt{r}} \bigg) \Bigg),
\end{multline}
where $G$ is Barnes' $G$-function, and $\mu_{\alpha}$, $\sigma^{2}$ and $\Sigma$ are given by
\begin{equation}\label{mu sigma and cov}
\mu_{\alpha}(x) = \frac{\sqrt{x}}{\pi} - \frac{\alpha}{2}, \qquad \sigma^{2}(x) = \frac{\log(4 \sqrt{x})}{2 \pi^{2}}, \qquad \Sigma(x_{k},x_{j}) = \frac{1}{2\pi^{2}} \log \frac{\sqrt{x_{k}}+\sqrt{x_{j}}}{\sqrt{x_{k}}-\sqrt{x_{j}}}.
\end{equation}
Alternatively, the asymptotics \eqref{F asymptotics thm s1 neq 0} can be rewritten as
\begin{equation}\label{thm product s1 neq 0}
E_{\alpha}(r\vec{x},\vec{u}) = \exp\Bigg(\sum_{1 \leq j <k \leq m}u_{j}u_{k}\Sigma(x_{k},x_{j})\Bigg)\prod_{j=1}^{m}E_{\alpha}(rx_{j},u_{j}) \bigg( 1+\bigO \bigg( \frac{\log r}{\sqrt{r}} \bigg) \bigg).
\end{equation}
Furthermore, the error terms in \eqref{F asymptotics thm s1 neq 0} and \eqref{thm product s1 neq 0} are uniform in $u_{1},...,u_{m}$ in compact subsets of $(-\infty,+\infty)$ and uniform in $x_{1},...,x_{m}$ in compact subsets of $(0,+\infty)$, as long as there exists $\delta > 0$ such that
\begin{equation}\label{condition on xj in terms of delta}
\min_{1 \leq j < k \leq m} x_{k}-x_{j} \geq \delta.
\end{equation}
\end{theorem}
\begin{remark}
The Barnes $G$-function is a well-known special function defined by
\begin{equation*}
G(1+z) = (2\pi)^{z/2} \exp \left( -  \frac{z + z^{2}(1+\gamma_{\mathrm{E}})}{2} \right) \prod_{k=1}^{\infty} \left[ \left( 1 + \frac{z}{k} \right)^{k} \exp \left( \frac{z^{2}}{2k}-z \right) \right],
\end{equation*}
where $\gamma_{\mathrm{E}} \approx 0.5772$ is the Euler's gamma constant. It is related to the $\Gamma$ function by the functional relation $G(z+1) = \Gamma(z)G(z)$ (see \cite{NIST} for further properties).
\end{remark}
\begin{remark}\label{remark: heuristic}
As can be seen by comparing \eqref{F asymptotics thm s1 neq 0} with \eqref{asymp m=1 with s1 = 0}, there is a notable difference between the cases $s_{1}=0$ and $s_{1}>0$ in the large $r$ asymptotics of $\log F_{\alpha}(rx_{1},s_{1})$: the leading term is of order $\bigO(r)$ if $s_{1}=0$, while it is of order $\bigO(\sqrt{r})$ if $s_{1}$ is bounded away from $0$. We also recall that $F_{\alpha}(rx_{1},s_{1})$ can be interpreted as a gap probability for $s_{1} \in [0,1]$. In this case, the difference between the leading terms means that it is significantly more likely to observe a large gap in the thinned Bessel point process (roughly $\sim e^{-c \, \sqrt{r}}$, $c>0$) rather than in the unthinned Bessel point process ($\sim e^{-c \, r}$). From \eqref{def of F}, for general values of $s_{1} \in [0,+\infty)$ we have
\begin{equation}\label{sum gap thinned}
F_{\alpha}(rx_{1},s_{1}) = \sum_{k \geq 0} \mathbb{P}_{\alpha}\big(N_{(0,rx_{1})}=k\big) s_{1}^{k}.
\end{equation}
Using Theorem \ref{thm:s1 neq 0}, it is possible to understand which terms in the sum contribute the most when $s_{1}$ is bounded away from $0$. The first term is simply $\mathbb{P}_{\alpha}\big(N_{(0,rx_{1})}=0\big)$, which we know by \eqref{asymp m=1 with s1 = 0} to be roughly $\sim e^{-c \, r}$ for large $r$. In fact, for typical realizations of the Bessel point process, the number of points lying on $(0,rx_{1})$ is of order $\bigO(\sqrt{r})$ as $r \to + \infty$, see \eqref{Soshnikov expo and variance}. Thus, when $k$ is of order $\bigO(\sqrt{r})$, $\mathbb{P}_{\alpha}\big(N_{(0,rx_{1})}=k\big)$ is of order $\bigO(1)$ and $s_{1}^{k}$ is $\bigO(e^{-c \, \sqrt{r}})$. Hence by Theorem \ref{thm:s1 neq 0}, for large $r$ and for $s_{1}$ bounded away from $0$, the dominant terms in \eqref{sum gap thinned} are those for which $k$ is of order $\bigO(\sqrt{r})$.

\end{remark}
\begin{remark}\label{remark: rigidity}
By combining Theorem \ref{thm:s1 neq 0} with \cite[Theorem 1.2]{ChCl4}, the following global rigidity upper bound for the Bessel point process has been established in \cite{ChCl4}: for any $\epsilon>0$, we have
\begin{align}\label{upper bound rigidity}
\lim_{k_0\to\infty}\mathbb P\left( \sup_{k \geq k_{0}} \frac{|\frac{1}{\pi}b_k^{1/2}-k|}{\log k}\leq \frac{1}{\pi} + \epsilon\right)=1,
\end{align}
where $b_{k}$ is the $k$-th smallest point in the Bessel point process. In fact, the proof of \eqref{upper bound rigidity} uses only the first exponential moment asymptotics \eqref{asymp m=1 with s1 > 0}, or equivalently, \eqref{F asymptotics thm s1 neq 0} with $m=1$. This upper bound is expected to be sharp, although this has not been proved (see also \cite[Remark 1.3]{ChCl4} for a heuristic discussion, and \cite[Figure 2]{ChCl4} for numerical simulations). Sharp lower bounds are in general very difficult to obtain, and require, among other things, the asymptotics of the second exponential moment (see e.g. \cite{CFLW}). These asymptotics are also provided by Theorem \ref{thm:s1 neq 0} with $m=2$.
\end{remark}
\begin{remark}
From \eqref{thm product s1 neq 0} we see that, asymptotically, the $m$-th exponential moment can be written as the product of two terms: the first term is a constant pre-factor which depends only on the constants $\Sigma(x_{k},x_{j})$, and the second term is the product of $m$ first exponential moment. Note that this phenomenon holds also for the Airy point process, see \cite{ChCl3}.
\end{remark}

In \cite[Theorem 2]{SoshnikovSineAiryBessel}, Soshnikov obtained the following asymptotics for the expectation and variance of the counting function as $r \to + \infty$
\begin{align}\label{Soshnikov expo and variance}
\mathbb{E}_{\alpha}[N_{(0,rx)}] = \frac{\sqrt{rx}}{\pi} + \bigO (1), \quad \mbox{ and } \quad \mathrm{Var}_{\alpha}[N_{(0,rx)}] = \frac{\log rx}{4 \pi^{2}} + \bigO (1).
\end{align}
Theorem \ref{thm:s1 neq 0} allows to improve on these asymptotics.
\begin{corollary}\label{coro: s1>0 consequence of thm}
Let $x>0$ and $x_{2}>x_{1}>0$ be fixed. As $r \to + \infty$, we have
\begin{align}
& \mathbb{E}_{\alpha}[N_{(0,rx)}] = \mu_{\alpha}(rx) + \bigO \bigg( \frac{\log r}{\sqrt{r}} \bigg) = \frac{\sqrt{rx}}{\pi} - \frac{\alpha}{2} + \bigO \bigg( \frac{\log r}{\sqrt{r}} \bigg), \label{expected and variance asymp 1} \\
& \mathrm{Var}_{\alpha}[N_{(0,rx)}] = \sigma^{2}(rx) + \frac{1 + \gamma_{\mathrm{E}}}{2 \pi^{2}} + \bigO \bigg( \frac{\log r}{\sqrt{r}} \bigg) \nonumber \\
& \hspace{2cm} = \frac{\log rx}{4 \pi^{2}} + \frac{1 + \log 4 + \gamma_{\mathrm{E}}}{2 \pi^{2}} + \bigO \bigg( \frac{\log r}{\sqrt{r}} \bigg), \label{expected and variance asymp 2} \\
& \mathrm{Cov}_{\alpha}[ N_{(0,rx_{1})},N_{(0,rx_{2})} ] = \Sigma(x_{2},x_{1}) + \bigO \bigg( \frac{\log r}{\sqrt{r}} \bigg) = \frac{1}{2\pi^{2}} \log \frac{\sqrt{x_{2}}+\sqrt{x_{1}}}{\sqrt{x_{2}}-\sqrt{x_{1}}} + \bigO \bigg( \frac{\log r}{\sqrt{r}} \bigg), \label{asymp for the cov thm} 
\end{align}
where $\gamma_{\mathrm{E}} \approx 0.5772$ is Euler's gamma constant.
\end{corollary}
\begin{remark}
Basic properties of the expectation and the variance imply that
\begin{align}
& \mathbb{E}_{\alpha}[N_{(rx_{1},rx_{2})}] = \mu_{\alpha}(rx_{2})-\mu_{\alpha}(rx_{1}) + \bigO \bigg( \frac{\log r}{\sqrt{r}} \bigg), \label{expected and variance asymp 3} \\
& \mathrm{Var}_{\alpha}[N_{(rx_{1},rx_{2})}] = \sigma^{2}(rx_{1}) + \sigma^{2}(rx_{2}) + \frac{1+\gamma_{\mathrm{E}}}{\pi^{2}} - 2 \, \Sigma(x_{2},x_{1}) + \bigO \bigg( \frac{\log r}{\sqrt{r}} \bigg) \nonumber \\
& \hspace{2.42cm} = \frac{\log r}{2\pi^{2}} + \frac{\log (16 \sqrt{x_{1}x_{2}})}{2\pi^{2}} + \frac{1+\gamma_{\mathrm{E}}}{\pi^{2}} - \frac{1}{\pi^{2}} \log \frac{\sqrt{x_{2}}+\sqrt{x_{1}}}{\sqrt{x_{2}}-\sqrt{x_{1}}} + \bigO \bigg( \frac{\log r}{\sqrt{r}} \bigg). \label{expected and variance asymp 4}
\end{align}
In \cite[Theorem 2]{SoshnikovSineAiryBessel}, Soshnikov obtained \footnote{The quantity $\nu_{k}(r)$ in \cite[Theorem 2]{SoshnikovSineAiryBessel} corresponds to $N_{(0,kr)}-N_{(0,(k-1)r)} = N_{((k-1)r,kr)}$ in our notation.}
\begin{align}\label{error Soshnikov}
\mathrm{Var}_{\alpha}[N_{(r,2r)}] = \frac{\log r}{4\pi^{2}}  + \bigO(1),
\end{align}
which does not agree with the leading term of \eqref{expected and variance asymp 4} (a factor $2$ is missing). There is a similar error in \cite[Theorem 1]{SoshnikovSineAiryBessel} for $k \geq 2$, see \cite[Remark 1]{ChCl3} (and furthermore the constant $\frac{11}{12\pi^{2}}$ in \cite[Theorem 1]{SoshnikovSineAiryBessel} should instead be $\frac{3}{4\pi^{2}}$).
\end{remark}
\begin{proof}[Proof of Corollary \ref{coro: s1>0 consequence of thm}]
Expanding \eqref{def of E alpha} for $m=1$ as $u \to 0$, we obtain
\begin{equation}\label{expansion F m=1}
E_{\alpha}(x,u) = 1 + u \, \mathbb{E}_{\alpha}[N_{(0,x)}] + \frac{u^{2}}{2} \mathbb{E}_{\alpha}[N_{(0,x)}^{2}] + \bigO(u^{3}).
\end{equation}
On the other hand, we can also expand the large $r$ asymptotics of $E_{\alpha}(rx,u)$ (given by the right-hand side of \eqref{F asymptotics thm s1 neq 0} with $m=1$) as $u \to 0$, since these asymptotics are valid uniformly for $u$ in compact subsets of $ \mathbb{R}$ (in particular in a neighborhood of $0$). Comparing this expansion with \eqref{expansion F m=1}, we obtain \eqref{expected and variance asymp 1} and \eqref{expected and variance asymp 2}, where $\gamma_{\mathrm{E}}$ comes from the expansion of the Barnes' $G$-functions (see \cite[formula 5.17.3]{NIST}). The covariance between the two occupancy numbers $N_{(0,x_{1})}$ and $N_{(0,x_{2})}$ can be obtained from \eqref{def of E alpha} with $m=2$ as follows:
\begin{equation}\label{cov in the expansion}
\begin{array}{r c l}
\ds \frac{E_{\alpha}\big((x_{1},x_{2}),(u,u)\big)}{E_{\alpha}(x_{1},u)E_{\alpha}(x_{2},u)} & = & \ds \frac{\mathbb{E}_{\alpha}\big[ e^{u N_{(0,x_{1})}}e^{u N_{(0,x_{2})}} \big]}{\mathbb{E}_{\alpha}\big[ e^{u N_{(0,x_{1})}} \big]\mathbb{E}_{\alpha}\big[ e^{u N_{(0,x_{2})}} \big]}  \\[0.4cm]
& = &  \ds 1 + \mbox{Cov}_{\alpha}(N_{(0,x_{1})},N_{(0,x_{2})})u^{2} + \bigO(u^{3}), \quad \mbox{as } u \to 0.
\end{array}
\end{equation}
After the rescaling $(x_{1},x_{2}) \mapsto r(x_{1},x_{2})$, we can obtain large $r$ asymptotics for the left-hand side of the above expression using Theorem \ref{thm:s1 neq 0}. 
By an expansion as $u \to 0$ of these asymptotics, and a comparison with \eqref{cov in the expansion}, we obtain \eqref{asymp for the cov thm}.
\end{proof}
We can also deduce from Theorem \ref{thm:s1 neq 0} the following central limit theorems. Corollary \ref{coro: CLTs} (a) is the most natural central limit theorem that follows from Theorem \ref{thm:s1 neq 0}, while Corollary \ref{coro: CLTs} (b) allows for a comparison with \cite[Theorem 2]{SoshnikovSineAiryBessel}.
\begin{corollary}\label{coro: CLTs}
(a) Consider the random variables $N_{j}^{(r)}$ defined by
\begin{align*}
N_{j}^{(r)} = \frac{N_{(0,rx_{j})}-\mu_{\alpha}(rx_{j})}{\sqrt{\sigma^{2}(rx_{j})}}, \qquad j=1,\ldots,m.
\end{align*}
As $r \to +\infty$, we have
\begin{align}\label{convergence in distribution 1}
\big( N_{1}^{(r)},N_{2}^{(r)},\ldots,N_{m}^{(r)}\big) \quad \overset{d}{\longrightarrow} \quad \mathcal{N}(\vec{0},I_{m}),
\end{align}
where ``$\overset{d}{\longrightarrow}$" means convergence in distribution, $I_{m}$ is the $m \times m$ identity matrix, and $\mathcal{N}(\vec{0},I_{m})$ is a multivariate normal random variable of mean $\vec{0}=(0,\ldots,0)$ and covariance matrix $I_{m}$.

\medskip

\noindent (b) Consider the random variables $\widehat{N}_{j}^{(r)}$ defined by
\begin{align*}
\widehat{N}_{1}^{(r)} = \frac{N_{(0,rx_{1})}-\mu_{\alpha}(rx_{1})}{\sqrt{\sigma^{2}(rx_{1})}}, \qquad 
\widehat{N}_{j}^{(r)} = \frac{N_{(rx_{j-1},rx_{j})}-\big(\mu_{\alpha}(rx_{j}) - \mu_{\alpha}(rx_{j-1})\big)}{\sqrt{\sigma^{2}(rx_{j})+\sigma^{2}(rx_{j-1})}}, \quad j=2,\ldots,m.
\end{align*}
As $r \to +\infty$, we have
\begin{align}
\big( \widehat{N}_{1}^{(r)},\widehat{N}_{2}^{(r)},\ldots,\widehat{N}_{m}^{(r)}\big) \quad \overset{d}{\longrightarrow} \quad \mathcal{N}(\vec{0},\Sigma_{m}),
\end{align}
where $\Sigma_{m}$ is given by
\begin{align*}
& (\Sigma_{m})_{1,j} = (\Sigma_{m})_{j,1} = \delta_{1,j}-\tfrac{1}{\sqrt{2}}\delta_{2,j}, & & 1 \leq j  \leq m. \\
& (\Sigma_{m})_{i,j} = \delta_{i,j}-\tfrac{1}{2}\delta_{i,j+1}- \tfrac{1}{2}\delta_{i,j-1}, & & 2 \leq i,j  \leq m.
\end{align*}
\end{corollary}
\begin{remark}
In Corollary \ref{coro: CLTs} (b), we have $(\Sigma_{m})_{1,2} = (\Sigma_{m})_{2,1} = -1/\sqrt{2}$, and this does not agree with \cite[Theorem 2]{SoshnikovSineAiryBessel} (where it was obtained that $(\Sigma_{m})_{1,2} = (\Sigma_{m})_{2,1} = -1/2$). This is related to the factor $2$ missing in \eqref{error Soshnikov}.
\end{remark}
\begin{proof}[Proof of Corollary \ref{coro: CLTs}.]
(a) After changing variables $u_{j} = t_{j}/\sqrt{\sigma^{2}(rx_{j})}$ in \eqref{F asymptotics thm s1 neq 0} (recall that $E_{\alpha}$ is given by \eqref{def of E alpha}), we obtain the following asymptotics
\begin{align*}
\mathbb{E}_{\alpha} \Bigg[ \exp\bigg(\sum_{j=1}^{m} t_{j}N_{j}^{(r)} \bigg) \Bigg] = \exp \Bigg( \sum_{j=1}^{m} \frac{t_{j}^{2}}{2} + \bigO \bigg( \frac{1}{\log r} \bigg) \Bigg), \qquad \mbox{as } r \to + \infty.
\end{align*}
Thus, for each $t_{1},\ldots,t_{m} \in \mathbb{R}^{m}$, $(N_{1}^{(r)},\ldots,N_{m}^{(r)})_{r \geq 1}$ is a sequence of random variables whose moment generating functions converge to $\exp \big( \sum_{j=1}^{m} t_{j}^{2}/2 \big)$ as $r \to + \infty$. The convergence is pointwise in $t_{1},\ldots,t_{m} \in \mathbb{R}^{m}$, but this is sufficient to imply the convergence in distribution \eqref{convergence in distribution 1} (see e.g. \cite{Billingsley}).

\medskip

\noindent (b) First, we make the change of variables $u_{j}=v_{j}-v_{j+1}$, $j =1,\ldots,m$ with $v_{m+1}:=0$ in \eqref{F asymptotics thm s1 neq 0}. As $r \to + \infty$, we then have
\begin{align*}
& \mathbb{E}_{\alpha} \Bigg[ \exp\bigg(\sum_{j=1}^{m} v_{j}N_{(rx_{j-1},rx_{j})} \bigg) \Bigg] = \exp \Bigg( v_{1} \mu_{\alpha}(rx_{1})+\sum_{j=2}^{m}  v_{j} \big(\mu_{\alpha}(rx_{j})-\mu_{\alpha}(rx_{j-1})\big) \\
& +\frac{v_{1}^{2}}{2}\sigma^{2}(rx_{1}) + \sum_{j=2}^{m} \frac{v_{j}^{2}}{2}\big( \sigma^{2}(rx_{j-1})+\sigma^{2}(rx_{j})\big) -  \sum_{j=2}^{m} v_{j-1}v_{j}\sigma^{2}(rx_{j-1}) \\ 
&+ \sum_{1 \leq j < k \leq m} (v_{j}-v_{j+1}) (v_{k}-v_{k+1}) \Sigma(x_{k},x_{j})+ \sum_{j=1}^{m} \log G(1+\tfrac{v_{j}-v_{j+1}}{2\pi i})G(1-\tfrac{v_{j}-v_{j+1}}{2\pi i}) + \bigO \bigg( \frac{\log r}{\sqrt{r}} \bigg) \Bigg).
\end{align*}
Second, we apply the rescaling $v_{1} = t_{1}/\sqrt{\sigma^{2}(rx_{1})}$ and $v_{j} = t_{j}/\sqrt{\sigma^{2}(rx_{j})+\sigma^{2}(rx_{j-1})}$ for $j \geq 2$. This gives
\begin{align*}
\mathbb{E}_{\alpha} \Bigg[ \exp\bigg(\sum_{j=1}^{m} t_{j}\widehat{N}_{j}^{(r)} \bigg) \Bigg] = \exp \Bigg( \frac{1}{2}\sum_{j=1}^{m}t_{j}^{2}-\frac{1}{\sqrt{2}}t_{1}t_{2}-\frac{1}{2}\sum_{j=2}^{m-1}t_{j}t_{j+1} + \bigO \bigg( \frac{1}{\log r} \bigg) \Bigg), \qquad \mbox{as } r \to + \infty.
\end{align*}
As in part (a), the pointwise convergence in $t_{1},\ldots,t_{m} \in \mathbb{R}^{m}$ is sufficient to conclude.
\end{proof}

\paragraph{Main results for $s_{1} = 0$.} This situation gives information about the Bessel point process, when we condition on the event that no particle lies on the left-most interval, i.e. $N_{(0,x_{1})} = 0$. Conditional point processes have been studied in other contexts in e.g. \cite{ChCl2}. For $s_{1}=0$ and $s_{2},\ldots,s_{m}\in (0,+\infty)$, we define
\begin{equation}\label{def tj thm s1 = 0}
u_{j} = \left\{ \begin{array}{l l}
\ds \log \frac{s_{j}}{s_{j+1}} & \mbox{for } j = 2,...,m-1, \\[0.3cm]
\ds \log s_{m} & \mbox{for } j=m,
\end{array} \right.
\end{equation}
and consider the following conditional expectation
\begin{align}\label{def of Ec}
E_{\alpha}^{c}(\vec{x},\vec{u}) = \mathbb{E}_{\alpha}^{c}\Bigg[ \prod_{j=2}^{m} e^{u_{j}N_{(x_{1},x_{j})}} \Bigg] := \mathbb{E}_{\alpha}\Bigg[ \prod_{j=2}^{m} e^{u_{j}N_{(x_{1},x_{j})}} \Big| N_{(0,x_{1})} = 0 \Bigg] =
\frac{F_{\alpha}(\vec{x},\vec{s})}{F_{\alpha}(x_{1},0)},
\end{align}
where we recall (see \eqref{def of F}) that $F_{\alpha}(x_{1},0) = \mathbb{P}_{\alpha} ( N_{(0,x_{1})} = 0 )$.
\begin{theorem}\label{thm:s1=0}
Let $m \in \mathbb{N}_{>0}$, $\alpha > -1$,
\begin{align*}
\vec{u}=(u_2,...,u_m) \in \mathbb{R}^{m-1} \qquad \mbox{ and } \qquad \vec{x} = (x_{1},...,x_{m}) \in (\mathbb{R}^{+})^{m}
\end{align*}
be such that $0 < x_1 < x_2 < ... < x_m < +\infty$. As $r \to + \infty$, we have
\begin{multline}\label{F asymptotics thm s1 = 0}
E_{\alpha}^{c}(r\vec{x},\vec{u}) = \exp \Bigg( \sum_{j=2}^{m}  u_{j} \, \tilde{\mu}_{\alpha}(r,x_{j}) + \sum_{j=2}^{m} \frac{u_{j}^{2}}{2} \sigma^{2}(r(x_{j}-x_{1})) + \sum_{2 \leq j < k \leq m} u_{j} u_{k} \Sigma(x_{k}-x_{1},x_{j}-x_{1}) \\  + \sum_{j=2}^{m} \log G(1+\tfrac{u_{j}}{2\pi i})G(1-\tfrac{u_{j}}{2\pi i}) + \bigO \bigg( \frac{\log r}{\sqrt{r}} \bigg) \Bigg),
\end{multline}
where the functions $\sigma^{2}$ and $\Sigma$ are defined in \eqref{mu sigma and cov}, and $\tilde{\mu}_{\alpha}$ is given by
\begin{align*}
\tilde{\mu}_{\alpha}(r,x) = \frac{\sqrt{r(x-x_{1})}}{\pi} - \frac{\alpha}{\pi} \arccos\bigg(\frac{\sqrt{x_{1}}}{\sqrt{x}}\bigg).
\end{align*}
Furthermore, the error term is uniform in $u_{2},...,u_{m}$ in compact subsets of $(-\infty,+\infty)$ and uniform in $x_{1},...,x_{m}$ in compact subsets of $(0,+\infty)$, as long as there exists $\delta > 0$ such that
\begin{equation}\label{condition on xj in terms of delta s1 = 0}
\min_{1 \leq j < k \leq m} x_{k}-x_{j} \geq \delta.
\end{equation}
\end{theorem}
We can deduce from Theorem \ref{thm:s1=0} the following new asymptotics for the expectation and variance of the counting function. The proof is similar to the one of Corollary \ref{coro: s1>0 consequence of thm}, and we omit it.
\begin{corollary}\label{coro: exp and var s1 = 0}
Let $x_{3}>x_{2}>x_{1}>0$ be fixed. As $r \to + \infty$, we have
\begin{align*}
& \mathbb{E}_{\alpha}\big[ N_{(rx_{1},rx_{2})} | N_{(0,rx_{1})} = 0 \big] = \tilde{\mu}(r,x_{2}) + \bigO \bigg( \frac{\log r}{\sqrt{r}} \bigg), \\
&  \mathrm{Var}_{\alpha}\big[ N_{(rx_{1},rx_{2})} | N_{(0,rx_{1})} = 0 \big] = \sigma^{2}(r(x_{2}-x_{1})) + \frac{1 + \gamma_{\mathrm{E}}}{2 \pi^{2}} + \bigO \bigg( \frac{\log r}{\sqrt{r}} \bigg), \\
& \mathrm{Cov}_{\alpha}\Big[ (N_{(rx_{1},rx_{2})}| N_{(0,rx_{1})}=0) \, , \, (N_{(rx_{1},rx_{3})}| N_{(0,rx_{1})}=0) ] = \Sigma(x_{3}-x_{1},x_{2}-x_{1}) + \bigO \bigg( \frac{\log r}{\sqrt{r}} \bigg).
\end{align*}
\end{corollary}
We also obtain the following central limit theorem. The proof is similar to the one done in Corollary \ref{coro: CLTs} and we omit it.
\begin{corollary}\label{coro: CLTs 3}
Consider the conditional random variables $\widetilde{N}_{j}^{(r)}$ defined by
\begin{align*}
\widetilde{N}_{j}^{(r)} = \frac{(N_{(rx_{1},rx_{j})}|N_{(0,rx_{1})} =0)-\tilde{\mu}(r,x_{j})}{\sqrt{\sigma^{2}(r(x_{j}-x_{1}))}}, \qquad j=2,\ldots,m.
\end{align*}
As $r \to +\infty$, we have
\begin{align}
\big(\widetilde{N}_{2}^{(r)},\ldots,\widetilde{N}_{m}^{(r)}\big) \quad \overset{d}{\longrightarrow} \quad \mathcal{N}(\vec{0},I_{m-1}).
\end{align}
\end{corollary}
\paragraph{Outline.} Section \ref{section: model rh problem} is divided into two parts. In the first part, we recall a model RH problem (whose solution is denoted $\Phi$) introduced in \cite{ChDoe}, which is of central importance in the present paper. In the second part, we obtain a differential identity which expresses $\partial_{s_k} \log F_{\alpha}(r \vec{x},\vec{s})$ (for an arbitrary $k\in\{1,...,m\}$) in terms of $\Phi$. We obtain large $r$ asymptotics for $\Phi$ with $s_{1} \in (0,+\infty)$ in Section \ref{Section: Steepest descent with s1>0} via a Deift/Zhou steepest descent. In Section \ref{Section: integration s1 >0}, we use the analysis of Section \ref{Section: Steepest descent with s1>0} to obtain large $r$ asymptotics for $\partial_{s_k} \log F_{\alpha}(r \vec{x},\vec{s})$. We also proceed with successive integrations of these asymptotics in $s_{1},...,s_{m}$, which proves Theorem \ref{thm:s1 neq 0}. Section \ref{Section: Steepest descent with s1=0} and Section \ref{Section: integration s1 =0} are devoted to the proof of Theorem \ref{thm:s1=0} (with $s_{1} = 0$), and are organised similarly to Section \ref{Section: Steepest descent with s1>0} and Section \ref{Section: integration s1 >0}.

\paragraph{Approach.} In \cite{DeiftKrasVasi}, the authors obtained the asymptotics \eqref{asymp m=1 with s1 = 0} by expressing $F_{\alpha}(rx_{1},0)$ as a limit as $n \to + \infty$ of $n \times n$ Toeplitz determinants (whose symbol has an hard edge) and then performing a steepest descent on an RH problem for orthogonal polynomials on the unit circle. The parameter $n$ is thus an extra parameter which disappears in the limit. It is a priori possible for us to generalize the same strategy by relating $F_{\alpha}(r\vec{x},\vec{s})$ with Toeplitz determinants (with jump-type Fisher-Hartwig singularities accumulating near an hard-edge), but on a technical level this appears not obvious at all. Our approach takes advantage of the known result \eqref{asymp m=1 with s1 = 0} (only needed to prove Theorem \ref{thm:s1=0}, but not Theorem \ref{thm:s1 neq 0}), and is more direct in the sense that the parameter $n$ does not appear in the analysis.

\section{Model RH problem for $\Phi$ and a differential identity}\label{section: model rh problem}
As mentioned in the outline, the model RH problem introduced in \cite{ChDoe} is of central importance in the present paper, and we recall its properties here. In order to have compact and uniform notations, it is convenient for us to define $x_{0} = 0$ and $s_{m+1} = 1$, but they are not included in the notations for $\vec{x}$ and $\vec{s}$. To summarize, the parameters $x_{0}$, $s_{m+1}$, $\vec{x} = (x_{1},...,x_{m})$ and $\vec{s} = (s_{1},...,s_{m})$ are such that
\begin{equation}\label{param for phi}
0=x_{0}<x_{1}<...<x_{m}<+\infty, \qquad s_{1},...,s_{m} \in [0,+\infty) \quad \mbox{ and } s_{m+1} = 1.
\end{equation}
The model RH problem we consider depends on $\alpha$, $\vec{x}$ and $\vec{s}$, and its solution is denoted by $\Phi(z;\vec{x},\vec{s})$, where the dependence in $\alpha$ is omitted. When there is no confusion,  we will just denote it by $\Phi(z)$ where the dependence in $\vec{x}$ and $\vec{s}$ is also omitted. There is existence (if the parameters satisfy \eqref{param for phi}) and uniqueness for $\Phi$, and furthermore it satisfies $\det \Phi \equiv 1$. The RH problem for $\Phi$ is more easily stated in terms of the following matrices:
\begin{equation}
\sigma_{3} = \begin{pmatrix}
1 & 0 \\ 0 & -1
\end{pmatrix}, \qquad \sigma_{+} = \begin{pmatrix}
0 & 1 \\ 0 & 0
\end{pmatrix}, \qquad N = \frac{1}{\sqrt{2}}\begin{pmatrix}
1 & i \\ i & 1
\end{pmatrix}.
\end{equation}
We also define for $y \in \mathbb{R}$ the following piecewise constant matrix:
\begin{equation}\label{def of H}
H_{y}(z) = \left\{  \begin{array}{l l}

I, & \mbox{ for } -\frac{2\pi}{3}< \arg(z-y)< \frac{2\pi}{3},\\

\begin{pmatrix}
1 & 0 \\
-e^{\pi i \alpha} & 1 \\
\end{pmatrix}, & \mbox{ for } \frac{2\pi}{3}< \arg(z-y)< \pi, \\

\begin{pmatrix}
1 & 0 \\
e^{-\pi i \alpha} & 1 \\
\end{pmatrix}, & \mbox{ for } -\pi< \arg(z-y)< -\frac{2\pi}{3}, \\

\end{array} \right.
\end{equation}
where the principal branch is chosen for the argument, such that $\arg(z-y) = 0$ for $z>y$.

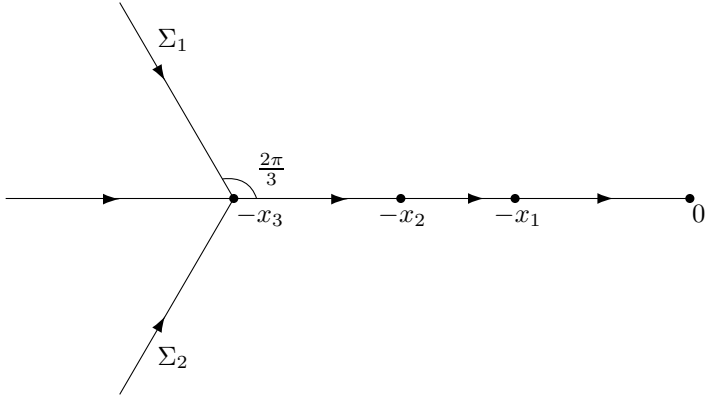
\begin{figure}[t]
    \begin{center}
    \setlength{\unitlength}{1truemm}
    \begin{picture}(100,55)(20,10)
        \put(50,40){\line(1,0){60}}
        \put(50,40){\line(-1,0){30}}
        \put(50,40){\thicklines\circle*{1.2}}
        \put(72,40){\thicklines\circle*{1.2}}
        \put(87,40){\thicklines\circle*{1.2}}
        \put(100,40){\thicklines\vector(1,0){.0001}}
        \put(83,40){\thicklines\vector(1,0){.0001}}
        \put(110,40){\thicklines\circle*{1.2}}
        \put(50,40){\line(-0.5,0.866){15}}
        \put(50,40){\line(-0.5,-0.866){15}}
        \qbezier(53,40)(52,43)(48.5,42.598)
        \put(53,43){$\frac{2\pi}{3}$}
        \put(50.3,36.8){$-x_{3}$}
        \put(69,36.8){$-x_{2}$}
        \put(84.2,36.8){$-x_{1}$}
        \put(110.3,36.8){$0$}
        \put(40,60){$\Sigma_{1}$}
        \put(40,18){$\Sigma_{2}$}
        \put(65,39.9){\thicklines\vector(1,0){.0001}}
        \put(35,39.9){\thicklines\vector(1,0){.0001}}
        \put(41,55.588){\thicklines\vector(0.5,-0.866){.0001}}
        \put(41,24.412){\thicklines\vector(0.5,0.866){.0001}}
    \end{picture}
    \caption{\label{figPhi}The jump contour for $\Phi$ with $m=3$.}
\end{center}
\end{figure}

\subsubsection*{RH problem for $\Phi$}
\begin{itemize}
\item[(a)] $\Phi : \mathbb{C}\setminus \Sigma_{\Phi} \to \mathbb{C}^{2\times 2}$ is analytic, where the contour $\Sigma_{\Phi} = ((-\infty,0]\cup \Sigma_{1} \cup \Sigma_{2})$ is oriented as shown in Figure \ref{figPhi} with
\begin{equation*}
\Sigma_{1} = -x_{m}+ e^{\frac{2\pi i}{3}}\mathbb{R}^{+}, \qquad \Sigma_{2} = -x_{m}+ e^{-\frac{2\pi i}{3}}\mathbb{R}^{+}.
\end{equation*}
\item[(b)] The limits of $\Phi(z)$ as $z$ approaches $\Sigma_{\Phi}\setminus \{0,-x_{1},...,-x_{m}\}$ from the left ($+$ side) and from the right ($-$ side) exist, are continuous on $\Sigma_{\Phi}\setminus \{0,-x_{1},...,-x_{m}\}$ and are denoted by $\Phi_+$ and $\Phi_-$ respectively. Furthermore they are related by:
\begin{align}
& \Phi_{+}(z) = \Phi_{-}(z) \begin{pmatrix}
1 & 0 \\ e^{\pi i \alpha} & 1
\end{pmatrix}, & & z \in \Sigma_{1}, \\
& \Phi_{+}(z) = \Phi_{-}(z) \begin{pmatrix}
0 & 1 \\ -1 & 0
\end{pmatrix}, & & z \in (-\infty,-x_{m}), \\
& \Phi_{+}(z) = \Phi_{-}(z) \begin{pmatrix}
1 & 0 \\ e^{-\pi i \alpha} & 1
\end{pmatrix}, & & z \in \Sigma_{2}, \\
& \Phi_{+}(z) = \Phi_{-}(z) \begin{pmatrix}
e^{\pi i \alpha} & s_{j} \\ 0 & e^{-\pi i \alpha}
\end{pmatrix}, & & z \in (-x_{j},-x_{j-1}),
\end{align}
where $j =1,...,m$.
\item[(c)] As $z \to \infty$, we have 
\begin{equation}\label{Phi inf}
\Phi(z) = \Big(I + \Phi_{1}z^{-1}+\bigO(z^{-2})\Big)z^{-\frac{\sigma_{3}}{4}}Ne^{\sqrt{z}\sigma_{3}},
\end{equation}
where the principal branch is chosen for each root, and the matrix $\Phi_{1} = \Phi_{1}(\vec{x},\vec{s})$ is independent of $z$ and traceless.

As $z$ tends to $-x_{j}$, $j \in \{1,...,m\}$, $\Phi$ takes the form
\begin{equation}\label{def of Gj}
\Phi(z) = G_{j}(z) \begin{pmatrix}
1 & \frac{s_{j+1}-s_{j}}{2\pi i} \log (z+x_{j}) \\ 0 & 1
\end{pmatrix} V_{j}(z) e^{\frac{\pi i \alpha}{2}\theta(z)\sigma_{3}}H_{-x_{m}}(z),
\end{equation}
where $G_{j}(z) = G_{j}(z;\vec{x},\vec{s})$ is analytic in a neighborhood of $(-x_{j+1},-x_{j-1})$, satisfies $\det G_{j} \equiv 1$, and $\theta(z)$, $V_{j}(z)$ are piecewise constant and defined by
\begin{align}
& \theta(z) = \left\{  \begin{array}{l l}
+1, & \mbox{Im } z > 0, \\
-1, & \mbox{Im } z < 0,
\end{array} \right. & & V_{j}(z) = \left\{  \begin{array}{l l}
I, & \mbox{Im } z > 0, \\
\begin{pmatrix}
1 & -s_{j} \\ 0 & 1
\end{pmatrix}, & \mbox{Im } z < 0.
\end{array} \right.
\end{align}

As $z$ tends to $0$, $\Phi$ takes the form
\begin{equation}\label{def of G_0}
\Phi(z) = G_{0}(z)z^{\frac{\alpha}{2}\sigma_{3}} \begin{pmatrix}
1 & s_{1}h(z) \\ 0 & 1
\end{pmatrix}, \qquad \alpha > -1,
\end{equation}
where $G_{0}(z)$ is analytic in a neighborhood of $(-x_{1},\infty)$, satisfies $\det G_{0} \equiv 1$ and 
\begin{equation}\label{def of h}
h(z) = \left\{ \begin{array}{l l}
\displaystyle \frac{1}{2 i \sin(\pi \alpha)}, & \alpha \notin \mathbb{N}_{\geq 0}, \\[0.35cm]
\displaystyle \frac{(-1)^{\alpha}}{2\pi i} \log z, & \alpha \in \mathbb{N}_{\geq 0}.
\end{array} \right.
\end{equation}
The quantities $\Phi_{1}$ and $G_{j}$, $j=0,\ldots,m$ also depend on $\alpha$, even though it is not indicated in the notation.
\end{itemize}

\subsection*{Differential identity}
It is known \cite[Theorem 2]{Soshnikov} that $F_{\alpha}(\vec{x},\vec{s})$ can also be expressed as a Fredholm determinant with $m$ discontinuities as follows
\begin{equation}\label{F Fredholm}
F_{\alpha}(\vec{x},\vec{s}) = \det (1-\mathcal{K}_{\vec{x},\vec{s}}), \qquad  \mathcal{K}_{\vec{x},\vec{s}}:=\hat\chi_{(0,x_{m})}\sum_{j=1}^{m}(1-s_{j})\mathcal{K}_{\alpha}^{\mathrm{Be}}\hat\chi_{(x_{j-1},x_{j})},
\end{equation}
where $\mathcal{K}_{\alpha}^{\mathrm{Be}}$ denotes the integral operator acting on $L^2(\R^+)$ whose kernel is $K_{\alpha}^{\mathrm{Be}}$ (given by \eqref{Bessel kernel}), and where $\hat\chi_A$ is the projection operator onto $L^2(A)$. 
We consider the function $K_{\vec{x},\vec{s}}: \mathbb{R}^{+} \times \mathbb{R}^{+} \to \mathbb{R}$ given by
\begin{equation}
K_{\vec{x},\vec{s}}(u,v) = \chi_{(0,x_{m})}(u) \sum_{j=1}^{m} (1-s_{j})K_{\alpha}^{\mathrm{Be}}(u,v) \chi_{(x_{j-1},x_{j})}(v), \qquad u,v>0,
\end{equation}
where for a given $A \subset \mathbb{R}$, $\chi_{A}$ denotes the characteristic function of $A$. For notational convenience, we omit the dependence of $K_{\vec{x},\vec{s}}$ and $\mathcal{K}_{\vec{x},\vec{s}}$ in $\alpha$. Since $s_{j}\geq 0$ for all $j=0,\ldots,m$, we deduce from \eqref{def of F} that $F_{\alpha}(\vec{x},\vec{s})=\det(1-\mathcal{K}_{\vec{x},\vec{s}})>0$. In particular, $1-\mathcal{K}_{\vec{x},\vec{s}}$ is invertible. Therefore, by standard properties of trace class operators, for any $k \in \{1,...,m\}$ we have
\begin{equation}\label{lol1}
\begin{array}{r c l}
\ds \partial_{s_{k}} \log \det (1-\mathcal{K}_{\vec{x},\vec{s}}) & = & \ds - \mbox{Tr}\Big( (1-\mathcal{K}_{\vec{x},\vec{s}})^{-1}\partial_{s_{k}} \mathcal{K}_{\vec{x},\vec{s}} \Big) \\
& = & \ds \frac{1}{1-s_{k}} \mbox{Tr}\Big( (1-\mathcal{K}_{\vec{x},\vec{s}})^{-1}\mathcal{K}_{\vec{x},\vec{s}}\chi_{(x_{k-1},x_{k})} \Big) \\
& = & \ds \frac{1}{1-s_{k}} \mbox{Tr} \Big( \mathcal{R}_{\vec{x},\vec{s}}\chi_{(x_{k-1},x_{k})} \Big) = \frac{1}{1-s_{k}}\int_{x_{k-1}}^{x_{k}}R_{\vec{x},\vec{s}}(u,u)du,
\end{array}
\end{equation}
where $\mathcal{R}_{\vec{x},\vec{s}}$ is the resolvent operator defined by
\begin{equation}
1+\mathcal{R}_{\vec{x},\vec{s}} = (1-\mathcal{K}_{\vec{x},\vec{s}})^{-1},
\end{equation}
and where $R_{\vec{x},\vec{s}}$ is the associated kernel. From \cite[equation (4.19)]{ChDoe}, for $u \in (x_{k-1},x_{k})$ we have
\begin{equation}
R_{\vec{x},\vec{s}}(u,u) = \frac{-e^{\pi i \alpha}}{2\pi i}(1-s_{k}) [\Phi_{-}^{-1}(-u;\vec{x},\vec{s})\partial_{u}(\Phi_{-}(-u;\vec{x},\vec{s}))]_{21}.
\end{equation}
Therefore, we obtain the following differential identity
\begin{equation}\label{diff identity integral form}
\partial_{s_{k}} \log \det (1-\mathcal{K}_{\vec{x},\vec{s}}) = \frac{e^{\pi i \alpha}}{2\pi i} \int_{-x_{k}}^{-x_{k-1}}[\Phi_{-}^{-1}(u;\vec{x},\vec{s})\partial_{u}\Phi_{-}(u;\vec{x},\vec{s})]_{21}du.
\end{equation}
Note that we implicitly assumed $s_{k} \neq 1$ in \eqref{lol1}, and thus \eqref{diff identity integral form} is a priori only true under this assumption. However, both sides of \eqref{diff identity integral form} are continuous as $s_{k} \to 1$, and therefore \eqref{diff identity integral form} also holds for $s_{k} = 1$ by continuity. (In fact both sides are analytic for $s_{k}$ in a small complex neighborhood of $1$. This follows from \cite[Theorem 2]{Soshnikov} and the fact that $\left. \det (1-\mathcal{K}_{\vec{x},\vec{s}}) \right|_{s_{k} = 1} > 0$ for the left-hand side, and from standard properties for RH problems for the right-hand side.)
\begin{remark}
It is quite remarkable that there are differential identities for $\log \det (1-\mathcal{K}_{\vec{x},\vec{s}})$ in terms of an RH problem. The reason behind this is that the kernel $\mathcal{K}_{\vec{x},\vec{s}}$ is so-called \textit{integrable} in the sense of Its, Izergin, Korepin and Slavnov \cite{IIKS}. This fact was also used extensively in $\cite{ChDoe}$ (even though the differential identities obtained in \cite{ChDoe} are different from \eqref{diff identity integral form}).
\end{remark}
In the rest of this section, we aim to simplify the integral on the right-hand side of \eqref{diff identity integral form}, following ideas presented in \cite[Section 3]{BothnerBuckingham} and using some results of \cite{ChDoe}. To prepare ourselves for that matter, we define for $r>0$ the following quantities
\begin{equation}\label{def of Phi tilde}
\widetilde \Phi(z;r) = \widetilde E(r) \Phi(rz;r\vec{x},\vec{s}), \qquad \widetilde E(r) = \begin{pmatrix}
1 & 0 \\ \frac{i}{\sqrt{r}}\Phi_{1,12}(r\vec{x},\vec{s}) & 1
\end{pmatrix}e^{\frac{\pi i}{4}\sigma_{3}}r^{\frac{\sigma_{3}}{4}},
\end{equation}
where we have omitted the dependence of $\widetilde \Phi$ and $\widetilde{E}$ in $\vec{x}$ and $\vec{s}$. It was shown in \cite[equation (3.15)]{ChDoe} that $\widetilde{\Phi}$ satisfies a Lax pair, and in particular
\begin{equation}
\partial_{z} \widetilde \Phi(z;r) = \widetilde{A}(z;r)\widetilde \Phi(z;r),
\end{equation}
where $\widetilde{A}$ is traceless, depends also on $\vec{x}$ and $\vec{s}$ and takes the form
\begin{equation}
\widetilde{A}(z;r) = \begin{pmatrix}
0 & 0 \\ \frac{\sqrt{r}}{2} & 0
\end{pmatrix} + \sum_{j=0}^{m} \frac{\widetilde{A}_{j}(r)}{z+x_{j}},
\end{equation}
for some traceless matrices $\widetilde{A}_{j}(r)=\widetilde{A}_{j}(r;\vec{x},\vec{s})$. Therefore, we have
\begin{multline}\label{def of A}
A(z;r) := \partial_{z} \Big( \Phi(rz;r\vec{x},\vec{s}) \Big) \Phi^{-1}(rz;r\vec{x},\vec{s}) = \widetilde{E}(r)^{-1} \widetilde{A}(z;r)\widetilde{E}(r) = \begin{pmatrix}
0 & 0 \\ \frac{ir}{2} & 0
\end{pmatrix} + \sum_{j=0}^{m} \frac{A_{j}(r)}{z+x_{j}},
\end{multline}
where $A_{j}(r) = \widetilde{E}(r)^{-1}\widetilde{A}_{j}(r)\widetilde{E}(r)$, $j = 0,1,...,m$. We will use later the following relations between the matrices $A_{j}$ and $G_{j}$. For $j = 1,...,m$, using \eqref{def of Gj} and $\det G_{j} \equiv 1$, we have
\begin{equation}\label{A_j}
\begin{array}{r c l}
\ds A_{j}(r) & = & \ds \frac{s_{j+1}-s_{j}}{2\pi i}(G_{j} \sigma_{+}G_{j}^{-1})(-rx_{j};r\vec{x},\vec{s}) \\[0.3cm]
& = & \ds \frac{s_{j+1}-s_{j}}{2\pi i} \begin{pmatrix}
-G_{j,11}G_{j,21} & G_{j,11}^{2} \\ -G_{j,21}^{2} & G_{j,11}G_{j,21}
\end{pmatrix},
\end{array}
\end{equation}
and, from \eqref{def of G_0} and $\det G_{0} \equiv 1$, we have
\begin{equation}\label{A_0}
A_{0}(r) = \left\{ \begin{array}{l l}
\ds \frac{s_{1}}{2\pi i}(G_{0}\sigma_{+}G_{0}^{-1})(0;r\vec{x},\vec{s}), & \mbox{if } \alpha = 0, \\[0.25cm]
\ds \frac{\alpha}{2}(G_{0}\sigma_{3}G_{0}^{-1})(0;r\vec{x},\vec{s}), & \mbox{if } \alpha \neq 0.
\end{array} \right.
\end{equation}
Now, we rewrite the integrand on the right-hand side of \eqref{diff identity integral form} (with $\vec{x} \mapsto r \vec{x}$) using \eqref{def of A}. Since $A$ is traceless and $\det \Phi \equiv 1$, we have
\begin{multline}
[\Phi^{-1}(rz;r\vec{x},\vec{s})\partial_{z} \big( \Phi(rz;r\vec{x},\vec{s}) \big)]_{21} = [\Phi^{-1}(rz;r\vec{x},\vec{s})A(z;r)\Phi(rz;r\vec{x},\vec{s})]_{21} \\ = \Phi_{11}^{2} A_{21} - \Phi_{21}^{2}A_{12}-2\Phi_{11}\Phi_{21}A_{11},
\end{multline}
which can be rewritten (again via \eqref{def of A}) as
\begin{multline}\label{explicit integrant}
[\Phi^{-1}(rz;r\vec{x},\vec{s})\partial_{z} \big( \Phi(rz;r\vec{x},\vec{s}) \big)]_{21} = (\Phi\sigma_{+}\Phi^{-1})_{12}(rz;r\vec{x},\vec{s}) \bigg[ \frac{ir}{2} + \sum_{j=0}^{m} \frac{A_{j,21}(r)}{z+x_{j}} \bigg] \\
+ (\Phi\sigma_{+}\Phi^{-1})_{21}(rz;r\vec{x},\vec{s})\sum_{j=0}^{m} \frac{A_{j,12}(r)}{z+x_{j}} + 2 (\Phi\sigma_{+}\Phi^{-1})_{11}(rz;r\vec{x},\vec{s})\sum_{j=0}^{m} \frac{A_{j,11}(r)}{z+x_{j}}.
\end{multline}
Let us define 
\begin{equation}
\widehat{F}(z) = \partial_{s_{k}}\Phi(rz;r\vec{x},\vec{s}) \Phi(rz;r\vec{x},\vec{s})^{-1},
\end{equation}
where $k \in \{1,\ldots,m\}$ is fixed and omitted in the notation. From the RH problem for $\Phi$, we deduce that $\widehat{F}$ satisfies the following RH problem.
\subsubsection*{RH problem for $\widehat{F}$} 
\begin{itemize}
\item[(a)] $\widehat{F}:\mathbb{C}\setminus [-x_{k},-x_{k-1}] \to \mathbb{C}^{2\times 2}$ is analytic.
\item[(b)] $\widehat{F}$ satisfies the jumps
\begin{equation}
\widehat{F}_{+}(z) = \widehat{F}_{-}(z) + e^{\pi i \alpha}(\Phi_{-}\sigma_{+}\Phi_{-}^{-1})(rz;r\vec{x},\vec{s}), \qquad z \in (-x_{k},-x_{k-1}).
\end{equation}
\item[(c)] $\widehat{F}$ satisfies the following asymptotic behaviors
\begin{align}
& \hspace{-1cm} \widehat{F}(z) = \frac{\partial_{s_{k}}\Phi_{1}(r\vec{x},\vec{s})}{rz} + \bigO(z^{-2}), & & \mbox{as } z \to \infty, \label{F at inf} \\
& \hspace{-1cm} \widehat{F}(z) = \frac{\partial_{s_{k}}(s_{j+1}-s_{j})}{s_{j+1}-s_{j}}A_{j}(r) \log(r(z+x_{j})) + \widehat{F}_{j}+o(1), & & \mbox{as } z \to -x_{j}, \, j = 1,...,m, \nonumber
\end{align}
where $\widehat{F}_{j} = (\partial_{s_{k}}G_{j}G_{j}^{-1})(-rx_{j};r\vec{x},\vec{s})$.
Furthermore, as $z \to 0$, we have
\begin{equation}\label{F at 0}
\widehat{F}(z) = \left\{ \begin{array}{l l}
\ds \widehat{F}_{0} + o(1), & \mbox{if } \alpha > 0, \\
\ds \frac{\partial_{s_{k}}s_{1}}{s_{1}}A_{0}(r) \log(rz)+\widehat{F}_{0}+o(1), & \mbox{if } \alpha = 0, \\
\ds \frac{\partial_{s_{k}}(s_{1})(rz)^{\alpha}}{2i\sin(\pi \alpha)}(G_{0}\sigma_{+}G_{0}^{-1})(0;r\vec{x},\vec{s})+\widehat{F}_{0}+o(1), & \mbox{if } \alpha < 0,
\end{array} \right.
\end{equation}
where $\widehat{F}_{0} = (\partial_{s_{k}}G_{0}G_{0}^{-1})(0;r\vec{x},\vec{s})$.
\end{itemize}
The RH problem for $\widehat{F}$ can be solved explicitly using Cauchy's formula, we have
\begin{equation}
\widehat{F}(z) = \frac{e^{\pi i \alpha}}{2\pi i} \int_{-x_{k}}^{-x_{k-1}} \frac{(\Phi_{-} \sigma_{+} \Phi_{-}^{-1})(ru;r\vec{x},\vec{s})}{u-z}du.
\end{equation}
Expanding the above expression as $z \to \infty$ and comparing with \eqref{F at inf}, we obtain
\begin{equation}
-\frac{e^{\pi i \alpha}}{2\pi i} \int_{-x_{k}}^{-x_{k-1}}(\Phi_{-}\sigma_{+}\Phi_{-}^{-1})(ru;r\vec{x},\vec{s})du = \frac{\partial_{s_{k}}\Phi_{1}(r\vec{x},\vec{s})}{r}.
\end{equation}
Substituting \eqref{explicit integrant} into \eqref{diff identity integral form} (with $\vec{x} \mapsto r\vec{x}$), we can simplify the integral using the expansions of $\widehat{F}$ at $\infty$ and at $-x_{j}$, $j=0,1,...,m$ (given by \eqref{F at inf}-\eqref{F at 0}). Note that $\det A_{j} \equiv 0$ for $j = 1,...,m$. Therefore, the logarithmic part in the expansions of $\widehat{F}(z)$ as $z \to -x_{j}$ for $j = 1,...,m$ does not contribute in \eqref{diff identity integral form}. One concludes the same for $j = 0$ if $\alpha = 0$. If $\alpha < 0$, the $\bigO(z^{\alpha})$ term in the $z \to 0$ expansion of $\widehat{F}$ also does not contribute in \eqref{diff identity integral form}, this follows from the relation
\begin{equation}
(G_{0}\sigma_{3}G_{0}^{-1})_{21}(G_{0}\sigma_{+}G_{0}^{-1})_{12}+(G_{0}\sigma_{3}G_{0}^{-1})_{12}(G_{0}\sigma_{+}G_{0}^{-1})_{21} + 2(G_{0}\sigma_{3}G_{0}^{-1})_{11}(G_{0}\sigma_{+}G_{0}^{-1})_{11} = 0,
\end{equation}
where we have used $\det G_{0} \equiv 1$. Therefore, for any $\alpha > -1$, we obtain
\begin{equation*}
\partial_{s_{k}} \log \det (1-\mathcal{K}_{r\vec{x},\vec{s}})= -\frac{i}{2}\partial_{s_{k}} \Phi_{1,12}(r\vec{x},\vec{s})+\sum_{j=0}^{m}[A_{j,21}(r)\widehat{F}_{j,12}+A_{j,12}(r)\widehat{F}_{j,21}+2A_{j,11}(r)\widehat{F}_{j,11}].
\end{equation*}
Finally, substituting in the above equality the explicit forms for the $A_{j}$'s and $\widehat{F}_{j}$'s given by \eqref{A_j}-\eqref{A_0} and below \eqref{F at inf}-\eqref{F at 0}, and simplifying the result with the identities $\det G_{j} \equiv 1$, we obtain
\begin{equation}\label{DIFF identity final form general case}
\partial_{s_{k}} \log F_{\alpha}(r\vec{x},\vec{s})= K_{\infty} + \sum_{j=1}^{m}K_{-x_{j}} + K_{0},
\end{equation}
where
\begin{align}
& K_{\infty} = -\frac{i}{2}\partial_{s_{k}} \Phi_{1,12}(r\vec{x},\vec{s}), \label{K inf} \\
& K_{-x_{j}} = \frac{s_{j+1}-s_{j}}{2\pi i} \Big( G_{j,11}\partial_{s_{k}} G_{j,21} - G_{j,21}\partial_{s_{k}} G_{j,11} \Big)(-rx_{j};r\vec{x},\vec{s}), \label{K -xj} \\
& K_{0} = \left\{ \begin{array}{l l}
\ds \frac{s_{1}}{2\pi i} \Big( G_{0,11}\partial_{s_{k}} G_{0,21} - G_{0,21}\partial_{s_{k}} G_{0,11} \Big)(0;r\vec{x},\vec{s}) & \mbox{if } \alpha = 0, \\[0.3cm]
\alpha \Big( G_{0,21} \partial_{s_{k}}G_{0,12} - G_{0,11} \partial_{s_{k}}G_{0,22} \Big)(0;r\vec{x},\vec{s}) & \mbox{if } \alpha \neq 0.
\end{array} \right. \label{K 0}
\end{align}

\section{Large $r$ asymptotics for $\Phi$ with $s_{1} \in (0,+\infty)$}\label{Section: Steepest descent with s1>0}
In this section, we perform a Deift/Zhou steepest descent analysis to obtain large $r$ asymptotics for $\Phi(rz;r\vec{x},\vec{s})$ in different regions of the complex $z$-plane. On the level of the parameters, we assume that $s_{1},...,s_{m}$ are in a compact subset of $(0,+\infty)$ and that $x_{1},...,x_{m}$ are in a compact subset of $(0,+\infty)$ in such a way that there exists $\delta > 0$ independent of $r$ such that
\begin{equation}
\min_{1 \leq j < k \leq m} x_{k}-x_{j} \geq \delta.
\end{equation}

\subsection{Normalization of the RH problem with a $g$-function}\label{subsection: g function s1 neq 0}
In the first transformation, we normalize the behavior at $\infty$ of the RH problem for $\Phi(rz;r\vec{x},\vec{s})$ by removing the term that grows exponentially with $z$. This transformation is standard in the literature (see e.g. \cite{Deift}) and uses a so-called $g$-function. In view of \eqref{Phi inf},  we define our $g$-function by
\begin{equation}\label{def of g s1 neq 0}
g(z)=\sqrt{z},
\end{equation}
where the principal branch is taken. Define 
\begin{equation}\label{def of T s1 neq 0}
T(z)= r^{\frac{\sigma_{3}}{4}}\Phi(rz;r\vec{x},\vec{s})e^{-\sqrt{r}g(z)\sigma_{3}}.
\end{equation}
The asymptotics \eqref{Phi inf} of $\Phi$ then lead after a straightforward calculation to 
\begin{equation}
\label{eq:Tasympinf s1 neq 0}
T(z) = \left( I + \frac{T_{1}}{z} + \bigO\left(z^{-2}\right) \right) z^{-\frac{\sigma_3}{4}} N, \qquad T_{1} = r^{\frac{\sigma_{3}}{4}}\frac{\Phi_{1}(r\vec{x},\vec{s})}{r}r^{-\frac{\sigma_{3}}{4}}
\end{equation}
as $z \to\infty$. In particular, $T_{1,12} = \frac{\Phi_{1,12}}{\sqrt{r}}$. The jumps for $T$ are obtained straightforwardly from those of $\Phi$ and the relation $g_{+}(z)+g_{-}(z) = 0$ for $z \in (-\infty,0)$. Since $s_{j} \neq 0$, the jump matrix for $T$ on $(-x_{j},-x_{j-1})$ can be factorized as follows
\begin{multline}\label{factorization of the jump}
\begin{pmatrix}
e^{\pi i \alpha}e^{-2\sqrt{r}g_{+}(z)} & s_{j} \\ 0 & e^{-\pi i \alpha}e^{-2\sqrt{r}g_{-}(z)}
\end{pmatrix} = \begin{pmatrix}
1 & 0 \\
s_{j}^{-1}e^{-\pi i \alpha}e^{-2\sqrt{r}g_{-}(z)} & 1
\end{pmatrix} \\ \times \begin{pmatrix}
0 & s_{j} \\ -s_{j}^{-1} & 0
\end{pmatrix} \begin{pmatrix}
1 & 0 \\ 
s_{j}^{-1}e^{\pi i \alpha}e^{-2\sqrt{r}g_{+}(z)} & 1
\end{pmatrix}.
\end{multline}
\subsection{Opening of the lenses}\label{subsection: S with s1 neq 0}
Around each interval $(-x_{j},-x_{j-1})$, $j = 1,...,m$, we open lenses $\gamma_{j,+}$ and $\gamma_{j,-}$, lying in the upper and lower half plane respectively, as shown in Figure \ref{fig:contour for S s1 neq 0}. Let us also denote $\Omega_{j,+}$ (resp. $\Omega_{j,-}$) for the region inside the lenses around $(-x_{j},-x_{j-1})$ in the upper half plane (resp. in the lower half plane). In view of \eqref{factorization of the jump}, we define the next transformation by
\begin{equation}\label{def of S s1 neq 0}
S(z) = T(z) \prod_{j=1}^{m} \left\{ \begin{array}{l l}
\begin{pmatrix}
1 & 0 \\
-s_{j}^{-1}e^{\pi i \alpha}e^{-2\sqrt{r}g(z)} & 1
\end{pmatrix}, & \mbox{if } z \in \Omega_{j,+}, \\
\begin{pmatrix}
1 & 0 \\
s_{j}^{-1}e^{-\pi i \alpha}e^{-2\sqrt{r}g(z)} & 1
\end{pmatrix}, & \mbox{if } z \in \Omega_{j,-}, \\
I, & \mbox{if } z \in \mathbb{C}\setminus(\Omega_{j,+}\cup \Omega_{j,-}).
\end{array} \right.
\end{equation}
It is straightforward to verify from the RH problem for $\Phi$ and from Section \ref{subsection: g function s1 neq 0} that $S$ satisfies the following RH problem.
\subsubsection*{RH problem for $S$}
\begin{enumerate}[label={(\alph*)}]
\item[(a)] $S : \C \backslash \Gamma_{S} \rightarrow \C^{2\times 2}$ is analytic, with
\begin{equation}\label{eq:defGamma s1 neq 0}
\Gamma_{S}=(-\infty,0)\cup \gamma_{+}\cup \gamma_{-}, \qquad \gamma_{\pm} = \bigcup_{j=1}^{m+1} \gamma_{j,\pm},
\end{equation}
where $\gamma_{m+1,\pm} := -x_{m} + e^{\pm \frac{2\pi i}{3}}(0,+\infty)$, and $\Gamma_{S}$ is oriented as shown in Figure \ref{fig:contour for S s1 neq 0}.
\item[(b)] The jumps for $S$ are given by
\begin{align*}
& S_{+}(z) = S_{-}(z)\begin{pmatrix}
0 & s_{j} \\ -s_{j}^{-1} & 0
\end{pmatrix}, & & z \in (-x_{j},-x_{j-1}), \, j = 1,...,m+1, \\
& S_{+}(z) = S_{-}(z)\begin{pmatrix}
1 & 0 \\
s_{j}^{-1}e^{\pm \pi i \alpha}e^{-2\sqrt{r}g(z)} & 1
\end{pmatrix}, & & z \in \gamma_{j,\pm}, \, j = 1,...,m+1,
\end{align*}
where $x_{m+1} := +\infty$ (we recall that $x_{0} = 0$ and $s_{m+1} = 1$).
\item[(c)] As $z \rightarrow \infty$, we have
\begin{equation}
\label{eq:Sasympinf s1 neq 0}
S(z) = \left( I + \frac{T_1}{z} + \bigO\left( z^{-2} \right) \right) z^{-\frac{\sigma_3}{4} } N.
\end{equation}
As $z \to -x_j$ from outside the lenses, $j = 1, ..., m$, we have
\begin{equation}\label{local behaviour near -xj of S s1 neq 0}
S(z) = \begin{pmatrix}
\bigO(1) & \bigO(\log(z+x_{j})) \\
\bigO(1) & \bigO(\log(z+x_{j}))
\end{pmatrix}.
\end{equation}
As $z \to 0$ from outside the lenses, we have
\begin{equation}\label{local behaviour near 0 of S s1 neq 0}
\displaystyle S(z) = \left\{ \begin{array}{l l}
\begin{pmatrix}
\bigO(1) & \bigO(\log z) \\
\bigO(1) & \bigO(\log z) 
\end{pmatrix}, & \mbox{ if } \alpha = 0, \\[0.3cm]
\begin{pmatrix}
\bigO(1) & \bigO(1) \\
\bigO(1) & \bigO(1) 
\end{pmatrix}z^{\frac{\alpha}{2}\sigma_{3}}, & \mbox{ if } \alpha > 0, \\[0.3cm]
\begin{pmatrix}
\bigO(z^{\frac{\alpha}{2}}) & \bigO(z^{\frac{\alpha}{2}}) \\
\bigO(z^{\frac{\alpha}{2}}) & \bigO(z^{\frac{\alpha}{2}}) 
\end{pmatrix}, & \mbox{ if } \alpha < 0.
\end{array}  \right.
\end{equation}
\end{enumerate}
Since $\Re g(z) > 0$ for all $z \in \mathbb{C}\setminus (-\infty,0]$ and $\Re g_{\pm}(z) = 0$ for $z \in (-\infty,0)$, the jump matrices for $S$ tend to the identity matrix exponentially fast as $r \to + \infty$ on the lenses. This convergence is uniform for $z$ outside of fixed neighborhoods of $-x_{j}$, $j \in \{0,1,...,m\}$, but is not uniform as $r \to + \infty$ and simultaneously $z \to -x_{j}$, $j \in \{0,1,...,m\}$.
\begin{figure}
\centering
\begin{tikzpicture}
\draw[fill] (0,0) circle (0.05);
\draw (0,0) -- (8,0);
\draw (0,0) -- (120:3);
\draw (0,0) -- (-120:3);
\draw (0,0) -- (-3,0);

\draw (0,0) .. controls (1,1.3) and (2,1.3) .. (3,0);
\draw (0,0) .. controls (1,-1.3) and (2,-1.3) .. (3,0);
\draw (3,0) .. controls (3.5,1) and (4.5,1) .. (5,0);
\draw (3,0) .. controls (3.5,-1) and (4.5,-1) .. (5,0);
\draw (5,0) .. controls (6,1.3) and (7,1.3) .. (8,0);
\draw (5,0) .. controls (6,-1.3) and (7,-1.3) .. (8,0);

\draw[fill] (3,0) circle (0.05);
\draw[fill] (5,0) circle (0.05);
\draw[fill] (8,0) circle (0.05);

\node at (0.15,-0.3) {$-x_{m}$};
\node at (3,-0.3) {$-x_{2}$};
\node at (5.1,-0.3) {$-x_{1}$};
\node at (8.4,-0.3) {$0=x_{0}$};
\node at (-3,-0.3) {$-\infty=-x_{m+1}$};

\draw[black,arrows={-Triangle[length=0.18cm,width=0.12cm]}]
(-120:1.5) --  ++(60:0.001);
\draw[black,arrows={-Triangle[length=0.18cm,width=0.12cm]}]
(120:1.3) --  ++(-60:0.001);
\draw[black,arrows={-Triangle[length=0.18cm,width=0.12cm]}]
(180:1.5) --  ++(0:0.001);

\draw[black,arrows={-Triangle[length=0.18cm,width=0.12cm]}]
(0:1.5) --  ++(0:0.001);
\draw[black,arrows={-Triangle[length=0.18cm,width=0.12cm]}]
(0:4) --  ++(0:0.001);
\draw[black,arrows={-Triangle[length=0.18cm,width=0.12cm]}]
(0:6.5) --  ++(0:0.001);

\draw[black,arrows={-Triangle[length=0.18cm,width=0.12cm]}]
(1.55,0.97) --  ++(0:0.001);
\draw[black,arrows={-Triangle[length=0.18cm,width=0.12cm]}]
(1.55,-0.97) --  ++(0:0.001);

\draw[black,arrows={-Triangle[length=0.18cm,width=0.12cm]}]
(4.05,0.76) --  ++(0:0.001);
\draw[black,arrows={-Triangle[length=0.18cm,width=0.12cm]}]
(4.05,-0.76) --  ++(0:0.001);

\draw[black,arrows={-Triangle[length=0.18cm,width=0.12cm]}]
(6.55,0.97) --  ++(0:0.001);
\draw[black,arrows={-Triangle[length=0.18cm,width=0.12cm]}]
(6.55,-0.97) --  ++(0:0.001);

\end{tikzpicture}
\caption{Jump contours $\Gamma_{S}$ for the RH problem for $S$ with $m=3$ and $s_{1} \neq 0$.}
\label{fig:contour for S s1 neq 0}
\end{figure}
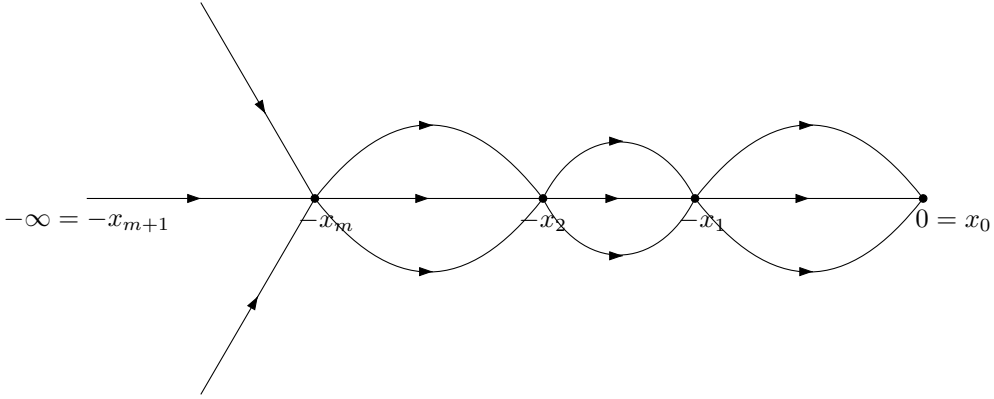
\subsection{Global parametrix}\label{subsection: Global param s1 neq 0}
By ignoring the jumps for $S$ that are pointwise exponentially close to the identity matrix as $r \to + \infty$, we are left with an RH problem which is independent of $r$, and whose solution is called the global parametrix and denoted $P^{(\infty)}$. It will appear later in Section \ref{subsection Small norm s1 neq 0} that $P^{(\infty)}$ is a good approximation for $S$ away from neighborhoods of $-x_{j}$, $j = 0,1,...,m$.
\subsubsection*{RH problem for $P^{(\infty)}$}
\begin{enumerate}[label={(\alph*)}]
\item[(a)] $P^{(\infty)} : \C \backslash (-\infty,0] \rightarrow \C^{2\times 2}$ is analytic.
\item[(b)] The jumps for $P^{(\infty)}$ are given by
\begin{align*}
& P^{(\infty)}_{+}(z) = P^{(\infty)}_{-}(z)\begin{pmatrix}
0 & s_{j} \\ -s_{j}^{-1} & 0
\end{pmatrix}, & & z \in (-x_{j},-x_{j-1}), \, j = 1,...,m+1.
\end{align*}
\item[(c)] As $z \rightarrow \infty$, we have
\begin{equation}
\label{eq:Pinf asympinf s1 neq 0}
P^{(\infty)}(z) = \left( I + \frac{P^{(\infty)}_{1}}{z} + \bigO\left( z^{-2} \right) \right) z^{-\frac{\sigma_3}{4}} N,
\end{equation}
for a certain matrix $P_{1}^{(\infty)}$ independent of $z$.
\item[(d)] As $z \to -x_{j}$, $j \in \{1,...,m\}$, we have $P^{(\infty)}(z) = \begin{pmatrix}
\bigO(1) & \bigO(1) \\
\bigO(1) & \bigO(1)
\end{pmatrix}$.

As $z \to 0$, we have $P^{(\infty)}(z) = \begin{pmatrix}
\bigO(z^{-1/4}) & \bigO(z^{-1/4}) \\
\bigO(z^{-1/4}) & \bigO(z^{-1/4})
\end{pmatrix}$.
\end{enumerate}
Note that condition (d) for the RH problem for $P^{(\infty)}$ does not come from the RH problem for $S$. It is added to ensure uniqueness of the solution. The construction of $P^{(\infty)}$ relies on a so-called Szeg\"{o} function $D$ (see \cite{Kuijlaars2}). In our case, we need to define $D$ as follows
\begin{align*}
& D(z) = \exp \left( \frac{\sqrt{z}}{2\pi} \sum_{j=1}^{m} \log s_{j} \int_{x_{j-1}}^{x_{j}}\frac{du}{\sqrt{u}(z+u)} \right).
\end{align*}
It satisfies the following jumps
\begin{align*}
& D_{+}(z)D_{-}(z) = s_{j}, & &  \mbox{for } z \in (-x_{j},-x_{j-1}), \, j = 1,...,m+1.
\end{align*}
Furthermore, as $z \to \infty$, we have
\begin{equation}
 D(z) =  \exp \left( \sum_{\ell = 1}^{k} \frac{d_{\ell}}{z^{\ell-\frac{1}{2}}} + \bigO(z^{-k-\frac{1}{2}}) \right),
\end{equation}
where $k \in \mathbb{N}_{>0}$ is arbitrary and
\begin{align}
& d_{\ell} = \frac{(-1)^{\ell - 1}}{2\pi} \sum_{j=1}^{m}  \log s_{j} \int_{x_{j-1}}^{x_{j}} u^{\ell - \frac{3}{2}}du = \frac{(-1)^{\ell -1}}{\pi (2\ell -1)}\sum_{j=1}^{m} \log s_{j} \Big( x_{j}^{\ell -\frac{1}{2}}-x_{j-1}^{\ell-\frac{1}{2}} \Big) \nonumber.
\end{align}
Let us finally define
\begin{equation}\label{def of Pinf s1 neq 0}
P^{(\infty)}(z) = \begin{pmatrix}
1 & 0 \\ id_{1} & 1
\end{pmatrix}z^{-\frac{\sigma_{3}}{4}}ND(z)^{-\sigma_{3}},
\end{equation}
where the principal branch is taken for the root. From the above properties of $D$, one can check that $P^{(\infty)}$ satisfies criteria (a), (b) and (c) of the RH problem for $P^{(\infty)}$, with
\begin{equation}\label{Pinf 1 12 s1 neq 0}
P_{1,12}^{(\infty)} = i d_{1}.
\end{equation}
The rest of the current section is devoted to the computations of the first terms in the asymptotics of $D(z)$ as $z \to -x_{j}$, $j = 0,1,...,m$. It will in particular prove that $P^{(\infty)}$ defined in \eqref{def of Pinf s1 neq 0} satisfies condition (d) of the RH problem for $P^{(\infty)}$. After integrations, we can rewrite $D$ as follows
\begin{equation}\label{lol9}
D(z) = \prod_{j=1}^{m} D_{s_{j}}(z),
\end{equation}
where
\begin{equation}\label{lol8}
D_{s_{j}}(z) = \left( \frac{(\sqrt{z}-i\sqrt{x_{j-1}})(\sqrt{z}+i\sqrt{x_{j}})}{(\sqrt{z}-i\sqrt{x_{j}})(\sqrt{z}+i\sqrt{x_{j-1}})} \right)^{\frac{\log s_{j}}{2\pi i}}.
\end{equation}
As $z \to -x_{j}$, $j \in \{1,...,m\}$, $\Im z > 0$, we have
\begin{equation}
D_{s_{j}}(z) = \sqrt{s_{j}}T_{j,j}^{\frac{\log s_{j}}{2\pi i}}(z+x_{j})^{-\frac{\log s_{j}}{2\pi i}}(1+\bigO(z+x_{j})), \quad T_{j,j} = 4x_{j}\frac{\sqrt{x_{j}}-\sqrt{x_{j-1}}}{\sqrt{x_{j}}+\sqrt{x_{j-1}}}.
\end{equation}
As $z \to -x_{j-1}$, $j \in \{2,...,m\}$, $\Im z > 0$, we have
\begin{equation}
D_{s_{j}}(z) = T_{j,j-1}^{\frac{\log s_{j}}{2\pi i}}(z+x_{j-1})^{\frac{\log s_{j}}{2\pi i}}(1+\bigO(z+x_{j-1})), \quad T_{j,j-1} = \frac{1}{4x_{j-1}}\frac{\sqrt{x_{j}}+\sqrt{x_{j-1}}}{\sqrt{x_{j}}-\sqrt{x_{j-1}}}.
\end{equation}
For $j \in \{1,...,m\}$, as $z \to -x_{k}$, $k \in \{1,...,m\}$, $k \neq j,j-1$, $\Im z > 0$, we have
\begin{equation}\label{lol7}
D_{s_{j}}(z) = T_{j,k}^{\frac{\log s_{j}}{2\pi i}}(1+\bigO(z+x_{k})), \quad T_{j,k} = \frac{(\sqrt{x_{k}}-\sqrt{x_{j-1}})(\sqrt{x_{k}}+\sqrt{x_{j}})}{(\sqrt{x_{k}}-\sqrt{x_{j}})(\sqrt{x_{k}}+\sqrt{x_{j-1}})}.
\end{equation}
From the above expansions, we obtain, as $z \to -x_{j}$, $j \in \{1,...,m\}$, $\Im z > 0$ that
\begin{equation}\label{lol6}
D(z) = \sqrt{s_{j}}\Big( \prod_{k=1}^{m} T_{k,j}^{\frac{\log s_{k}}{2\pi i}} \Big) (z+x_{j})^{\beta_{j}}(1+\bigO(z+x_{j})),
\end{equation}
where $\beta_{j}$ is given by that
\begin{equation}\label{def of beta_j s1 neq 0}
\beta_{j} = \frac{1}{2\pi i}\log \frac{s_{j+1}}{s_{j}}, \quad \mbox{ or equivalently } \quad e^{-2i\pi \beta_{j}} = \frac{s_{j}}{s_{j+1}}, \qquad j = 1,...,m.
\end{equation}
with $s_{m+1} := 1$.
It will be more convenient to rewrite the product in \eqref{lol6} in terms of the $\beta_{k}$'s as follows
\begin{equation}\label{T and T tilde s1 neq 0}
\prod_{k=1}^{m} T_{k,j}^{\frac{\log s_{k}}{2\pi i}} = (4x_{j})^{-\beta_{j}}\prod_{\substack{k=1 \\ k \neq j}}^{m} \widetilde{T}_{k,j}^{-\beta_{k}}, \quad \mbox{ where } \quad \widetilde{T}_{k,j} = \frac{\sqrt{x_{j}}+\sqrt{x_{k}}}{|\sqrt{x_{j}}-\sqrt{x_{k}}|}.
\end{equation}
We will also need the first two terms of the asymptotics of $D$ at the origin. From \eqref{lol9}-\eqref{lol8}, we obtain
\begin{equation}\label{asymp D at 0 s1 neq 0}
D(z) = \sqrt{s_{1}}\Big(1-d_{0}\sqrt{z}+\bigO(z)\Big), \qquad \mbox{ as } z \to 0,
\end{equation}
where
\begin{equation}
d_{0} = \frac{\log s_{1}}{\pi \sqrt{x_{1}}} - \sum_{j=2}^{m} \frac{\log s_{j}}{\pi}\Big( \frac{1}{\sqrt{x_{j-1}}}-\frac{1}{\sqrt{x_{j}}} \Big).
\end{equation}
Note that for all $\ell \in \{0,1,2,...\}$, we can rewrite $d_{\ell}$ in terms of the $\beta_{j}$'s as follows
\begin{equation}\label{d_ell in terms of beta_j s1 neq 0}
d_{\ell} = \frac{2i(-1)^{\ell}}{2\ell-1}\sum_{j=1}^{m} \beta_{j} x_{j}^{\ell-\frac{1}{2}}.
\end{equation}

\subsection{Local parametrices}
In this section, we aim to find approximations for $S$ in small neighborhoods of $0$,$-x_{1}$,...,$-x_{m}$. This is the part of the RH analysis where we use the assumption that there exists $\delta > 0$ such that \eqref{condition on xj in terms of delta} holds. By \eqref{condition on xj in terms of delta}, there exist small disks $\mathcal{D}_{-x_{j}}$ centred at  $-x_{j}$, $j=0,1,...,m$, whose radii are fixed (independent of $r$), but sufficiently small such that they do not intersect. The local parametrix around $-x_{j}$, $j \in \{0,1,...,m\}$, is defined in $\mathcal{D}_{-x_{j}}$ and is denoted by $P^{(-x_{j})}$. It satisfies an RH problem with the same jumps as $S$ (inside $\mathcal{D}_{-x_{j}}$) and a behavior near $-x_{j}$ ``close" to $S$. Furthermore, on the boundary of the disk, $P^{(-x_{j})}$ needs to ``match" with $P^{(\infty)}$ (called the matching condition). More precisely, we require
\begin{equation}\label{match at the center s1 neq 0}
S(z) P^{(-x_{j})}(z)^{-1} = \bigO(1), \qquad \mbox{ as } z \to -x_{j},
\end{equation}
and
\begin{equation}\label{matching weak s1 neq 0}
P^{(-x_{j})}(z) = (I+o(1))P^{(\infty)}(z), \qquad \mbox{ as } r \to +\infty,
\end{equation}
uniformly for $z \in \partial \mathcal{D}_{-x_{j}}$.
\subsubsection{Local parametrices around $-x_{j}$, $j = 1,...,m$}\label{subsection: local param HG s1 neq 0}
For $j \in \{1,...,m\}$, $P^{(-x_{j})}$ can be explicitly expressed in terms of confluent hypergeometric functions. This construction is standard (see e.g. \cite{ItsKrasovsky,FoulquieMartinezSousa}) and involves a model RH problem whose solution is denoted $\Phi_{\mathrm{HG}}$ (the details are presented in the appendix, Section \ref{subsection: model RHP with HG functions}). Let us first consider the function
\begin{equation}
f_{-x_{j}}(z) = -2 \left\{ \begin{array}{l l}
g(z)-g_{+}(-x_{j}), & \mbox{if } \Im z > 0 \\
-(g(z)-g_{-}(-x_{j})), & \mbox{if } \Im z < 0
\end{array} \right. = -2i\big(\sqrt{-z}-\sqrt{x_{j}}\big).
\end{equation}
This is a conformal map from $\mathcal{D}_{-x_{j}}$ to a neighborhood of $0$ and its expansion as $z \to -x_{j}$ is given by
\begin{equation}\label{expansion conformal map s1 neq 0}
f_{-x_{j}}(z) = i c_{-x_{j}} (z+x_{j})(1+\bigO(z+x_{j})), \quad \mbox{ with } \quad c_{-x_{j}} = \frac{1}{\sqrt{x_{j}}} > 0.
\end{equation}
Note also that $f_{-x_{j}}(\mathbb{R}\cap \mathcal{D}_{-x_{j}})\subset i \mathbb{R}$. Now, we use the freedom we had in the choice of the lenses by requiring that $f_{-x_{j}}$ maps the jump contour for $P^{(-x_{j})}$ onto a subset of $\Sigma_{\mathrm{HG}}$ (see Figure \ref{Fig:HG}):
\begin{equation}\label{deformation of the lenses local param xj s1 neq 0}
f_{-x_{j}}((\gamma_{j+1,+}\cup \gamma_{j,+})\cap \mathcal{D}_{-x_{j}}) \subset \Gamma_{3} \cup \Gamma_{2}, \qquad f_{-x_{j}}((\gamma_{j+1,-}\cup \gamma_{j,-})\cap \mathcal{D}_{-x_{j}}) \subset \Gamma_{5} \cup \Gamma_{6},
\end{equation}
where $\Gamma_{3}$, $\Gamma_{2}$, $\Gamma_{5}$ and $\Gamma_{6}$ are as shown in Figure \ref{Fig:HG}.
Let us define $P^{(-x_{j})}$ by
\begin{equation}\label{lol10}
P^{(-x_{j})}(z) = E_{-x_{j}}(z) \Phi_{\mathrm{HG}}(\sqrt{r}f_{-x_{j}}(z);\beta_{j})(s_{j}s_{j+1})^{-\frac{\sigma_{3}}{4}}e^{-\sqrt{r}g(z)\sigma_{3}}e^{\frac{\pi i \alpha}{2}\theta(z)\sigma_{3}},
\end{equation}
where $E_{-x_{j}}$ is analytic inside $\mathcal{D}_{-x_{j}}$ (and will be determined explicitly below) and where the parameter $\beta_{j}$ for $\Phi_{\mathrm{HG}}$ is given by \eqref{def of beta_j s1 neq 0}. Since $E_{-x_{j}}$ is analytic, it is straightforward from the jumps for $\Phi_{\mathrm{HG}}$ (given by \eqref{jumps PHG3}) to verify that $P^{(-x_{j})}$ given by \eqref{lol10} satisfies the same jumps as $S$ inside $\mathcal{D}_{-x_{j}}$. In order to fulfil the matching condition \eqref{matching weak s1 neq 0}, using \eqref{Asymptotics HG}, we need to choose
\begin{multline}\label{def of Ej s1 neq 0}
E_{-x_{j}}(z) = P^{(\infty)}(z) e^{-\frac{\pi i \alpha}{2}\theta(z)\sigma_{3}} (s_{j} s_{j+1})^{\frac{\sigma_{3}}{4}} \left\{ \begin{array}{l l}
\ds \sqrt{\frac{s_{j}}{s_{j+1}}}^{\sigma_{3}}, & \Im z > 0 \\
\begin{pmatrix}
0 & 1 \\ -1 & 0
\end{pmatrix}, & \Im z < 0
\end{array} \right\} \times \\ e^{\sqrt{r}g_{+}(-x_{j})\sigma_{3}}(\sqrt{r}f_{-x_{j}}(z))^{\beta_{j}\sigma_{3}}.
\end{multline}
It can be verified from the jumps for $P^{(\infty)}$ that $E_{-x_{j}}$ defined by \eqref{def of Ej s1 neq 0} has no jump at all inside $\mathcal{D}_{-x_{j}}$. Furthermore, $E_{-x_{j}}(z)$ is bounded as $z \to -x_{j}$ and $E_{-x_{j}}$ is then analytic in the full disk $\mathcal{D}_{-x_{j}}$, as desired. Since $P^{(-x_{j})}$ and $S$ have exactly the same jumps on $(\mathbb{R}\cup \gamma_{+}\cup \gamma_{-})\cap \mathcal{D}_{-x_{j}}$, $S(z)P^{(-x_{j})}(z)^{-1}$ is analytic in $\mathcal{D}_{-x_{j}} \setminus \{-x_{j}\}$. As $z \to -x_{j}$ from outside the lenses, by condition (d) in the RH problem for $S$ and by \eqref{lol 35}, $S(z)P^{(-x_{j})}(z)^{-1}$ behaves as $\bigO(\log(z+x_{j}))$. This means that the singularity is removable and \eqref{match at the center s1 neq 0} holds. We will need later a more detailed knowledge than \eqref{matching weak s1 neq 0}. Using \eqref{Asymptotics HG}, one shows that
\begin{equation}\label{matching strong -x_j s1 neq 0}
P^{(-x_{j})}(z)P^{(\infty)}(z)^{-1} = I + \frac{1}{\sqrt{r}f_{-x_{j}}(z)}E_{-x_{j}}(z) \Phi_{\mathrm{HG},1}(\beta_{j})E_{-x_{j}}(z)^{-1} + \bigO(r^{-1}),
\end{equation}
as $r \to + \infty$, uniformly for $z \in \partial \mathcal{D}_{-x_{j}}$, where $\Phi_{\mathrm{HG},1}(\beta_{j})$ is given by \eqref{def of tau} (with $\beta_{j}$ given by \eqref{def of beta_j s1 neq 0}). Also, a direct computation using \eqref{def of Pinf s1 neq 0}, \eqref{lol6}-\eqref{T and T tilde s1 neq 0} and \eqref{expansion conformal map s1 neq 0} shows that
\begin{equation}\label{E_j at -x_j s1 neq 0}
E_{-x_{j}}(-x_{j}) = \begin{pmatrix}
1 & 0 \\ id_{1} & 1
\end{pmatrix} e^{-\frac{\pi i}{4}\sigma_{3}} x_{j}^{-\frac{\sigma_{3}}{4}}N\Lambda_{j}^{\sigma_{3}},
\end{equation}
where
\begin{equation}\label{def Lambda_j s1 neq 0}
\Lambda_{j} = e^{-\frac{\pi i \alpha}{2}} (4x_{j})^{\beta_{j}} \bigg( \prod_{\substack{k=1 \\ k \neq j}}^{m} \widetilde{T}_{k,j}^{\beta_{k}} \bigg)e^{\sqrt{r}g_{+}(-x_{j})}r^{\frac{\beta_{j}}{2}}c_{-x_{j}}^{\beta_{j}}.
\end{equation}
\subsubsection{Local parametrices around $0$}\label{subsection: local param bessel s1 neq 0}
The local parametrix $P^{(0)}$ can be constructed in terms of Bessel functions, and relies on the model RH problem for $\Phi_{\mathrm{Be}}$ (this model RH problem is well-known, see e.g. \cite{KMcLVAV}, and is presented in the appendix, Section \ref{subsection:Model Bessel}). Let us first consider the function
\begin{equation}\label{conformal map near 0 s1 neq 0}
f_{0}(z) = \frac{g(z)^{2}}{4} = \frac{z}{4}.
\end{equation}
This is a conformal map from $\mathcal{D}_{0}$ to a neighborhood of $0$. Similarly to the previous local parametrices, we use the freedom in the choice of the lenses by requiring that
\begin{equation}\label{deformation of the lenses local param 0 s1 neq 0}
f_{0}(\gamma_{1,+}) \subset e^{\frac{2\pi i}{3}}\mathbb{R}^{+}, \qquad f_{0}(\gamma_{1,-}) \subset e^{-\frac{2\pi i}{3}}\mathbb{R}^{+}.
\end{equation} 
Thus the jump contour for $P^{(0)}$ is mapped by $f_{0}$ onto a subset of $\Sigma_{\mathrm{Be}}$ ($\Sigma_{\mathrm{Be}}$ is the jump contour for $\Phi_{\mathrm{Be}}$, see Figure \ref{figBessel}). We take $P^{(0)}$ in the form
\begin{equation}\label{def of P^-x1 s1 neq 0}
P^{(0)}(z) = E_{0}(z)\Phi_{\mathrm{Be}}(rf_{0}(z);\alpha)s_{1}^{-\frac{\sigma_{3}}{2}}e^{-\sqrt{r}g(z)\sigma_{3}},
\end{equation}
where $E_{0}$ is analytic inside $\mathcal{D}_{0}$ (and will be determined below). From \eqref{Jump for P_Be}, it is straightforward to verify that $P^{(0)}$ given by \eqref{def of P^-x1 s1 neq 0} has the same jumps as $S$ inside $\mathcal{D}_{0}$. In order to satisfy the matching condition \eqref{matching weak s1 neq 0}, by \eqref{large z asymptotics Bessel}, we defined $E_{0}$ by
\begin{equation}
E_{0}(z) = P^{(\infty)}(z)s_{1}^{\frac{\sigma_{3}}{2}}N^{-1}\left( 2\pi \sqrt{r}f_{0}(z)^{1/2} \right)^{\frac{\sigma_{3}}{2}}.
\end{equation}
It can be verified from the jumps for $P^{(\infty)}$ that $E_{0}$ has no jumps in $\mathcal{D}_{0}$, and has a removable singularity at $0$. Therefore, $E_{0}$ is analytic in $\mathcal{D}_{0}$. We will need later a more detailed knowledge of \eqref{matching weak s1 neq 0}. Using \eqref{large z asymptotics Bessel}, one shows that
\begin{equation}\label{matching strong 0 s1 neq 0}
P^{(0)}(z)P^{(\infty)}(z)^{-1} = I + \frac{1}{\sqrt{r}f_{0}(z)^{1/2}}P^{(\infty)}(z)s_{1}^{\frac{\sigma_{3}}{2}}\Phi_{\mathrm{Be},1}(\alpha)s_{1}^{-\frac{\sigma_{3}}{2}}P^{(\infty)}(z)^{-1} + \bigO(r^{-1}),
\end{equation}
as $r \to +\infty$ uniformly for $z \in \partial \mathcal{D}_{0}$, where $\Phi_{\mathrm{Be},1}(\alpha)$ is given below \eqref{large z asymptotics Bessel}. Furthermore, using \eqref{def of Pinf s1 neq 0}, \eqref{asymp D at 0 s1 neq 0} and \eqref{conformal map near 0 s1 neq 0}, we obtain
\begin{equation}\label{E0 at 0 s1 neq 0}
E_{0}(0) = \begin{pmatrix}
1 & 0 \\ id_{1} & 1
\end{pmatrix} \begin{pmatrix}
1 & -id_{0} \\
0 & 1
\end{pmatrix} (\pi \sqrt{r})^{\frac{\sigma_{3}}{2}}.
\end{equation}

\subsection{Small norm problem}\label{subsection Small norm s1 neq 0}

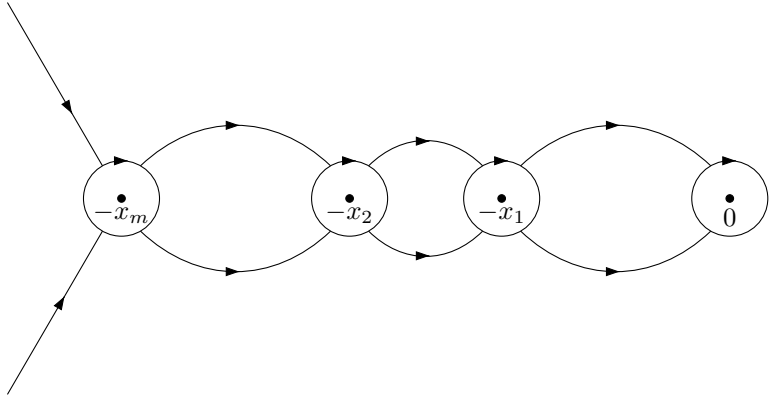
\begin{figure}
\centering
\begin{tikzpicture}
\draw[fill] (0,0) circle (0.05);
\draw (0,0) circle (0.5);

\draw (120:0.5) -- (120:3);
\draw (-120:0.5) -- (-120:3);

\draw ($(0,0)+(60:0.5)$) .. controls (1,1.15) and (2,1.15) .. ($(3,0)+(120:0.5)$);
\draw ($(0,0)+(-60:0.5)$) .. controls (1,-1.15) and (2,-1.15) .. ($(3,0)+(-120:0.5)$);
\draw ($(3,0)+(60:0.5)$) .. controls (3.7,0.88) and (4.3,0.88) .. ($(5,0)+(120:0.5)$);
\draw ($(3,0)+(-60:0.5)$) .. controls (3.7,-0.88) and (4.3,-0.88) .. ($(5,0)+(-120:0.5)$);
\draw ($(5,0)+(60:0.5)$) .. controls (6,1.15) and (7,1.15) .. ($(8,0)+(120:0.5)$);
\draw ($(5,0)+(-60:0.5)$) .. controls (6,-1.15) and (7,-1.15) .. ($(8,0)+(-120:0.5)$);

\draw[fill] (3,0) circle (0.05);
\draw (3,0) circle (0.5);
\draw[fill] (5,0) circle (0.05);
\draw (5,0) circle (0.5);
\draw[fill] (8,0) circle (0.05);
\draw (8,0) circle (0.5);

\node at (0.,-0.2) {$-x_{m}$};
\node at (3,-0.2) {$-x_{2}$};
\node at (5,-0.2) {$-x_{1}$};
\node at (8,-0.25) {$0$};

\draw[black,arrows={-Triangle[length=0.18cm,width=0.12cm]}]
(-120:1.5) --  ++(60:0.001);
\draw[black,arrows={-Triangle[length=0.18cm,width=0.12cm]}]
(120:1.3) --  ++(-60:0.001);
\draw[black,arrows={-Triangle[length=0.18cm,width=0.12cm]}]
($(0.08,0)+(90:0.5)$) --  ++(0:0.001);

\draw[black,arrows={-Triangle[length=0.18cm,width=0.12cm]}]
($(3.08,0)+(90:0.5)$) --  ++(0:0.001);
\draw[black,arrows={-Triangle[length=0.18cm,width=0.12cm]}]
($(5.08,0)+(90:0.5)$) --  ++(0:0.001);
\draw[black,arrows={-Triangle[length=0.18cm,width=0.12cm]}]
($(8.08,0)+(90:0.5)$) --  ++(0:0.001);

\draw[black,arrows={-Triangle[length=0.18cm,width=0.12cm]}]
(1.55,0.97) --  ++(0:0.001);
\draw[black,arrows={-Triangle[length=0.18cm,width=0.12cm]}]
(1.55,-0.97) --  ++(0:0.001);

\draw[black,arrows={-Triangle[length=0.18cm,width=0.12cm]}]
(4.05,0.76) --  ++(0:0.001);
\draw[black,arrows={-Triangle[length=0.18cm,width=0.12cm]}]
(4.05,-0.76) --  ++(0:0.001);

\draw[black,arrows={-Triangle[length=0.18cm,width=0.12cm]}]
(6.55,0.97) --  ++(0:0.001);
\draw[black,arrows={-Triangle[length=0.18cm,width=0.12cm]}]
(6.55,-0.97) --  ++(0:0.001);

\end{tikzpicture}
\caption{Jump contours $\Sigma_{R}$ for the RH problem for $R$ with $m=3$ and $s_{1} \neq 0$.}
\label{fig:contour for R s1 neq 0}
\end{figure}

The last transformation of the steepest descent is defined by
\begin{equation}\label{def of R s1 neq 0}
R(z) = \left\{ \begin{array}{l l}
S(z)P^{(\infty)}(z)^{-1}, & \mbox{for } z \in \mathbb{C}\setminus \bigcup_{j=0}^{m}\mathcal{D}_{-x_{j}}, \\
S(z)P^{(-x_{j})}(z)^{-1}, & \mbox{for } z \in \mathcal{D}_{-x_{j}}, \, j \in \{0,1,...,m\}.
\end{array} \right.
\end{equation}
By definition of the local parametrices, $R$ has no jumps and is bounded (by \eqref{match at the center s1 neq 0}) inside the $m+1$ disks. Therefore, $R$ is analytic on $\mathbb{C}\setminus \Sigma_{R}$, where $\Sigma_{R}$ consists of the boundaries of the disks, and the part of the lenses away from the disks, as shown in Figure \ref{fig:contour for R s1 neq 0}. For $z \in \Sigma_{R} \cap (\gamma_{+} \cup  \gamma_{-})$, from \eqref{def of Pinf s1 neq 0} and from the discussion at the end of Section \ref{subsection: S with s1 neq 0}, the jumps $J_{R} := R_{-}^{-1}R_{+}$ satisfy
\begin{equation}\label{estimates jumps for R lenses}
J_{R}(z) = P^{(\infty)}(z)S_{-}(z)^{-1}S_{+}(z)P^{(\infty)}(z)^{-1} = I + \bigO(e^{-c\sqrt{r}\sqrt{z}}), \qquad \mbox{as } r \to + \infty,
\end{equation}
for a certain $c>0$. Let us orient the boundaries of the disks in the clockwise direction (as in Figure \ref{fig:contour for R s1 neq 0}). For $z \in \bigcup_{j=0}^{m} \partial \mathcal{D}_{-x_{j}}$, from \eqref{matching strong -x_j s1 neq 0} and \eqref{matching strong 0 s1 neq 0}, we have
\begin{equation}\label{estimates jumps for R disks}
J_{R}(z) = P^{(\infty)}(z)P^{(-x_{j})}(z)^{-1} = I + \bigO \bigg(\frac{1}{\sqrt{r}} \bigg), \qquad \mbox{ as } r \to  +\infty.
\end{equation}
Therefore, $R$ satisfies a small norm RH problem. By standard theory for small norm RH problems \cite{DKMVZ2,DKMVZ1}, $R$ exists for sufficiently large $r$ and satisfies 
\begin{align}
& R(z) = I + \frac{R^{(1)}(z)}{\sqrt{r}} + \bigO(r^{-1}), \qquad R^{(1)}(z) = \bigO(1), & & \mbox{ as } r \to  +\infty \label{eq: asymp R inf s1 neq 0} 
\end{align}
uniformly for $z \in \mathbb{C}\setminus \Sigma_{R}$. Also, the factors $\sqrt{r}^{\pm \beta_{j}}$ in the entries of $E_{-x_{j}}$ (see \eqref{def of Ej s1 neq 0}) induce factors of the form $\sqrt{r}^{\pm 2\beta_{j}}$ in the entries of $J_{R}$ (see \eqref{matching strong -x_j s1 neq 0}). Thus, we have
\begin{equation}\label{eq: asymp der beta R inf s1 neq 0}
\partial_{\beta_{j}}R(z) = \frac{\partial_{\beta_{j}}R^{(1)}(z)}{\sqrt{r}} + \bigO \bigg( \frac{\log r}{r} \bigg), \qquad \partial_{\beta_{j}}R^{(1)}(z) = \bigO(\log r), \qquad \mbox{ as } r \to  +\infty.
\end{equation}
Furthermore, since the asymptotics \eqref{estimates jumps for R lenses} and \eqref{estimates jumps for R disks} hold uniformly for $\beta_{1},...,\beta_{m}$ in compact subsets of $i \mathbb{R}$, and uniformly in $x_{1},...,x_{m}$ in compact subsets of $(0,+\infty)$ as long as there exists $\delta > 0$ which satisfies \eqref{condition on xj in terms of delta}, the asymptotics \eqref{eq: asymp R inf s1 neq 0} and \eqref{eq: asymp der beta R inf s1 neq 0} also hold uniformly in $\beta_{1},...,\beta_{m},x_{1},...,x_{m}$ in the same way.

\vspace{0.2cm}\hspace{-0.55cm}The goal for the rest of this section is to obtain $R^{(1)}(z)$ for $z \in \mathbb{C}\setminus \bigcup_{j=0}^{m}\mathcal{D}_{-x_{j}}$ and for $z = 0$ explicitly. Since $R$ satisfies the equation
\begin{equation}
R(z) = I + \frac{1}{2\pi i} \int_{\Sigma_{R}} \frac{R_{-}(s)(J_{R}(s)-I)}{s-z}ds
\end{equation}
and since
\begin{equation}
J_{R}(z) = I + \frac{J_{R}^{(1)}(z)}{\sqrt{r}} + \bigO(r^{-1}), \qquad J_{R}^{(1)}(z) = \bigO(1),
\end{equation}
as $r \to \infty$ uniformly for $z \in \bigcup_{j=0}^{m}\mathcal{D}_{-x_{j}}$, we obtain that $R^{(1)}$ is simply given by
\begin{equation}
R^{(1)}(z) = \frac{1}{2\pi i}\int_{\bigcup_{j=0}^{m}\partial\mathcal{D}_{-x_{j}}} \frac{J_{R}^{(1)}(s)}{s-z}ds.
\end{equation}
We recall that the expressions for $J_{R}^{(1)}$ are given by \eqref{matching strong -x_j s1 neq 0} and \eqref{matching strong 0 s1 neq 0}. These expressions can be analytically continued on the interior of the disks, except at the centers where they have poles of order 1. Since the disks are oriented in the clockwise direction, by a direct residue calculation we have
\begin{equation}\label{expression for R^1 s1 neq 0}
R^{(1)}(z) = \sum_{j=0}^{m} \frac{1}{z+x_{j}}\mbox{Res}(J_{R}^{(1)}(s),s = -x_{j}), \qquad \mbox{ for } z \in \mathbb{C}\setminus \bigcup_{j=0}^{m}\mathcal{D}_{-x_{j}},
\end{equation}
and
\begin{equation}\label{expression for R^1 at 0 s1 neq 0}
R^{(1)}(0) = -\mbox{Res}\Big(\frac{J_{R}^{(1)}(s)}{s},s = 0\Big)+ \sum_{j=1}^{m} \frac{1}{x_{j}}\mbox{Res}(J_{R}^{(1)}(s),s = -x_{j}).
\end{equation}
From \eqref{asymp D at 0 s1 neq 0}, \eqref{conformal map near 0 s1 neq 0} and \eqref{matching strong 0 s1 neq 0}, we obtain
\begin{equation}
\mbox{Res}(J_{R}^{(1)}(s),s = 0) = \frac{d_{1}(1-4\alpha^{2})}{8}\begin{pmatrix}
-1 & -id_{1}^{-1} \\ -id_{1} & 1
\end{pmatrix},
\end{equation}
and with increasing effort
\begin{equation}
\mbox{Res}\Big(\frac{J_{R}^{(1)}(s)}{s},s = 0\Big) = \frac{1}{2} \begin{pmatrix}
-(d_{0}+d_{1}d_{0}^{2}) & -id_{0}^{2} \\
-i \big( \alpha^{2} + d_{0}^{2}d_{1}^{2}+2d_{0}d_{1}+\frac{3}{4} \big) & d_{0} + d_{1}d_{0}^{2}
\end{pmatrix}.
\end{equation}
From \eqref{expansion conformal map s1 neq 0} and \eqref{matching strong -x_j s1 neq 0}-\eqref{def Lambda_j s1 neq 0}, for $j \in \{1,...,m\}$, we have
\begin{multline*}
\mbox{Res}\left( J_{R}^{(1)}(s),s=-x_{j} \right) = \frac{\beta_{j}^{2}}{ic_{-x_{j}}}\begin{pmatrix}
1 & 0 \\ id_{1} & 1
\end{pmatrix}e^{-\frac{\pi i}{4}\sigma_{3}}x_{j}^{-\frac{\sigma_{3}}{4}} N \begin{pmatrix}
-1 & \widetilde{\Lambda}_{j,1} \\ -\widetilde{\Lambda}_{j,2} & 1
\end{pmatrix} \\ \times N^{-1} x_{j}^{\frac{\sigma_{3}}{4}}e^{\frac{\pi i}{4}\sigma_{3}}\begin{pmatrix}
1 & 0 \\ -id_{1} & 1
\end{pmatrix},
\end{multline*}
where
\begin{equation}
\widetilde{\Lambda}_{j,1} = \tau(\beta_{j})\Lambda_{j}^{2} \qquad \mbox{ and } \qquad \widetilde{\Lambda}_{j,2} = \tau(-\beta_{j})\Lambda_{j}^{-2}.
\end{equation}

\section{Proof of Theorem \ref{thm:s1 neq 0}}\label{Section: integration s1 >0}
This section is divided into two parts. In the first part, using the RH analysis done in Section \ref{Section: Steepest descent with s1>0}, we find large $r$ asymptotics for the differential identity
\begin{equation}
\partial_{s_{k}} \log F_{\alpha}(r\vec{x},\vec{s})= K_{\infty} + \sum_{j=1}^{m}K_{-x_{j}} + K_{0},
\end{equation}
which was obtained in \eqref{DIFF identity final form general case} with the quantities $K_{\infty}$, $K_{-x_{j}}$ and $K_{0}$ defined in \eqref{K inf}-\eqref{K 0}. In the second part, we integrate these asymptotics over the parameters $s_{1}$,...,$s_{m}$.

\subsection{Large $r$ asymptotics for the differential identity}

\paragraph{Asymptotics for $K_{\infty}$.} For $z$ outside the disks and outside the lenses, by \eqref{def of R s1 neq 0} we have
\begin{equation}
S(z) = R(z) P^{(\infty)}(z).
\end{equation}
As $z \to \infty$, we can write
\begin{equation}
R(z) = I + \frac{R_{1}}{z} + \bigO(z^{-2}),
\end{equation}
for a certain matrix $R_{1}$ independent of $z$. Thus, by \eqref{eq:Sasympinf s1 neq 0} and \eqref{eq:Pinf asympinf s1 neq 0}, we have
\begin{align*}
& T_{1} = R_{1} + P_{1}^{(\infty)}.
\end{align*}
From \eqref{eq: asymp R inf s1 neq 0} and the above expression, we infer that
\begin{align*}
& T_{1} = P_{1}^{(\infty)} + \frac{R_{1}^{(1)}}{\sqrt{r}} + \bigO(r^{-1}), \qquad \mbox{as } r \to + \infty,
\end{align*}
where $R_{1}^{(1)}$ is defined through the expansion
\begin{equation}
R^{(1)}(z) = \frac{R_{1}^{(1)}}{z} + \bigO(z^{-2}), \qquad \mbox{ as } z \to \infty.
\end{equation}
Using \eqref{K inf}, \eqref{eq:Tasympinf s1 neq 0}, \eqref{Pinf 1 12 s1 neq 0}, \eqref{eq: asymp der beta R inf s1 neq 0} and \eqref{expression for R^1 s1 neq 0}, the first part of the differential identity $K_{\infty}$ is given by
\begin{multline}\label{K inf asymp s1 neq 0}
K_{\infty} = - \frac{i}{2}\sqrt{r}\partial_{s_{k}}T_{1,12} = -\frac{i}{2}\bigg( \partial_{s_{k}} P_{1,12}^{(\infty)}\sqrt{r} + \partial_{s_{k}}R_{1,12}^{(1)} + \bigO \bigg( \frac{\log r}{\sqrt{r}} \bigg) \bigg) \\ = \frac{1}{2}\partial_{s_{k}}d_{1} \sqrt{r} - \sum_{j=1}^{m} \frac{ \partial_{s_{k}}\big( \beta_{j}^{2}(\widetilde{\Lambda}_{j,1}-\widetilde{\Lambda}_{j,2}+2i) \big)}{4i c_{-x_{j}}\sqrt{x_{j}}}+ \bigO \bigg( \frac{\log r}{\sqrt{r}} \bigg).
\end{multline}
\paragraph{Asymptotics for $K_{-x_{j}}$ with $j \in \{1,...,m\}$.} By inverting the transformations \eqref{def of S s1 neq 0} and \eqref{def of R s1 neq 0}, and using the expression for $P^{(-x_{j})}$ given by  \eqref{lol10}, for $z$ outside the lenses and inside $\mathcal{D}_{-x_{j}}$, we have
\begin{equation}\label{lol11}
T(z) = R(z)E_{-x_{j}}(z)\Phi_{\mathrm{HG}}(\sqrt{r}f_{-x_{j}}(z);\beta_{j})(s_{j}s_{j+1})^{-\frac{\sigma_{3}}{4}}e^{\frac{\pi i \alpha}{2}\theta(z) \sigma_{3}}e^{-\sqrt{r}g(z)\sigma_{3}}.
\end{equation}
If furthermore $\Im z > 0$, then by \eqref{expansion conformal map s1 neq 0} and \eqref{model RHP HG in different sector} we have
\begin{equation}
\Phi_{\mathrm{HG}}(\sqrt{r}f_{-x_{j}}(z);\beta_{j}) = \widehat{\Phi}_{\mathrm{HG}}(\sqrt{r}f_{-x_{j}}(z);\beta_{j}).
\end{equation}
Note from \eqref{def of beta_j s1 neq 0} and the connection formula for the $\Gamma$-function (see e.g. \cite[equation 5.5.3]{NIST}) that
\begin{equation}\label{relation Gamma beta_j and s_j s1 neq 0}
\frac{\sin (\pi \beta_{j})}{\pi} = \frac{1}{\Gamma(\beta_{j})\Gamma(1-\beta_{j})} = \frac{s_{j+1}-s_{j}}{2\pi i \sqrt{s_{j}s_{j+1}}}.
\end{equation}
Therefore, using \eqref{expansion conformal map s1 neq 0} and \eqref{precise asymptotics of Phi HG near 0}, as $z \to -x_{j}$ from the upper half plane and outside the lenses, we have
\begin{equation}\label{lol 2 s1 neq 0}
\Phi_{\mathrm{HG}}(\sqrt{r}f_{-x_{j}}(z);\beta_{j})(s_{j}s_{j+1})^{-\frac{\sigma_{3}}{4}} = \begin{pmatrix}
\Psi_{j,11} & \Psi_{j,12} \\
\Psi_{j,21} & \Psi_{j,22}
\end{pmatrix} (I + \bigO(z+x_{j})) \begin{pmatrix}
1 & \frac{s_{j+1}-s_{j}}{2\pi i}\log(r(z + x_{j})) \\
0 & 1
\end{pmatrix} ,
\end{equation}
where the principal branch is taken for the log and
\begin{align}
& \Psi_{j,11} = \frac{\Gamma(1-\beta_{j})}{(s_{j}s_{j+1})^{\frac{1}{4}}}, \qquad \Psi_{j,12} = \frac{(s_{j}s_{j+1})^{\frac{1}{4}}}{\Gamma(\beta_{j})} \left( \log(c_{-x_{j}}r^{-1/2}) - \frac{i\pi}{2} + \frac{\Gamma^{\prime}(1-\beta_{j})}{\Gamma(1-\beta_{j})}+2\gamma_{\mathrm{E}} \right), \nonumber \\
& \Psi_{j,21} = \frac{\Gamma(1+\beta_{j})}{(s_{j}s_{j+1})^{\frac{1}{4}}}, \qquad \Psi_{j,22} = \frac{-(s_{j}s_{j+1})^{\frac{1}{4}}}{\Gamma(-\beta_{j})} \left( \log(c_{-x_{j}}r^{-1/2}) - \frac{i\pi}{2} + \frac{\Gamma^{\prime}(-\beta_{j})}{\Gamma(-\beta_{j})} + 2\gamma_{\mathrm{E}} \right). \label{Psi j entries}
\end{align}
From \eqref{def of Gj}, \eqref{def of T s1 neq 0}, \eqref{lol11} and \eqref{lol 2 s1 neq 0} we have
\begin{equation}\label{lol12}
G_{j}(-rx_{j};r\vec{x},\vec{s}) = r^{-\frac{\sigma_{3}}{4}} R(-x_{j})E_{-x_{j}}(-x_{j})\begin{pmatrix}
\Psi_{j,11} & \Psi_{j,12} \\ \Psi_{j,21} & \Psi_{j,22}
\end{pmatrix}.
\end{equation}
In fact $K_{-x_{j}}$ does not depend on the pre-factor $r^{-\frac{\sigma_{3}}{4}}$ in \eqref{lol12}. Let us define 
\begin{equation}\label{def of Hj}
H_{j} = r^{\frac{\sigma_{3}}{4}}G_{j}(-rx_{j};r\vec{x},\vec{s}) = R(-x_{j})E_{-x_{j}}(-x_{j})\begin{pmatrix}
\Psi_{j,11} & \Psi_{j,12} \\ \Psi_{j,21} & \Psi_{j,22}
\end{pmatrix}.
\end{equation}
By a straightforward computation, we rewrite \eqref{K -xj} as follows:
\begin{equation}\label{K xj in terms of Hj}
\sum_{j=1}^{m} K_{-x_{j}} = \sum_{j=1}^m\frac{s_{j+1}-s_j}{2\pi i} (H_{j,11} \partial_{s_{k}}H_{j,21}-H_{j,21}\partial_{s_{k}}H_{j,11}).
\end{equation}
Using $\Gamma(1+z)=z\Gamma(z)$ (see e.g. \cite[equation 5.5.1]{NIST}) and \eqref{relation Gamma beta_j and s_j s1 neq 0}, we note that
\begin{equation}\label{Psi_j first column connection formula}
\Psi_{j,11}\Psi_{j,21} = \beta_{j} \frac{2\pi i}{s_{j+1}-s_{j}}, \qquad j =1,...,m.
\end{equation}
Also, from \eqref{E_j at -x_j s1 neq 0}, we have
\begin{align}
& \partial_{s_{k}}E_{-x_{j},11}(-x_{j}) = E_{-x_{j},11}(-x_{j}) \partial_{s_{k}} \log \Lambda_{j}, \qquad  \partial_{s_{k}}E_{-x_{j},12}(-x_{j}) = -E_{-x_{j},12}(-x_{j}) \partial_{s_{k}} \log \Lambda_{j}, \nonumber \\[0.2cm]
& \partial_{s_{k}}E_{-x_{j},21}(-x_{j}) = E_{-x_{j},21}(-x_{j}) \partial_{s_{k}} \log \Lambda_{j} + i E_{-x_{j},11}(-x_{j})\partial_{s_{k}}d_{1}, \label{partial E} \\[0.2cm]
& \partial_{s_{k}}E_{-x_{j},22}(-x_{j}) = -E_{-x_{j},22}(-x_{j}) \partial_{s_{k}} \log \Lambda_{j} + i E_{-x_{j},12}(-x_{j})\partial_{s_{k}}d_{1}. \nonumber
\end{align}
Therefore, using \eqref{eq: asymp R inf s1 neq 0}, \eqref{eq: asymp der beta R inf s1 neq 0}, $\det E_{-x_{j}}(-x_{j})=1$ and \eqref{def of Hj}-\eqref{partial E}, as $r \to + \infty$ we obtain
\begin{multline}\label{K xj part 1 asymp s1 neq 0}
\sum_{j=1}^{m} K_{-x_{j}} = \sum_{j=1}^{m} \frac{s_{j+1}-s_j}{2\pi i} \Big( \Psi_{j,11} \partial_{s_{k}}\Psi_{j,21} - \Psi_{j,21}\partial_{s_{k}}\Psi_{j,11} \Big) - \sum_{j=1}^{m} 2\beta_{j} \partial_{s_{k}} \log \Lambda_{j} \\  + i \partial_{s_{k}}d_{1} \sum_{j=1}^{m}\frac{s_{j+1}-s_j}{2\pi i}(E_{-x_{j},11}(-x_{j})\Psi_{j,11}+E_{-x_{j},12}(-x_{j})\Psi_{j,21})^{2} + \bigO \bigg( \frac{\log r}{\sqrt{r}} \bigg).
\end{multline}
Again using \eqref{E_j at -x_j s1 neq 0} and \eqref{Psi j entries}, we can simplify \eqref{K xj part 1 asymp s1 neq 0} further by noting that
\begin{multline}\label{K xj part 2 asymp s1 neq 0}
i \partial_{s_{k}}d_{1} \sum_{j=1}^{m}\frac{s_{j+1}-s_j}{2\pi i}(E_{-x_{j},11}(-x_{j})\Psi_{j,11}+E_{-x_{j},12}(-x_{j})\Psi_{j,21})^{2} = \sum_{j=1}^{m} \frac{\partial_{s_{k}}d_{1}}{2\sqrt{x_{j}}}\Big( \beta_{j}^{2}(\widetilde{\Lambda}_{j,1}+\widetilde{\Lambda}_{j,2})+2i\beta_{j} \Big).
\end{multline}
\paragraph{Asymptotics for $K_{0}$.} Note that we did not use the explicit expression for $R^{(1)}(-x_{j})$ to compute the asymptotics for $K_{-x_{j}}$ up to and including the constant term. The computations for $K_{0}$ are more involved and require explicitly $R^{(1)}(0)$ (given by \eqref{expression for R^1 at 0 s1 neq 0}). We start by evaluating $G_{0}(0;r\vec{x},\vec{s})$. For $z$ outside the lenses and inside $\mathcal{D}_{0}$, by \eqref{def of S s1 neq 0}, \eqref{def of P^-x1 s1 neq 0} and \eqref{def of R s1 neq 0} we have
\begin{equation}\label{lol14}
T(z) = R(z)E_{0}(z)\Phi_{\mathrm{Be}}(rf_{0}(z);\alpha)s_{1}^{-\frac{\sigma_{3}}{2}}e^{-\sqrt{r}g(z)\sigma_{3}}.
\end{equation}
From \eqref{conformal map near 0 s1 neq 0}, \eqref{lol14} and \eqref{precise asymptotics of Phi Bessel near 0}, as $z \to 0$ from outside the lenses, we have
\begin{equation}
T(z) = R(z)E_{0}(z)\Phi_{\mathrm{Be},0}(rf_{0}(z);\alpha)2^{-\alpha \sigma_{3}}s_{1}^{-\frac{\sigma_{3}}{2}}(rz)^{\frac{\alpha}{2}\sigma_{3}}\begin{pmatrix}
1 & s_{1}h(\frac{rz}{4}) \\ 0 & 1
\end{pmatrix}e^{-\sqrt{r}g(z)\sigma_{3}}.
\end{equation}
On the other hand, using \eqref{def of G_0} and  \eqref{def of T s1 neq 0}, as $z \to 0$ we have
\begin{equation}
T(z) = r^{\frac{\sigma_{3}}{4}}G_{0}(rz;r\vec{x},\vec{s})(rz)^{\frac{\alpha}{2}\sigma_{3}}\begin{pmatrix}
1 & s_{1} h(rz) \\ 0 & 1
\end{pmatrix}e^{-\sqrt{r}g(z)\sigma_{3}}.
\end{equation}
Therefore, we obtain
\begin{equation}\label{lol15}
G_{0}(0;r\vec{x},\vec{s}) = r^{-\frac{\sigma_{3}}{4}}R(0)E_{0}(0)\Psi_{0}, \qquad \Psi_{0} := \left\{ \begin{array}{l l}
\Phi_{\mathrm{Be},0}(0;\alpha)2^{-\alpha\sigma_{3}}s_{1}^{-\frac{\sigma_{3}}{2}}, & \mbox{if } \alpha \neq 0, \\
\Phi_{\mathrm{Be},0}(0;0)s_{1}^{-\frac{\sigma_{3}}{2}}\begin{pmatrix}
1 & - \frac{s_{1}}{\pi i}\log 2 \\ 0 & 1
\end{pmatrix}, & \mbox{if } \alpha = 0,
\end{array} \right.
\end{equation}
and $\Phi_{\mathrm{Be}}(0;\alpha)$ is computed in the appendix, see \eqref{precise matrix at 0 in asymptotics of Phi Bessel near 0}. In the same way as for $K_{-x_{j}}$, we define
\begin{equation}\label{def of H0 s1 neq 0}
H_{0} = r^{\frac{\sigma_{3}}{4}}G_{0}(0;r\vec{x},\vec{s}) = R(0)E_{0}(0) \Psi_{0},
\end{equation}
and we simplify $K_{0}$ (given by \eqref{K 0}) as follows
\begin{equation}\label{simplification K0 in terms of H0}
K_{0} = \left\{ \begin{array}{l l}
\ds \frac{s_{1}}{2\pi i} \Big( H_{0,11}\partial_{s_{k}} H_{0,21} - H_{0,21}\partial_{s_{k}} H_{0,11} \Big) & \mbox{if } \alpha = 0, \\[0.3cm]
\alpha \Big( H_{0,21} \partial_{s_{k}}H_{0,12} - H_{0,11} \partial_{s_{k}}H_{0,22} \Big) & \mbox{if } \alpha \neq 0.
\end{array} \right.
\end{equation}
We start with the case $\alpha = 0$. Using \eqref{E0 at 0 s1 neq 0}, \eqref{eq: asymp R inf s1 neq 0}-\eqref{eq: asymp der beta R inf s1 neq 0}, \eqref{lol15}-\eqref{simplification K0 in terms of H0}, and the fact that $R^{(1)}$ is traceless, after a careful calculation, we obtain the following asymptotics as $r \to + \infty$:
\begin{multline}\label{K 0 part 1 asymp s1 neq 0}
K_{0} = \frac{s_{1}}{2\pi i} (H_{0,11} \partial_{s_{k}}H_{0,21}-H_{0,21}\partial_{s_{k}}H_{0,11}) = \frac{1}{2} \partial_{s_{k}}d_{1} \sqrt{r} \\ -\frac{1}{2} \Big( d_{1} \partial_{s_{k}} (R_{11}^{(1)}(0)-R_{22}^{(1)}(0)) +id_{1}^{2} \partial_{s_{k}} R_{12}^{(1)}(0)+i \partial_{s_{k}} R_{21}^{(1)}(0) \Big) + \bigO\bigg( \frac{\log r}{\sqrt{r}} \bigg).
\end{multline}
The subleading term in \eqref{K 0 part 1 asymp s1 neq 0} can be evaluated using the explicit form for $R^{(1)}(0)$ given by \eqref{expression for R^1 at 0 s1 neq 0}: 
\begin{multline}\label{K 0 part 2 asymp s1 neq 0}
-\frac{1}{2} \Big( d_{1} \partial_{s_{k}} (R_{11}^{(1)}(0)-R_{22}^{(1)}(0)) +id_{1}^{2} \partial_{s_{k}} R_{12}^{(1)}(0)+i \partial_{s_{k}} R_{21}^{(1)}(0) \Big) = \frac{d_{0}\partial_{s_{k}}d_{1}}{2}\\+ \sum_{j=1}^{m} \frac{1}{4 i c_{-x_{j}}\sqrt{x_{j}}}\partial_{s_{k}}\big( \beta_{j}^{2}(\widetilde{\Lambda}_{j,1}-\widetilde{\Lambda}_{j,2}-2i)\big)-\partial_{s_{k}}d_{1}\sum_{j=1}^{m} \frac{\beta_{j}^{2}(\widetilde{\Lambda}_{j,1}+\widetilde{\Lambda}_{j,2})}{2c_{-x_{j}}x_{j}}.
\end{multline}
Now, we evalute $K_{0}$ for the case $\alpha \neq 0$. Using the formula $\alpha \Gamma(\alpha) = \Gamma(1+\alpha)$, \eqref{E0 at 0 s1 neq 0}, \eqref{eq: asymp R inf s1 neq 0}-\eqref{eq: asymp der beta R inf s1 neq 0}, \eqref{lol15}-\eqref{simplification K0 in terms of H0}, and the fact that $R^{(1)}(0)$ is traceless, after a lot of cancellations, we obtain
\begin{multline}
K_{0} = \alpha \Big( H_{0,21} \partial_{s_{k}}H_{0,12} - H_{0,11} \partial_{s_{k}}H_{0,22} \Big) = \frac{1}{2} \partial_{s_{k}}d_{1} \sqrt{r} - \frac{\alpha}{2} \partial_{s_{k}} \big(\log s_{1} \big) \\ -\frac{1}{2} \Big( d_{1} \partial_{s_{k}} (R_{11}^{(1)}(0)-R_{22}^{(1)}(0)) +id_{1}^{2} \partial_{s_{k}} R_{12}^{(1)}(0)+i \partial_{s_{k}} R_{21}^{(1)}(0) \Big) + \bigO\bigg( \frac{\log r}{\sqrt{r}} \bigg),
\end{multline}
as $r \to + \infty$, which is the same formula as \eqref{K 0 part 1 asymp s1 neq 0} for $\alpha = 0$, plus the extra factor $- \frac{\alpha}{2} \partial_{s_{k}} (\log s_{1} )$ which can be rewritten using \eqref{def of beta_j s1 neq 0} as follows:
\begin{align*}
- \frac{\alpha}{2} \partial_{s_{k}} \big(\log s_{1} \big) = \pi i \alpha \, \partial_{s_{k}} \big( \beta_{1} + \ldots + \beta_{m} \big).
\end{align*}

\paragraph{Asymptotics for the differential identity  \eqref{DIFF identity final form general case}.} By summing the contributions $K_{0}$, $K_{-x_{j}}$, $j=1,...,m$ and $K_{\infty}$  using \eqref{K inf asymp s1 neq 0}, \eqref{K xj part 1 asymp s1 neq 0}, \eqref{K xj part 2 asymp s1 neq 0}, \eqref{K 0 part 1 asymp s1 neq 0} and \eqref{K 0 part 2 asymp s1 neq 0}, and by substituting the expression for $c_{-x_{j}}$ given by \eqref{expansion conformal map s1 neq 0}, and the expression for $d_{0}$ given by \eqref{d_ell in terms of beta_j s1 neq 0}, a lot of terms cancel each other out and we obtain
\begin{multline}\label{lol5 s1 neq 0}
\partial_{s_k}\log F_{\alpha}(r\vec{x},\vec{s}) = \partial_{s_{k}}d_{1}\sqrt{r}  + \pi i \alpha \sum_{j=1}^{m} \partial_{s_{k}} \beta_{j} - \sum_{j=1}^{m} \Big( 2\beta_{j} \partial_{s_{k}} \log \Lambda_{j} + \partial_{s_{k}} (\beta_{j}^{2}) \Big) \\ + \sum_{j=1}^{m}  \frac{s_{j+1}-s_{j}}{2\pi i}\big( \Psi_{j,11}\partial_{s_{k}}\Psi_{j,21}-\Psi_{j,21}\partial_{s_{k}}\Psi_{j,11} \big) + \bigO\Big( \frac{\log r}{\sqrt{r}} \Big), \qquad \mbox{as } r \to + \infty.
\end{multline}
Using the explicit expressions for $\Psi_{j,11}$ and $\Psi_{j,21}$ (see \eqref{Psi j entries}) together with the relation \eqref{Psi_j first column connection formula}, we have
\begin{equation}\label{lol3 s1 neq 0}
\sum_{j=1}^{m} \frac{s_{j+1}-s_{j}}{2\pi i}\big( \Psi_{j,11}\partial_{s_{k}}\Psi_{j,21}-\Psi_{j,21}\partial_{s_{k}}\Psi_{j,11} \big) = \sum_{j=1}^{m}\beta_{j} \partial_{s_{k}} \log \frac{\Gamma(1+\beta_{j})}{\Gamma(1-\beta_{j})}.
\end{equation}
Also, using \eqref{def Lambda_j s1 neq 0}, we have
\begin{equation}\label{lol4 s1 neq 0}
\sum_{j=1}^{m} - 2\beta_{j} \partial_{s_{k}} \log \Lambda_{j} = -2 \sum_{j=1}^{m} \beta_{j} \partial_{s_{k}}(\beta_{j}) \log (4x_{j}c_{-x_{j}}\sqrt{r}) -2\sum_{j=1}^{m} \beta_{j} \sum_{\substack{\ell = 1 \\ \ell \neq j}}^{m} \partial_{s_{k}}(\beta_{\ell})\log(\widetilde{T}_{\ell,j}).
\end{equation}
It will more convenient to integrate with respect to $\beta_{1},...,\beta_{m}$ instead of $s_{1},...,s_{m}$. Therefore, we define
\begin{equation}\label{def of F tilde}
\widetilde{F}_{\alpha}(r \vec{x}, \vec{\beta}) = F_{\alpha}(r \vec{x},\vec{s}),
\end{equation}
where $\vec{\beta} = (\beta_{1},...,\beta_{m})$ and $\vec{s} = (s_{1},...,s_{m})$ are related via the relations \eqref{def of beta_j s1 neq 0}. By substituting \eqref{lol3 s1 neq 0} and \eqref{lol4 s1 neq 0} into \eqref{lol5 s1 neq 0}, and by writing the derivative with respect to $\beta_{k}$ instead of $s_{k}$, as $r \to + \infty$ we obtain
\begin{multline}\label{lol5 part 2 s1 neq 0}
\partial_{\beta_k}\log\widetilde{F}_{\alpha}(r \vec{x}, \vec{\beta}) = \partial_{\beta_{k}}d_{1} \sqrt{r}  -2 \sum_{j=1}^{m} \beta_{j} \partial_{\beta_{k}}(\beta_{j}) \log (4x_{j}c_{-x_{j}}\sqrt{r}) + \pi i \alpha \\ -2\sum_{j=1}^{m} \beta_{j} \sum_{\substack{\ell = 1 \\ \ell \neq j}}^{m} \partial_{\beta_{k}}(\beta_{\ell})\log(\widetilde{T}_{\ell,j}) - \sum_{j=1}^{m} \partial_{\beta_{k}}(\beta_{j}^{2}) + \sum_{j=1}^{m} \beta_{j} \partial_{\beta_{k}} \log \frac{\Gamma(1+\beta_{j})}{\Gamma(1-\beta_{j})} + \bigO\Big(\frac{\log r}{\sqrt{r}}\Big).
\end{multline}
Using the value of $d_{1}$ in \eqref{d_ell in terms of beta_j s1 neq 0} and the value of $c_{-x_{j}}$ in \eqref{expansion conformal map s1 neq 0}, the above asymptotics can be rewritten more explicitly as follows
\begin{multline}\label{DIFF IDENTITY s1 neq 0}
\partial_{\beta_k}\log \widetilde{F}_{\alpha}(r \vec{x}, \vec{\beta}) = -2i \sqrt{rx_{k}} -2 \beta_{k} \log(4\sqrt{rx_{k}}) + \pi i \alpha \\ - 2 \sum_{\substack{j=1 \\ j \neq k}}^{m}\beta_{j} \log(\widetilde{T}_{k,j}) - 2\beta_{k} + \beta_{k} \partial_{\beta_{k}} \log \frac{\Gamma(1+\beta_{k})}{\Gamma(1-\beta_{k})} + \bigO\Big( \frac{\log r}{\sqrt{r}} \Big).
\end{multline}
\subsection{Integration of the differential identity}
By the steepest descent of Section \ref{Section: Steepest descent with s1>0} (see in particular the discussion in Section \ref{subsection Small norm s1 neq 0}), the asymptotics \eqref{DIFF IDENTITY s1 neq 0} are valid uniformly for $\beta_{1},...,\beta_{m}$ in compact subsets of $i \mathbb{R}$. First, we use \eqref{DIFF IDENTITY s1 neq 0} with $\beta_{2} = 0 = \beta_{3} = ... = \beta_{m}$, and we integrate in $\beta_{1}$ from $\beta_{1} = 0$ to an arbitrary $\beta_{1} \in i \mathbb{R}$. It is important for us to note the following relation (see e.g. \cite{ItsKrasovsky}):
\begin{equation}\label{integral of Gamma with Barnes}
\int_{0}^{\beta} x \partial_{x} \log \frac{\Gamma(1+x)}{\Gamma(1-x)}dx = \beta^{2} + \log G(1+\beta)G(1-\beta),
\end{equation}
where $G$ is Barnes' $G$-function. Let us use the notation $\vec{\beta}_{1} = (\beta_{1},0,...,0)$. After integration of \eqref{DIFF IDENTITY s1 neq 0} (with $k=1$) from $\vec{\beta} = \vec{0} = (0,...,0)$ to $\vec{\beta} = \vec{\beta}_{1}$, we obtain
\begin{multline*}
\log \frac{\widetilde{F}_{\alpha}(r\vec{x},\vec{\beta}_{1})}{\widetilde{F}_{\alpha}(r\vec{x},\vec{0})} = -2i\beta_{1} \sqrt{r x_{1}} - \beta_{1}^{2} \log(4\sqrt{rx_{1}}) + \pi i \alpha \beta_{1} +\log(G(1+\beta_{1})G(1-\beta_{1})) + \bigO \bigg( \frac{\log r}{\sqrt{r}} \bigg),
\end{multline*}
as $r \to + \infty$. Now, we use \eqref{DIFF IDENTITY s1 neq 0} with $k = 2$ and $\beta_{3}=...=\beta_{m} = 0$, $\beta_{1}$ fixed but not necessarily $0$, and we integrate in $\beta_{2}$. With the notation $\vec{\beta}_{2} = (\beta_{1},\beta_{2},0,...,0)$, as $r \to + \infty$ we obtain
\begin{multline}
\log  \frac{\widetilde{F}_{\alpha}(r\vec{x},\vec{\beta}_{2})}{\widetilde{F}_{\alpha}(r\vec{x},\vec{\beta}_{1})} = -2i \beta_{2} \sqrt{rx_{2}} - \beta_{2}^{2} \log(4\sqrt{rx_{2}}) + \pi i \alpha \beta_{2} \\ - 2 \beta_{1}\beta_{2} \log(\widetilde{T}_{2,1}) +\log(G(1+\beta_{2})G(1-\beta_{2})) + \bigO \bigg( \frac{\log r}{\sqrt{r}} \bigg).
\end{multline}
By integrating successively in $\beta_{3},...,\beta_{m}$, and then by summing the expressions, we obtain
\begin{multline}
\log \frac{\widetilde{F}_{\alpha}(r\vec{x},\vec{\beta})}{\widetilde{F}_{\alpha}(r\vec{x},\vec{0})} = - \sum_{j=1}^{m} 2i \beta_{j} \sqrt{rx_{j}} - \sum_{j=1}^{m} \beta_{j}^{2} \log (4\sqrt{rx_{j}}) + \pi i \alpha \sum_{j=1}^{m} \beta_{j} \\ - 2 \sum_{1 \leq j < k \leq m} \beta_{j}\beta_{k} \log(\widetilde{T}_{j,k}) + \sum_{j=1}^{m} \log (G(1+\beta_{j})G(1-\beta_{j})) + \bigO \bigg( \frac{\log r}{\sqrt{r}} \bigg).
\end{multline}
By \eqref{def of F tilde} and \eqref{F Fredholm}, we have $\widetilde{F}_{\alpha}(r\vec{x},\vec{0}) = F_{\alpha}(r\vec{x},\vec{1}) = 1$. This finishes the proof of Theorem \ref{thm:s1 neq 0} (after identifying $u_{j}=-2\pi i \beta_{j}$).

\section{Large $r$ asymptotics for $\Phi$ with $s_{1} = 0$}\label{Section: Steepest descent with s1=0}
In this section, we perform an asymptotic analysis of $\Phi(z;r\vec{x},\vec{s})$ as $r \to + \infty$ and $s_{1} = 0$. This steepest descent differs from the one done in Section \ref{Section: Steepest descent with s1>0} in several aspects. In particular, we need a different $g$-function, the local parametrix at $-x_{1}$ is now built in terms of Bessel functions (instead of hypergeometric functions for $s_{1} > 0$), and there is no need for a local parametrix at $0$ (as opposed to Section \ref{Section: Steepest descent with s1>0}). On the level of the parameters, we assume that $s_{1}=0$, that $s_{2},...,s_{m}$ are in a compact subset of $(0,+\infty)$ and that $x_{1},...,x_{m}$ are in a compact subset of $(0,+\infty)$ in such a way that there exists $\delta > 0$ independent of $r$ such that
\begin{equation}
\min_{1 \leq j < k \leq m} x_{k}-x_{j} \geq \delta.
\end{equation}

\subsection{Normalization of the RH problem with a $g$-function}\label{subsection: g function s1 = 0}
Since $s_{1}=0$, we need a different $g$-function than \eqref{def of g s1 neq 0}. We define
\begin{equation}\label{def of g s1=0}
g(z)=\sqrt{z+x_{1}},
\end{equation}
where the principal branch is taken. It satisfies
\begin{equation}\label{g at inf s1 = 0}
g(z) = \sqrt{z} + \frac{x_{1}}{2}z^{-1/2} + \bigO(z^{-3/2}), \qquad \mbox{ as } z \to \infty.
\end{equation}
We define the first transformation $T$ similarly to \eqref{def of T s1 neq 0} (however with an extra pre-factor matrix to compensate the asymptotic behavior \eqref{g at inf s1 = 0} of the $g$-function)
\begin{equation}\label{def of T}
T(z) = \begin{pmatrix}
1 & 0 \\ i\frac{x_{1}}{2}\sqrt{r} & 1
\end{pmatrix}r^{\frac{\sigma_{3}}{4}}\Phi(rz;r\vec{x},\vec{s})e^{-\sqrt{r}g(z)\sigma_{3}}.
\end{equation}
The asymptotics \eqref{Phi inf} of $\Phi$ then lead after some calculation to 
\begin{equation}
\label{eq:Tasympinf}
T(z) = \left( I + \frac{T_{1}}{z} + \bigO\left(z^{-2}\right) \right) z^{-\frac{\sigma_3}{4}} N, \qquad T_{1,12} = \frac{\Phi_{1,12}(r\vec{x},\vec{s})}{\sqrt{r}} + i \sqrt{r} \frac{x_{1}}{2}
\end{equation}
as $z \to\infty$. For $z \in (-\infty,-x_{1})$, since $g_{+}(z)+g_{-}(z) = 0$, the jumps for $T$ can be factorized in the same way as \eqref{factorization of the jump}.
\subsection{Opening of the lenses}
Around each interval $(-x_{j},-x_{j-1})$, $j = 2,...,m$, we open lenses $\gamma_{j,+}$ and $\gamma_{j,-}$, lying in the upper and lower half plane respectively, as shown in Figure \ref{fig:contour for S}. Let us also denote $\Omega_{j,+}$ (resp. $\Omega_{j,-}$) for the region inside the lenses around $(-x_{j},-x_{j-1})$ in the upper half plane (resp. in the lower half plane).
The next transformation is defined by
\begin{equation}\label{def of S}
S(z)= T(z) \prod_{j=2}^{m} \left\{ \begin{array}{l l}
\begin{pmatrix}
1 & 0 \\
-s_{j}^{-1}e^{\pi i \alpha}e^{-2\sqrt{r}g(z)} & 1
\end{pmatrix}, & \mbox{if } z \in \Omega_{j,+}, \\
\begin{pmatrix}
1 & 0 \\
s_{j}^{-1}e^{-\pi i \alpha}e^{-2\sqrt{r}g(z)} & 1
\end{pmatrix}, & \mbox{if } z \in \Omega_{j,-}, \\
I, & \mbox{if } z \in \mathbb{C}\setminus(\Omega_{j,+}\cup \Omega_{j,-}).
\end{array} \right.
\end{equation}
It is straightforward to verify from the RH problem for $\Phi$ and from Section \ref{subsection: g function s1 = 0} that $S$ satisfies the following RH problem.
\subsubsection*{RH problem for $S$}
\begin{enumerate}[label={(\alph*)}]
\item[(a)] $S : \C \backslash \Gamma_{S} \rightarrow \C^{2\times 2}$ is analytic, with
\begin{equation}\label{eq:defGamma}
\Gamma_{S}=(-\infty,0)\cup \gamma_{+}\cup \gamma_{-}, \qquad \gamma_{\pm} = \bigcup_{j=2}^{m+1} \gamma_{j,\pm},
\end{equation}
where $\gamma_{m+1,\pm} := -x_{m} + e^{\pm \frac{2\pi i}{3}}(0,+\infty)$, and $\Gamma_{S}$ is oriented as shown in Figure \ref{fig:contour for S}.
\item[(b)] The jumps for $S$ are given by
\begin{align*}
& S_{+}(z) = S_{-}(z)\begin{pmatrix}
0 & s_{j} \\ -s_{j}^{-1} & 0
\end{pmatrix}, & & z \in (-x_{j},-x_{j-1}), \, j = 2,...,m+1, \\
& S_{+}(z) = S_{-}(z)e^{\pi i \alpha \sigma_{3}}, & & z \in (-x_{1},0),\\
& S_{+}(z) = S_{-}(z)\begin{pmatrix}
1 & 0 \\
s_{j}^{-1}e^{\pm \pi i \alpha}e^{-2\sqrt{r}g(z)} & 1
\end{pmatrix}, & & z \in \gamma_{j,\pm}, \, j = 2,...,m+1,
\end{align*}
where $x_{m+1} = +\infty$ (we recall that $x_{0}=0$ and $s_{m+1} = 1$).
\item[(c)] As $z \rightarrow \infty$, we have
\begin{equation}
\label{eq:Sasympinf}
S(z) = \left( I + \frac{T_1}{z} + \bigO\left( z^{-2} \right) \right) z^{-\frac{\sigma_3}{4} } N.
\end{equation}
As $z \to -x_j$ from outside the lenses, $j = 1, ..., m$, we have
\begin{equation}\label{local behaviour near -xj of S s1 = 0}
S(z) = \begin{pmatrix}
\bigO(1) & \bigO(\log(z+x_{j})) \\
\bigO(1) & \bigO(\log(z+x_{j}))
\end{pmatrix}.
\end{equation}
As $z \to 0$, we have
\begin{equation}\label{local behaviour of S s1 = 0}
S(z) = \begin{pmatrix}
\bigO(1) & \bigO(1) \\
\bigO(1) & \bigO(1)
\end{pmatrix}z^{\frac{\alpha}{2}\sigma_{3}}.
\end{equation}
\end{enumerate}
Since $\Re g(z) > 0$ for all $z \in \mathbb{C}\setminus (-\infty,-x_{1}]$ and $\Re g_{\pm}(z) = 0$ for $z \in (-\infty,-x_{1})$, the jump matrices for $S$ tend to the identity matrix exponentially fast as $r \to + \infty$ on the lenses. This convergence is uniform for $z$ outside of fixed neighborhoods of $-x_{j}$, $j \in \{1,...,m\}$, but is not uniform as $r \to + \infty$ and simultaneously $z \to -x_{j}$, $j \in \{1,...,m\}$.
\begin{figure}
\centering
\begin{tikzpicture}
\draw[fill] (0,0) circle (0.05);
\draw (0,0) -- (8,0);
\draw (0,0) -- (120:3);
\draw (0,0) -- (-120:3);
\draw (0,0) -- (-3,0);

\draw (0,0) .. controls (1,1.3) and (2,1.3) .. (3,0);
\draw (0,0) .. controls (1,-1.3) and (2,-1.3) .. (3,0);
\draw (3,0) .. controls (3.5,1) and (4.5,1) .. (5,0);
\draw (3,0) .. controls (3.5,-1) and (4.5,-1) .. (5,0);

\draw[fill] (3,0) circle (0.05);
\draw[fill] (5,0) circle (0.05);
\draw[fill] (8,0) circle (0.05);

\node at (0.15,-0.3) {$-x_{m}$};
\node at (3,-0.3) {$-x_{2}$};
\node at (5.1,-0.3) {$-x_{1}$};
\node at (8.4,-0.3) {$0=x_{0}$};
\node at (-3,-0.3) {$-\infty=-x_{m+1}$};

\draw[black,arrows={-Triangle[length=0.18cm,width=0.12cm]}]
(-120:1.5) --  ++(60:0.001);
\draw[black,arrows={-Triangle[length=0.18cm,width=0.12cm]}]
(120:1.3) --  ++(-60:0.001);
\draw[black,arrows={-Triangle[length=0.18cm,width=0.12cm]}]
(180:1.5) --  ++(0:0.001);

\draw[black,arrows={-Triangle[length=0.18cm,width=0.12cm]}]
(0:1.5) --  ++(0:0.001);
\draw[black,arrows={-Triangle[length=0.18cm,width=0.12cm]}]
(0:4) --  ++(0:0.001);
\draw[black,arrows={-Triangle[length=0.18cm,width=0.12cm]}]
(0:6.5) --  ++(0:0.001);

\draw[black,arrows={-Triangle[length=0.18cm,width=0.12cm]}]
(1.55,0.97) --  ++(0:0.001);
\draw[black,arrows={-Triangle[length=0.18cm,width=0.12cm]}]
(1.55,-0.97) --  ++(0:0.001);

\draw[black,arrows={-Triangle[length=0.18cm,width=0.12cm]}]
(4.05,0.76) --  ++(0:0.001);
\draw[black,arrows={-Triangle[length=0.18cm,width=0.12cm]}]
(4.05,-0.76) --  ++(0:0.001);

\end{tikzpicture}
\caption{Jump contours $\Gamma_{S}$ for the model RH problem for $S$ with $m=3$ and $s_{1} = 0$.}
\label{fig:contour for S}
\end{figure}
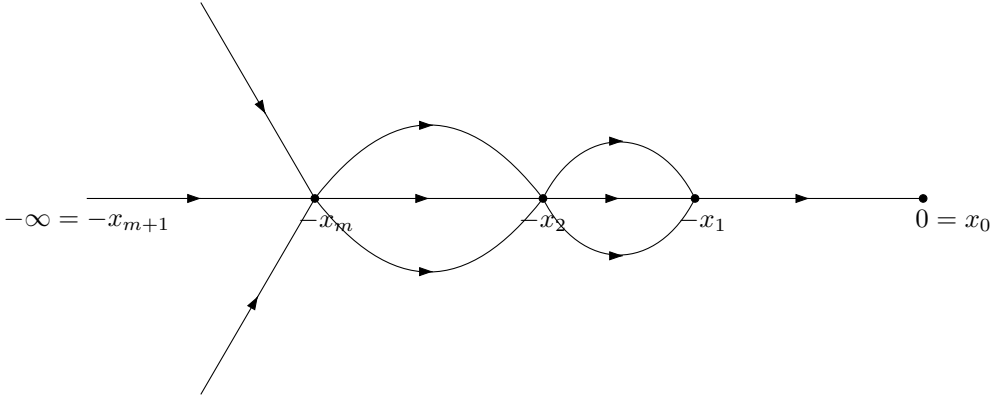
\subsection{Global parametrix}
By ignoring the jumps for $S$ that are pointwise exponentially close to the identity matrix as $r \to + \infty$, we are left with an RH problem for $P^{(\infty)}$ which is similar to the one done in Section \ref{subsection: Global param s1 neq 0}. However, there are some important differences: the jumps along $(-x_{1},0)$ and the behavior near $0$. It will appear later in Section \ref{subsection Small norm s1 = 0} that $P^{(\infty)}$ is a good approximation for $S$ away from neighborhoods of $-x_{j}$, $j = 1,...,m$. In particular, $P^{(\infty)}$ will be a good approximation for $S$ in a neighborhood of $0$, and thus we will not need a local parametrix near $0$ in this steepest descent analysis. 

\subsubsection*{RH problem for $P^{(\infty)}$}
\begin{enumerate}[label={(\alph*)}]
\item[(a)] $P^{(\infty)} : \C \backslash (-\infty,0] \rightarrow \C^{2\times 2}$ is analytic.
\item[(b)] The jumps for $P^{(\infty)}$ are given by
\begin{align*}
& P^{(\infty)}_{+}(z) = P^{(\infty)}_{-}(z)\begin{pmatrix}
0 & s_{j} \\ -s_{j}^{-1} & 0
\end{pmatrix}, & & z \in (-x_{j},-x_{j-1}), \, j = 2,...,m+1, \\
& P^{(\infty)}_{+}(z) = P^{(\infty)}_{-}(z)e^{\pi i \alpha \sigma_{3}}, & & z \in (-x_{1},0).
\end{align*}
\item[(c)] As $z \rightarrow \infty$, we have
\begin{equation}
\label{eq:Pinf asympinf}
P^{(\infty)}(z) = \left( I + \frac{P^{(\infty)}_{1}}{z} + \bigO\left( z^{-2} \right) \right) z^{-\frac{\sigma_3}{4}} N,
\end{equation}
for a certain matrix $P_{1}^{(\infty)}$ independent of $z$.
\item[(d)] As $z \to -x_{j}$, $j \in \{2,...,m\}$, we have $P^{(\infty)}(z) = \begin{pmatrix}
\bigO(1) & \bigO(1) \\ \bigO(1) & \bigO(1)
\end{pmatrix}$.

As $z \to -x_{1}$, we have $P^{(\infty)}(z) = \begin{pmatrix}
\bigO((z+x_{1})^{-\frac{1}{4}}) & \bigO((z+x_{1})^{-\frac{1}{4}}) \\ \bigO((z+x_{1})^{-\frac{1}{4}}) & \bigO((z+x_{1})^{-\frac{1}{4}})
\end{pmatrix}$.

As $z \to 0$, we have $P^{(\infty)}(z) = \begin{pmatrix}
\bigO(1) & \bigO(1) \\ \bigO(1) & \bigO(1)
\end{pmatrix}z^{\frac{\alpha}{2}\sigma_{3}}$.
\end{enumerate}
Note that the condition (d) for the RH problem for $P^{(\infty)}$ does not come from the RH problem for $S$ (with the exception of the behavior at $0$). It is added to ensure uniqueness of the solution. The construction of $P^{(\infty)}$ relies on the following Szeg\"{o} functions
\begin{align*}
& D_{\alpha}(z) = \exp\left( \frac{\alpha}{2} \sqrt{z+x_{1}}\int_{0}^{x_{1}} \frac{1}{\sqrt{x_{1}-u}}\frac{du}{z+u} \right) = \left( \frac{\sqrt{z+x_{1}}+\sqrt{x_{1}}}{\sqrt{z+x_{1}}-\sqrt{x_{1}}}  \right)^{\frac{\alpha}{2}}, \\
& D_{\vec{s}}(z) = \exp \left( \frac{\sqrt{z+x_{1}}}{2\pi} \sum_{j=2}^{m} \log s_{j} \int_{x_{j-1}}^{x_{j}}\frac{du}{\sqrt{u-x_{1}}(z+u)} \right).
\end{align*}
They satisfy the following jumps
\begin{align*}
& D_{\alpha,+}(z)D_{\alpha,-}(z) = 1, & & \mbox{for } z \in (-\infty,-x_{1}), \\
& D_{\alpha,+}(z) = D_{\alpha,-}(z)e^{-\pi i \alpha}, & & \mbox{for } z \in (-x_{1},0),\\
& D_{\vec{s},+}(z)D_{\vec{s},-}(z) = s_{j}, & &  \mbox{for } z \in (-x_{j},-x_{j-1}), \, j = 2,...,m+1.
\end{align*}
Furthermore, as $z \to \infty$, we have
\begin{equation}
\begin{array}{r c l}
\ds D_{\alpha}(z) & = & \ds \exp \left( \sum_{\ell = 1}^{k} \frac{d_{\ell,\alpha}}{(z+x_{1})^{\ell-\frac{1}{2}}} + \bigO(z^{-k-\frac{1}{2}}) \right), \\[0.2cm]
\ds D_{\vec{s}}(z) & = & \ds \exp \left( \sum_{\ell = 1}^{k} \frac{d_{\ell,\vec{s}}}{(z+x_{1})^{\ell-\frac{1}{2}}} + \bigO(z^{-k-\frac{1}{2}}) \right),
\end{array}
\end{equation}
where $k \in \mathbb{N}_{>0}$ is arbitrary and
\begin{align}
& d_{\ell,\alpha} = \frac{\alpha}{2} \int_{0}^{x_{1}}(x_{1}-u)^{\ell - \frac{3}{2}}du = \frac{\alpha \, x_{1}^{\ell - \frac{1}{2}}}{2\ell -1}, \\
& d_{\ell,\vec{s}} = \frac{(-1)^{\ell - 1}}{2\pi} \sum_{j=2}^{m}  \log s_{j} \int_{x_{j-1}}^{x_{j}} (u-x_{1})^{\ell - \frac{3}{2}}du = \frac{(-1)^{\ell -1}}{\pi (2\ell -1)}\sum_{j=2}^{m} \log s_{j} \Big( (x_{j}-x_{1})^{\ell -\frac{1}{2}}-(x_{j-1}-x_{1})^{\ell-\frac{1}{2}} \Big) \nonumber.
\end{align}
For $\ell \geq 1$, we define $d_{\ell} = d_{\ell,\alpha} + d_{\ell,\vec{s}}$.
Let us finally define
\begin{equation}\label{def of Pinf}
P^{(\infty)}(z) = \begin{pmatrix}
1 & 0 \\ id_{1} & 1
\end{pmatrix}(z+x_{1})^{-\frac{\sigma_{3}}{4}}ND(z)^{-\sigma_{3}},
\end{equation}
where the principal branch is taken for the root, and where $D(z) = D_{\alpha}(z)D_{\vec{s}}(z)$. From the above properties of $D_{\alpha}$ and $D_{\vec{s}}$, one can check that $P^{(\infty)}$ satisfies criteria (a), (b) and (c) of the RH problem for $P^{(\infty)}$, with
\begin{equation}\label{Pinf 1 12 s1 = 0}
P_{1,12}^{(\infty)} = i d_{1}.
\end{equation}
The rest of the current section consists of computing of the first terms in the asymptotics of $D(z)$ as $z \to -x_{j}$, $j = 0,1,...,m$. In particular, it will  prove that $P^{(\infty)}$ defined in \eqref{def of Pinf} satisfies condition (d) of the RH problem for $P^{(\infty)}$. After integrations, we can rewrite $D_{\vec{s}}$ as follows
\begin{equation}
D_{\vec{s}}(z) = \prod_{j=2}^{m} D_{s_{j}}(z),
\end{equation}
where
\begin{equation}
D_{s_{j}}(z) = \left( \frac{(\sqrt{z+x_{1}}-i\sqrt{x_{j-1}-x_{1}})(\sqrt{z+x_{1}}+i\sqrt{x_{j}-x_{1}})}{(\sqrt{z+x_{1}}-i\sqrt{x_{j}-x_{1}})(\sqrt{z+x_{1}}+i\sqrt{x_{j-1}-x_{1}})} \right)^{\frac{\log s_{j}}{2\pi i}}.
\end{equation}
As $z \to -x_{j}$, $j \in \{2,...,m\}$, $\Im z > 0$, we have
\begin{equation}
D_{s_{j}}(z) = \sqrt{s_{j}}T_{j,j}^{\frac{\log s_{j}}{2\pi i}}(z+x_{j})^{-\frac{\log s_{j}}{2\pi i}}(1+\bigO(z+x_{j})), \quad T_{j,j} = 4(x_{j}-x_{1})\frac{\sqrt{x_{j}-x_{1}}-\sqrt{x_{j-1}-x_{1}}}{\sqrt{x_{j}-x_{1}}+\sqrt{x_{j-1}-x_{1}}}.
\end{equation}
As $z \to -x_{j-1}$, $j \in \{3,...,m\}$, $\Im z > 0$, we have
\begin{equation}
D_{s_{j}}(z) = T_{j,j-1}^{\frac{\log s_{j}}{2\pi i}}(z+x_{j-1})^{\frac{\log s_{j}}{2\pi i}}(1+\bigO(z+x_{j-1})), \quad T_{j,j-1} = \frac{1}{4(x_{j-1}-x_{1})}\frac{\sqrt{x_{j}-x_{1}}+\sqrt{x_{j-1}-x_{1}}}{\sqrt{x_{j}-x_{1}}-\sqrt{x_{j-1}-x_{1}}}.
\end{equation}
For $j \in \{2,...,m\}$, as $z \to -x_{k}$, $k \in \{2,...,m\}$, $k \neq j,j-1$, $\Im z > 0$, we have
\begin{equation}
D_{s_{j}}(z) = T_{j,k}^{\frac{\log s_{j}}{2\pi i}}(1+\bigO(z+x_{k})), \quad T_{j,k} = \frac{(\sqrt{x_{k}-x_{1}}-\sqrt{x_{j-1}-x_{1}})(\sqrt{x_{k}-x_{1}}+\sqrt{x_{j}-x_{1}})}{(\sqrt{x_{k}-x_{1}}-\sqrt{x_{j}-x_{1}})(\sqrt{x_{k}-x_{1}}+\sqrt{x_{j-1}-x_{1}})}.
\end{equation}
From the above expansion, we obtain, as $z \to -x_{j}$, $j \in \{2,...,m\}$, $\Im z > 0$ that
\begin{equation}
D(z) = \sqrt{s_{j}}\Big( \prod_{k=2}^{m} T_{k,j}^{\frac{\log s_{k}}{2\pi i}} \Big)D_{\alpha,+}(-x_{j}) (z+x_{j})^{\beta_{j}}(1+\bigO(z+x_{j})),
\end{equation}
where $\beta_{2},\ldots,\beta_{m}$ are given by
\begin{equation}\label{def of beta_j}
\beta_{j} = \frac{1}{2\pi i}\log \frac{s_{j+1}}{s_{j}}, \quad \mbox{ or equivalently } \quad e^{-2i\pi \beta_{j}} = \frac{s_{j}}{s_{j+1}}, \qquad j = 2,...,m,
\end{equation}
with $s_{m+1}:=1$. Note that
\begin{equation}
\prod_{k=2}^{m} T_{k,j}^{\frac{\log s_{k}}{2\pi i}} = (4(x_{j}-x_{1}))^{-\beta_{j}}\prod_{\substack{k=2 \\ k \neq j}}^{m} \widetilde{T}_{k,j}^{-\beta_{k}}, \quad \mbox{ where } \quad \widetilde{T}_{k,j} = \frac{\sqrt{x_{j}-x_{1}}+\sqrt{x_{k}-x_{1}}}{|\sqrt{x_{j}-x_{1}}-\sqrt{x_{k}-x_{1}}|}.
\end{equation}
As $z \to -x_{1}$, we have
\begin{equation}
D_{\vec{s}}(z) = \sqrt{s_{2}}\Big(1-d_{0,\vec{s}}\sqrt{z+x_{1}}+\bigO(z+x_{1})\Big),
\end{equation}
where
\begin{equation}
d_{0,\vec{s}} = \frac{\log s_{2}}{\pi \sqrt{x_{2}-x_{1}}} - \sum_{j=3}^{m} \frac{\log s_{j}}{\pi}\Big( \frac{1}{\sqrt{x_{j-1}-x_{1}}}-\frac{1}{\sqrt{x_{j}-x_{1}}} \Big).
\end{equation}
As $z \to -x_{1}$, $\Im z > 0$, we have
\begin{equation}
D_{\alpha}(z) = e^{-\frac{\pi i \alpha}{2}}\Big(1-d_{0,\alpha}\sqrt{z+x_{1}}+\bigO(z+x_{1})\Big), \qquad d_{0,\alpha} = \frac{-\alpha}{\sqrt{x_{1}}}.
\end{equation}
It follows that as $z \to -x_{1}$, $\Im z > 0$, we have
\begin{equation}
D(z) = \sqrt{s_{2}}e^{-\frac{\pi i \alpha}{2}}\Big( 1-d_{0} \sqrt{z+x_{1}}+\bigO(z+x_{1}) \Big), \qquad d_{0} := d_{0,\alpha} + d_{0,\vec{s}}.
\end{equation}
Note that for all $\ell \in \{0,1,2,...\}$, we can rewrite $d_{\ell}$ in terms of the $\beta_{j}$'s as follows
\begin{equation}\label{d_ell in terms of beta_j}
d_{\ell} = \frac{\alpha \, x_{1}^{\ell - \frac{1}{2}}}{2\ell-1}+ \frac{2i(-1)^{\ell}}{2\ell-1}\sum_{j=2}^{m} \beta_{j} (x_{j}-x_{1})^{\ell-\frac{1}{2}}.
\end{equation}
As $z \to 0$, we have
\begin{equation}\label{asymp D at 0 s1 = 0}
D(z) = D_{0}z^{-\frac{\alpha}{2}}(1+\bigO(z)),
\end{equation}
where 
\begin{align}\label{D0}
D_{0} = \exp \bigg( \frac{\alpha}{2}\log(4x_{1}) -\sum_{j=2}^{m} 2 i \beta_{j} \arccos \bigg(  \frac{\sqrt{x_{1}}}{\sqrt{x_{j}}} \bigg) \bigg).
\end{align}

\subsection{Local parametrices}
In this section, we aim to find approximations for $S$ in small neighborhoods of $-x_{1}$,...,$-x_{m}$ (as already mentioned, there is no need for a local parametrix in a neighborhood of $0$). By \eqref{condition on xj in terms of delta s1 = 0}, there exist small disks $\mathcal{D}_{-x_{j}}$ centred at  $-x_{j}$, $j=1,...,m$, whose radii are fixed (independent of $r$), but sufficiently small such that they do not intersect. The local parametrix around $-x_{j}$, $j \in \{1,...,m\}$, is defined in $\mathcal{D}_{-x_{j}}$ and is denoted by $P^{(-x_{j})}$. It satisfies an RH problem with the same jumps as $S$ (inside $\mathcal{D}_{-x_{j}}$) and in addition we require
\begin{equation}\label{match at the center s1 = 0}
S(z) P^{(-x_{j})}(z)^{-1} = \bigO(1), \qquad \mbox{ as } z \to -x_{j},
\end{equation}
and
\begin{equation}\label{matching weak s1 = 0}
P^{(-x_{j})}(z) = (I+o(1))P^{(\infty)}(z), \qquad \mbox{ as } r \to +\infty,
\end{equation}
uniformly for $z \in \partial \mathcal{D}_{-x_{j}}$.

\subsubsection{Local parametrices around $-x_{j}$, $j = 2,...,m$}
For $j \in \{2,...,m\}$, $P^{(-x_{j})}$ can be explicitly expressed in terms of the model RH problem for $\Phi_{\mathrm{HG}}$ (see Section \ref{subsection: model RHP with HG functions}). This construction is very similar to the one done in Section \ref{subsection: local param HG s1 neq 0}, and we provide less details here. Let us first consider the function
\begin{equation}\label{conformal xj s1 = 0}
f_{-x_{j}}(z) = -2 \left\{ \begin{array}{l l}
g(z)-g_{+}(-x_{j}), & \mbox{if } \Im z > 0 \\
-(g(z)-g_{-}(-x_{j})), & \mbox{if } \Im z < 0
\end{array} \right. = -2i\big(\sqrt{-z-x_{1}}-\sqrt{x_{j}-x_{1}}\big).
\end{equation}
This is a conformal map from $\mathcal{D}_{-x_{j}}$ to a neighborhood of $0$, and its expansion as $z \to -x_{j}$ is given by
\begin{equation}\label{expansion conformal map}
f_{-x_{j}}(z) = i c_{-x_{j}} (z+x_{j})(1+\bigO(z+x_{j})) \quad \mbox{ with } \quad c_{-x_{j}} = \frac{1}{\sqrt{x_{j}-x_{1}}} > 0.
\end{equation}
Note also that $f_{-x_{j}}(\mathbb{R}\cap \mathcal{D}_{-x_{j}})\subset i \mathbb{R}$. Now, we deform the lenses in a similar way as in \eqref{deformation of the lenses local param xj s1 neq 0}, that is, such that $f_{-x_{j}}$ maps the jump contour for $P^{(-x_{j})}$ onto a subset of $\Sigma_{\mathrm{HG}}$ (see Figure \ref{Fig:HG}). It can be checked that the local parametrix is given by
\begin{equation}\label{lol10 s1 = 0}
P^{(-x_{j})}(z) = E_{-x_{j}}(z) \Phi_{\mathrm{HG}}(\sqrt{r}f_{-x_{j}}(z);\beta_{j})(s_{j}s_{j+1})^{-\frac{\sigma_{3}}{4}}e^{-\sqrt{r}g(z)\sigma_{3}}e^{\frac{\pi i \alpha}{2}\theta(z)\sigma_{3}},
\end{equation}
where $E_{-x_{j}}$ is analytic inside $\mathcal{D}_{-x_{j}}$ and given by
\begin{multline}
E_{-x_{j}}(z) = P^{(\infty)}(z) e^{-\frac{\pi i \alpha}{2}\theta(z)\sigma_{3}} (s_{j} s_{j+1})^{\frac{\sigma_{3}}{4}} \left\{ \begin{array}{l l}
\ds \sqrt{\frac{s_{j}}{s_{j+1}}}^{\sigma_{3}}, & \Im z > 0 \\
\begin{pmatrix}
0 & 1 \\ -1 & 0
\end{pmatrix}, & \Im z < 0
\end{array} \right\} \times \\ e^{\sqrt{r}g_{+}(-x_{j})\sigma_{3}}(\sqrt{r}f_{-x_{j}}(z))^{\beta_{j}\sigma_{3}}.
\end{multline}
We will need later a more detailed knowledge than \eqref{matching weak s1 = 0}. Using \eqref{Asymptotics HG}, one shows that
\begin{equation}\label{matching strong -x_j}
P^{(-x_{j})}(z)P^{(\infty)}(z)^{-1} = I + \frac{1}{\sqrt{r}f_{-x_{j}}(z)}E_{-x_{j}}(z) \Phi_{\mathrm{HG},1}(\beta_{j})E_{-x_{j}}(z)^{-1} + \bigO(r^{-1}),
\end{equation}
as $r \to + \infty$, uniformly for $z \in \partial \mathcal{D}_{-x_{j}}$, where $\Phi_{\mathrm{HG},1}(\beta_{j})$ is given by \eqref{def of tau} with the parameter $\beta_{j}$ given by \eqref{def of beta_j}. Also, a direct computation shows that
\begin{equation}\label{E_j at -x_j}
E_{-x_{j}}(-x_{j}) = \begin{pmatrix}
1 & 0 \\ id_{1} & 1
\end{pmatrix} e^{-\frac{\pi i}{4}\sigma_{3}} (x_{j}-x_{1})^{-\frac{\sigma_{3}}{4}}N\Lambda_{j}^{\sigma_{3}},
\end{equation}
where
\begin{equation}\label{def Lambda_j}
\Lambda_{j} = D_{\alpha,+}(-x_{j})^{-1}e^{-\frac{\pi i \alpha}{2}} (4(x_{j}-x_{1}))^{\beta_{j}} \bigg( \prod_{\substack{k=2 \\ k \neq j}}^{m} \widetilde{T}_{k,j}^{\beta_{k}} \bigg)e^{\sqrt{r}g_{+}(-x_{j})}r^{\frac{\beta_{j}}{2}}c_{-x_{j}}^{\beta_{j}}.
\end{equation}
\subsubsection{Local parametrix around $-x_{1}$}
The local parametrix $P^{(-x_{1})}$ can be expressed in terms of the model RH problem $\Phi_{\mathrm{Be}}(z;0)$ presented in Section \ref{subsection:Model Bessel}. This construction is similar to the one done in Section \ref{subsection: local param bessel s1 neq 0} (note however that in Section \ref{subsection: local param bessel s1 neq 0} we needed $\Phi_{\mathrm{Be}}(z;\alpha)$), and we provide less details here. Let us first consider the function
\begin{equation}\label{conformal map near -x1 s1 = 0}
f_{-x_{1}}(z) = \frac{g(z)^{2}}{4} = \frac{z+x_{1}}{4}.
\end{equation}
This is a conformal map from $\mathcal{D}_{-x_{1}}$ to a neighborhood of $0$. Similarly to \eqref{deformation of the lenses local param 0 s1 neq 0}, we choose $\gamma_{2,\pm}$ such that the jump contour for $P^{(-x_{1})}$ is mapped by $f_{-x_{1}}$ onto a subset of $\Sigma_{\mathrm{Be}}$ (see Figure \ref{figBessel}). It can be verified that $P^{(-x_{1})}$ is given by
\begin{equation}\label{def of P^-x1}
P^{(-x_{1})}(z) = E_{-x_{1}}(z)\Phi_{\mathrm{Be}}(rf_{-x_{1}}(z);0)s_{2}^{-\frac{\sigma_{3}}{2}}e^{-\sqrt{r}g(z)\sigma_{3}}e^{\frac{\pi i \alpha}{2}\theta(z)\sigma_{3}},
\end{equation}
where $E_{-x_{1}}$ is analytic inside $\mathcal{D}_{-x_{1}}$ and is given by
\begin{equation}
E_{-x_{1}}(z) = P^{(\infty)}(z)e^{-\frac{\pi i \alpha}{2}\theta(z)\sigma_{3}}s_{2}^{\frac{\sigma_{3}}{2}}N^{-1}\left( 2\pi \sqrt{r}f_{-x_{1}}(z)^{1/2} \right)^{\frac{\sigma_{3}}{2}}.
\end{equation}
We will need later a more detailed knowledge than \eqref{matching weak s1 = 0}. Using \eqref{large z asymptotics Bessel}, one shows that
\begin{multline}\label{matching strong 0}
P^{(-x_{1})}(z)P^{(\infty)}(z)^{-1} = I \\ + \frac{1}{\sqrt{r}f_{-x_{1}}(z)^{1/2}}P^{(\infty)}(z)e^{-\frac{\pi i \alpha}{2}\theta(z)\sigma_{3}}s_{2}^{\frac{\sigma_{3}}{2}}\Phi_{\mathrm{Be},1}(0)s_{2}^{-\frac{\sigma_{3}}{2}}e^{\frac{\pi i \alpha}{2}\theta(z)\sigma_{3}}P^{(\infty)}(z)^{-1} + \bigO(r^{-1}),
\end{multline}
as $r \to +\infty$ uniformly for $z \in \partial \mathcal{D}_{-x_{1}}$, where $\Phi_{\mathrm{Be},1}(0)$ is given below \eqref{large z asymptotics Bessel}. Furthermore, 
\begin{equation}\label{E -x1 at -x1 s1 = 0}
E_{-x_{1}}(-x_{1}) = \begin{pmatrix}
1 & 0 \\ id_{1} & 1
\end{pmatrix} \begin{pmatrix}
1 & -id_{0} \\
0 & 1
\end{pmatrix} (\pi \sqrt{r})^{\frac{\sigma_{3}}{2}}.
\end{equation}

\subsection{Small norm problem}\label{subsection Small norm s1 = 0}

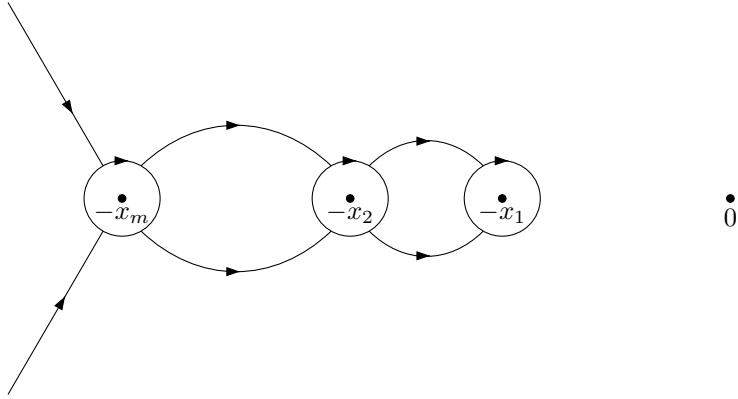
\begin{figure}
\centering
\begin{tikzpicture}
\draw[fill] (0,0) circle (0.05);
\draw (0,0) circle (0.5);

\draw (120:0.5) -- (120:3);
\draw (-120:0.5) -- (-120:3);

\draw ($(0,0)+(60:0.5)$) .. controls (1,1.15) and (2,1.15) .. ($(3,0)+(120:0.5)$);
\draw ($(0,0)+(-60:0.5)$) .. controls (1,-1.15) and (2,-1.15) .. ($(3,0)+(-120:0.5)$);
\draw ($(3,0)+(60:0.5)$) .. controls (3.7,0.88) and (4.3,0.88) .. ($(5,0)+(120:0.5)$);
\draw ($(3,0)+(-60:0.5)$) .. controls (3.7,-0.88) and (4.3,-0.88) .. ($(5,0)+(-120:0.5)$);

\draw[fill] (3,0) circle (0.05);
\draw (3,0) circle (0.5);
\draw[fill] (5,0) circle (0.05);
\draw (5,0) circle (0.5);
\draw[fill] (8,0) circle (0.05);

\node at (0.,-0.2) {$-x_{m}$};
\node at (3,-0.2) {$-x_{2}$};
\node at (5,-0.2) {$-x_{1}$};
\node at (8,-0.25) {$0$};

\draw[black,arrows={-Triangle[length=0.18cm,width=0.12cm]}]
(-120:1.5) --  ++(60:0.001);
\draw[black,arrows={-Triangle[length=0.18cm,width=0.12cm]}]
(120:1.3) --  ++(-60:0.001);
\draw[black,arrows={-Triangle[length=0.18cm,width=0.12cm]}]
($(0.08,0)+(90:0.5)$) --  ++(0:0.001);

\draw[black,arrows={-Triangle[length=0.18cm,width=0.12cm]}]
($(3.08,0)+(90:0.5)$) --  ++(0:0.001);
\draw[black,arrows={-Triangle[length=0.18cm,width=0.12cm]}]
($(5.08,0)+(90:0.5)$) --  ++(0:0.001);

\draw[black,arrows={-Triangle[length=0.18cm,width=0.12cm]}]
(1.55,0.97) --  ++(0:0.001);
\draw[black,arrows={-Triangle[length=0.18cm,width=0.12cm]}]
(1.55,-0.97) --  ++(0:0.001);

\draw[black,arrows={-Triangle[length=0.18cm,width=0.12cm]}]
(4.05,0.76) --  ++(0:0.001);
\draw[black,arrows={-Triangle[length=0.18cm,width=0.12cm]}]
(4.05,-0.76) --  ++(0:0.001);


\end{tikzpicture}
\caption{Jump contours $\Sigma_{R}$ for the RH problem for $R$ with $m=3$ and $s_{1} = 0$.}
\label{fig:contour for R s1 = 0}
\end{figure}

The last transformation of the steepest descent is defined by
\begin{equation}\label{def of R}
R(z) = \left\{ \begin{array}{l l}
S(z)P^{(\infty)}(z)^{-1}, & \mbox{for } z \in \mathbb{C}\setminus \bigcup_{j=1}^{m}\mathcal{D}_{-x_{j}}, \\
S(z)P^{(-x_{j})}(z)^{-1}, & \mbox{for } z \in \mathcal{D}_{-x_{j}}, \, j \in \{1,...,m\}.
\end{array} \right.
\end{equation}
The analysis of $R$ is similar to the one done in Section \ref{subsection Small norm s1 neq 0}, and we provide less details here. The main difference lies in the analysis of $R(z)$ for $z$ in a neighborhood of $0$. From the RH problems for $S$ and $P^{(\infty)}$, it is straightforward to verify that $R$ has no jumps along $(-x_{1},0)$ and is bounded as $z \to 0$. Thus $R$ is analytic in a neighborhood of $0$. Also, by definition of the local parametrices, $R$ is analytic on $\mathbb{C}\setminus \Sigma_{R}$, where $\Sigma_{R}$ consists of the boundaries of the disks, and the part of the lenses away from the disks, as shown in Figure \ref{fig:contour for R s1 = 0}. As in Section \ref{subsection Small norm s1 neq 0}, the jumps for $R$ on the lenses are uniformly exponentially close to $I$ as $r \to +\infty$. On the boundary of the disks, the jumps are close to $I$ by an error of order $\bigO(r^{-1/2})$. Therefore, $R$ satisfies a small norm RH problem. By standard theory \cite{DKMVZ2,DKMVZ1} (see also Section \ref{subsection Small norm s1 neq 0}), $R$ exists for sufficiently large $r$ and satisfies
\begin{align}
& R(z) = I + \frac{R^{(1)}(z)}{\sqrt{r}} + \bigO(r^{-1}) & & R^{(1)}(z) = \bigO(1), \label{eq: asymp R inf} \\[0cm]
& \partial_{\beta_{j}}R(z) = \frac{\partial_{\beta_{j}}R^{(1)}(z)}{\sqrt{r}} + \bigO \Big( \frac{\log r}{r} \Big), & & \partial_{\beta_{j}}R^{(1)}(z) = \bigO(\log r) \label{eq: asymp der beta R inf s1 = 0}
\end{align}
as $r \to + \infty$, uniformly for $z \in \mathbb{C}\setminus \Sigma_{R}$, uniformly for $\beta_{2},...,\beta_{m}$ in compact subsets of $i \mathbb{R}$, and uniformly in $x_{1},...,x_{m}$ in compact subsets of $(0,+\infty)$ as long as there exists $\delta > 0$ which satisfies \eqref{condition on xj in terms of delta s1 = 0}.

\vspace{0.2cm}\hspace{-0.55cm}The goal for the rest of this section is to obtain $R^{(1)}(z)$ for $z \in \mathbb{C}\setminus \bigcup_{j=1}^{m}\mathcal{D}_{-x_{j}}$ and for $z = -x_{1}$ explicitly. Let us take the clockwise orientation on the boundaries of the disks, and let us denote by $J_{R}(z)$ for the jumps of $R$. Since $J_{R}$ admits a large $r$ expansion of the form
\begin{equation}
J_{R}(z) = I + \frac{J_{R}^{(1)}(z)}{\sqrt{r}} + \bigO(r^{-1}), 
\end{equation}
as $r \to \infty$ uniformly for $z \in \bigcup_{j=1}^{m}\mathcal{D}_{-x_{j}}$, we obtain (in the same way as in Section \ref{subsection Small norm s1 neq 0}) that $R^{(1)}$ is simply given by
\begin{equation}
R^{(1)}(z) = \frac{1}{2\pi i}\int_{\bigcup_{j=1}^{m}\partial\mathcal{D}_{-x_{j}}} \frac{J_{R}^{(1)}(s)}{s-z}ds.
\end{equation}
By a direct residue calculation we have
\begin{equation}\label{expression for R^1}
R^{(1)}(z) = \sum_{j=1}^{m} \frac{1}{z+x_{j}}\mbox{Res}(J_{R}^{(1)}(s),s = -x_{j}), \qquad \mbox{ for } z \in \mathbb{C}\setminus \bigcup_{j=1}^{m}\mathcal{D}_{-x_{j}}
\end{equation}
and
\begin{equation}\label{expression for R^1 at -x1 for s1 = 0}
R^{(1)}(-x_{1}) = -\mbox{Res}\Big(\frac{J_{R}^{(1)}(s)}{s+x_{1}},s = -x_{1}\Big)+ \sum_{j=2}^{m} \frac{1}{x_{j}-x_{1}}\mbox{Res}(J_{R}^{(1)}(s),s = -x_{j}).
\end{equation}
From \eqref{matching strong 0}, we have
\begin{equation}
\mbox{Res}(J_{R}^{(1)}(s),s = -x_{1}) = \frac{d_{1}}{8}\begin{pmatrix}
-1 & -id_{1}^{-1} \\ -id_{1} & 1
\end{pmatrix},
\end{equation}
and with increasing effort
\begin{equation}
\mbox{Res}\Big(\frac{J_{R}^{(1)}(s)}{s+x_{1}},s = -x_{1}\Big) = \frac{1}{2} \begin{pmatrix}
-(d_{0}+d_{1}d_{0}^{2}) & -id_{0}^{2} \\
-i \big( d_{0}^{2}d_{1}^{2}+2d_{0}d_{1}+\frac{3}{4} \big) & d_{0} + d_{1}d_{0}^{2}
\end{pmatrix}.
\end{equation}
From \eqref{matching strong -x_j}-\eqref{def Lambda_j}, for $j \in \{2,...,m\}$, we have
\begin{multline*}
\mbox{Res}\left( J_{R}^{(1)}(s),s=-x_{j} \right) = \frac{\beta_{j}^{2}}{ic_{-x_{j}}}\begin{pmatrix}
1 & 0 \\ id_{1} & 1
\end{pmatrix}e^{-\frac{\pi i}{4}\sigma_{3}}(x_{j}-x_{1})^{-\frac{\sigma_{3}}{4}} N \begin{pmatrix}
-1 & \widetilde{\Lambda}_{j,1} \\ -\widetilde{\Lambda}_{j,2} & 1
\end{pmatrix} \\ \times N^{-1} (x_{j}-x_{1})^{\frac{\sigma_{3}}{4}}e^{\frac{\pi i}{4}\sigma_{3}}\begin{pmatrix}
1 & 0 \\ -id_{1} & 1
\end{pmatrix},
\end{multline*}
where
\begin{equation}
\widetilde{\Lambda}_{j,1} = \tau(\beta_{j})\Lambda_{j}^{2} \qquad \mbox{ and } \qquad \widetilde{\Lambda}_{j,2} = \tau(-\beta_{j})\Lambda_{j}^{-2}.
\end{equation}

\section{Proof of Theorem \ref{thm:s1=0}}\label{Section: integration s1 =0}
This section is divided into two parts in the same way as in Section \ref{Section: integration s1 >0}. In the first part, using the RH analysis done in Section \ref{Section: Steepest descent with s1=0}, we find large $r$ asymptotics for the differential identity
\begin{equation}
\partial_{s_{k}} \log F_{\alpha}(r\vec{x},\vec{s})= K_{\infty} + \sum_{j=1}^{m}K_{-x_{j}} + K_{0},
\end{equation}
which was obtained in \eqref{DIFF identity final form general case} with the quantities $K_{\infty}$, $K_{-x_{j}}$ and $K_{0}$ defined in \eqref{K inf}-\eqref{K 0}. In the second part, we integrate these asymptotics over the parameters $s_{2},...,s_{m}$. Some parts of the computations in this section are close to those done in Section \ref{Section: integration s1 >0}. However, it requires some adaptation and we provide the details for completeness.

\subsection{Large $r$ asymptotics for the differential identity}

\paragraph{Asymptotics for $K_{\infty}$.}

For $z$ outside the disks and outside the lenses, by \eqref{def of R} we have
\begin{equation}
S(z) = R(z) P^{(\infty)}(z).
\end{equation}
As $z \to \infty$, we can write
\begin{equation}
R(z) = I + \frac{R_{1}}{z} + \bigO(z^{-2}),
\end{equation}
for a certain matrix $R_{1}$ independent of $z$. Thus, by \eqref{eq:Sasympinf} and \eqref{eq:Pinf asympinf}, we have
\begin{align*}
& T_{1} = R_{1} + P_{1}^{(\infty)}.
\end{align*}
Using \eqref{eq: asymp R inf} and the above expressions, as $r \to +\infty$ we have
\begin{align*}
& T_{1} = P_{1}^{(\infty)} + \frac{R_{1}^{(1)}}{\sqrt{r}} + \bigO(r^{-1}),
\end{align*}
where $R_{1}^{(1)}$ is defined through the expansion
\begin{equation}
R^{(1)}(z) = \frac{R_{1}^{(1)}}{z} + \bigO(z^{-2}), \qquad \mbox{ as } z \to \infty.
\end{equation}
By \eqref{K inf}, \eqref{eq:Tasympinf}, \eqref{Pinf 1 12 s1 = 0}, \eqref{eq: asymp der beta R inf s1 = 0} and \eqref{expression for R^1}, the large $r$ asymptotics for $K_{\infty}$ are given by
\begin{multline}\label{K inf asymp s1 = 0}
K_{\infty} = - \frac{i}{2}\sqrt{r}\partial_{s_{k}}T_{1,12} = -\frac{i}{2}\Big( \partial_{s_{k}} P_{1,12}^{(\infty)}\sqrt{r} + \partial_{s_{k}}R_{1,12}^{(1)} + \bigO \Big( \frac{\log r}{\sqrt{r}} \Big) \Big) \\ = \frac{1}{2}\partial_{s_{k}}d_{1} \sqrt{r} - \sum_{j=2}^{m} \frac{ \partial_{s_{k}}\big( \beta_{j}^{2}(\widetilde{\Lambda}_{j,1}-\widetilde{\Lambda}_{j,2}+2i) \big)}{4i c_{-x_{j}}\sqrt{x_{j}-x_{1}}}+ \bigO \Big( \frac{\log r}{\sqrt{r}} \Big).
\end{multline}
\paragraph{Asymptotics for $K_{-x_{j}}$ with $j \in \{2,...,m\}$.} By inverting the transformations \eqref{def of S} and \eqref{def of R}, and using the expression for $P^{(-x_{j})}$ given by  \eqref{lol10 s1 = 0}, for $z$ outside the lenses and inside $\mathcal{D}_{-x_{j}}$, we have
\begin{equation}\label{lol11 s1 = 0}
T(z) = R(z)E_{-x_{j}}(z)\Phi_{\mathrm{HG}}(\sqrt{r}f_{-x_{j}}(z);\beta_{j})(s_{j}s_{j+1})^{-\frac{\sigma_{3}}{4}}e^{\frac{\pi i \alpha}{2}\theta(z) \sigma_{3}}e^{-\sqrt{r}g(z)\sigma_{3}}.
\end{equation}
If furthermore $\Im z > 0$, then by \eqref{conformal xj s1 = 0} and \eqref{model RHP HG in different sector} we have
\begin{equation}
\Phi_{\mathrm{HG}}(\sqrt{r}f_{-x_{j}}(z);\beta_{j}) = \widehat{\Phi}_{\mathrm{HG}}(\sqrt{r}f_{-x_{j}}(z);\beta_{j}).
\end{equation}
Note from \eqref{def of beta_j} and the connection formula for the $\Gamma$-function that
\begin{equation}\label{relation Gamma beta_j and s_j s1 = 0}
\frac{\sin (\pi \beta_{j})}{\pi} = \frac{1}{\Gamma(\beta_{j})\Gamma(1-\beta_{j})} = \frac{s_{j+1}-s_{j}}{2\pi i \sqrt{s_{j}s_{j+1}}}.
\end{equation}
Therefore, using \eqref{conformal xj s1 = 0} and \eqref{precise asymptotics of Phi HG near 0}, as $z \to -x_{j}$ from the upper half plane and outside the lenses, we have
\begin{equation}\label{lol 2 s1 = 0}
\Phi_{\mathrm{HG}}(\sqrt{r}f_{-x_{j}}(z);\beta_{j})(s_{j}s_{j+1})^{-\frac{\sigma_{3}}{4}} = \begin{pmatrix}
\Psi_{j,11} & \Psi_{j,12} \\
\Psi_{j,21} & \Psi_{j,22}
\end{pmatrix} (I + \bigO(z+x_{j})) \begin{pmatrix}
1 & \frac{s_{j+1}-s_{j}}{2\pi i}\log(r(z + x_{j})) \\
0 & 1
\end{pmatrix} ,
\end{equation}
where the principal branch is taken for the log and
\begin{align}
& \Psi_{j,11} = \frac{\Gamma(1-\beta_{j})}{(s_{j}s_{j+1})^{\frac{1}{4}}}, \qquad \Psi_{j,12} = \frac{(s_{j}s_{j+1})^{\frac{1}{4}}}{\Gamma(\beta_{j})} \left( \log(c_{-x_{j}}r^{-1/2}) - \frac{i\pi}{2} + \frac{\Gamma^{\prime}(1-\beta_{j})}{\Gamma(1-\beta_{j})}+2\gamma_{\mathrm{E}} \right), \nonumber \\
& \Psi_{j,21} = \frac{\Gamma(1+\beta_{j})}{(s_{j}s_{j+1})^{\frac{1}{4}}}, \qquad \Psi_{j,22} = \frac{-(s_{j}s_{j+1})^{\frac{1}{4}}}{\Gamma(-\beta_{j})} \left( \log(c_{-x_{j}}r^{-1/2}) - \frac{i\pi}{2} + \frac{\Gamma^{\prime}(-\beta_{j})}{\Gamma(-\beta_{j})} + 2\gamma_{\mathrm{E}} \right). \label{Psi j entries s1 = 0}
\end{align}
From \eqref{def of Gj}, \eqref{def of T}, \eqref{lol11 s1 = 0} and \eqref{lol 2 s1 = 0} we have
\begin{equation}\label{lol12 s1 = 0}
G_{j}(-rx_{j};r\vec{x},\vec{s}) = r^{-\frac{\sigma_{3}}{4}}\begin{pmatrix}
1 & 0 \\ - i\frac{x_{1}}{2}\sqrt{r} & 1
\end{pmatrix}R(-x_{j})E_{-x_{j}}(-x_{j})\begin{pmatrix}
\Psi_{j,11} & \Psi_{j,12} \\ \Psi_{j,21} & \Psi_{j,22}
\end{pmatrix}.
\end{equation}
In fact $K_{-x_{j}}$ does not depend on the first two pre-factors in \eqref{lol12 s1 = 0}. Let us define 
\begin{equation}\label{def of Hj s1 = 0}
H_{j} = \begin{pmatrix}
1 & 0 \\ i\frac{x_{1}}{2}\sqrt{r} & 1
\end{pmatrix}r^{\frac{\sigma_{3}}{4}}G_{j}(-rx_{j};r\vec{x},\vec{s}) = R(-x_{j})E_{-x_{j}}(-x_{j})\begin{pmatrix}
\Psi_{j,11} & \Psi_{j,12} \\ \Psi_{j,21} & \Psi_{j,22}
\end{pmatrix}.
\end{equation}
By a straightforward computation, we rewrite \eqref{K -xj} as
\begin{equation}\label{K xj in terms of Hj s1 = 0}
\sum_{j=2}^{m} K_{-x_{j}} = \sum_{j=2}^m\frac{s_{j+1}-s_j}{2\pi i} (H_{j,11} \partial_{s_{k}}H_{j,21}-H_{j,21}\partial_{s_{k}}H_{j,11}).
\end{equation}
Using the connection formula $\Gamma(z)\Gamma(1-z) = \frac{\pi}{\sin \pi z}$, we note that
\begin{equation}\label{Psi_j first column connection formula s1 = 0}
\Psi_{j,11}\Psi_{j,21} = \beta_{j} \frac{2\pi i}{s_{j+1}-s_{j}}, \qquad j =2,...,m.
\end{equation}
Also, from \eqref{E_j at -x_j}, $E_{-x_{j}}$ satisfies \eqref{partial E}. Therefore, using \eqref{eq: asymp R inf}, \eqref{eq: asymp der beta R inf s1 = 0}, $\det E_{-x_{j}}(-x_{j})=1$, \eqref{def of Hj s1 = 0}-\eqref{Psi_j first column connection formula s1 = 0} and \eqref{partial E}, as $r \to + \infty$ we obtain
\begin{multline}\label{K xj part 1 asymp s1 = 0}
\sum_{j=2}^m K_{-x_{j}} = \sum_{j=2}^{m} \frac{s_{j+1}-s_j}{2\pi i} \Big( \Psi_{j,11} \partial_{s_{k}}\Psi_{j,21} - \Psi_{j,21}\partial_{s_{k}}\Psi_{j,11} \Big) - \sum_{j=2}^{m} 2\beta_{j} \partial_{s_{k}} \log \Lambda_{j} \\  + i \partial_{s_{k}}d_{1} \sum_{j=2}^{m}\frac{s_{j+1}-s_j}{2\pi i}(E_{-x_{j},11}(-x_{j})\Psi_{j,11}+E_{-x_{j},12}(-x_{j})\Psi_{j,21})^{2} + \bigO \bigg( \frac{\log r}{\sqrt{r}} \bigg),
\end{multline}
where, by \eqref{E_j at -x_j} and \eqref{Psi j entries s1 = 0}, one has
\begin{multline}\label{K xj part 2 asymp s1 = 0}
i \partial_{s_{k}}d_{1} \sum_{j=2}^{m}\frac{s_{j+1}-s_j}{2\pi i}(E_{-x_{j},11}(-x_{j})\Psi_{j,11}+E_{-x_{j},12}(-x_{j})\Psi_{j,21})^{2} = \sum_{j=2}^{m} \frac{\partial_{s_{k}}d_{1}}{2\sqrt{x_{j}-x_{1}}}\Big( \beta_{j}^{2}(\widetilde{\Lambda}_{j,1}+\widetilde{\Lambda}_{j,2})+2i\beta_{j} \Big).
\end{multline}
\paragraph{Asymptotics for $K_{-x_{1}}$.} Note that we did not use the explicit expression for $R^{(1)}(-x_{j})$ to compute the asymptotics for $K_{-x_{j}}$ up to and including the constant term for $j = 2,...,m$. The computations for $K_{-x_{1}}$ are more involved and require explicitly $R^{(1)}(-x_{1})$ (given by \eqref{expression for R^1 at -x1 for s1 = 0}). We start by evaluating $G_{1}(-rx_{1};r\vec{x},\vec{s})$. For $z$ outside the lenses and inside $\mathcal{D}_{-x_{1}}$, by \eqref{def of S}, \eqref{def of P^-x1} and \eqref{def of R}, we have that
\begin{equation}\label{lol14 s1 = 0}
T(z) =
R(z)E_{-x_{1}}(z)\Phi_{\mathrm{Be}}(rf_{-x_{1}}(z);0)s_{2}^{-\frac{\sigma_{3}}{2}}e^{\frac{\pi i \alpha}{2}\theta(z) \sigma_{3}}e^{-\sqrt{r}g(z)\sigma_{3}}.
\end{equation}
From \eqref{conformal map near -x1 s1 = 0}, \eqref{lol14 s1 = 0} and \eqref{precise asymptotics of Phi Bessel near 0}, as $z \to -x_{1}$ from outside the lenses, we have
\begin{equation}
T(z) = R(z)E_{-x_{1}}(z)\Phi_{\mathrm{Be},0}(rf_{-x_{1}}(z);0)s_{2}^{-\frac{\sigma_{3}}{2}}\begin{pmatrix}
1 & \frac{s_{2}}{2\pi i}\log \frac{r(z+x_{1})}{4} \\ 0 & 1
\end{pmatrix}e^{\frac{\pi i \alpha}{2}\theta(z)\sigma_{3}}e^{-\sqrt{r}g(z)\sigma_{3}}.
\end{equation}
On the other hand, using \eqref{def of Gj} and  \eqref{def of T}, as $z \to -x_{1}$, $\Im z > 0$, we have
\begin{equation}
T(z) = \begin{pmatrix}
1 & 0 \\ i \frac{x_{1}}{2}\sqrt{r} & 1
\end{pmatrix} r^{\frac{\sigma_{3}}{4}}G_{1}(rz;r\vec{x},\vec{s})
\begin{pmatrix}
1 & \ds \frac{s_{2}}{2\pi i} \log (r(z+x_{1})) \\ 0 & 1
\end{pmatrix}e^{\frac{\pi i \alpha}{2}\sigma_{3}}e^{-\sqrt{r}g(z)\sigma_{3}}.
\end{equation}
Therefore, using also \eqref{precise matrix at 0 in asymptotics of Phi Bessel near 0}, we obtain
\begin{equation}\label{lol15 s1 = 0}
G_{1}(-rx_{1};r\vec{x},\vec{s}) = r^{-\frac{\sigma_{3}}{4}}\begin{pmatrix}
1 & 0 \\ -i\frac{x_{1}}{2}\sqrt{r} & 1
\end{pmatrix}R(-x_{1})E_{-x_{1}}(-x_{1})\begin{pmatrix}
\Psi_{1,11} & \Psi_{1,12} \\ \Psi_{1,21} & \Psi_{1,22}
\end{pmatrix},
\end{equation}
where
\begin{align*}
& \Psi_{1,11} = s_{2}^{-1/2}, \qquad \Psi_{1,12} = s_{2}^{1/2} \frac{\gamma_{\mathrm{E}} - \log 2}{\pi i}, \\
& \Psi_{1,21} = 0, \hspace{1.51cm} \Psi_{1,22} = s_{2}^{1/2}.
\end{align*}
In the same way as for $K_{-x_{j}}$ with $j = 2,...,m$, we define
\begin{equation}\label{def of H1 s1 neq 0}
H_{1} = \begin{pmatrix}
1 & 0 \\
i \frac{x_{1}}{2} \sqrt{r} & 1
\end{pmatrix} r^{\frac{\sigma_{3}}{4}}G_{1}(-rx_{1};r\vec{x},\vec{s}) = R(-x_{1})E_{-x_{1}}(-x_{1}) \begin{pmatrix}
\Psi_{1,11} & \Psi_{1,12} \\ \Psi_{1,21} & \Psi_{1,22}
\end{pmatrix},
\end{equation}
and we simplify $K_{-x_{1}}$ (given by \eqref{K -xj}) as follows
\begin{equation}\label{simplification K -x1 in terms of H-x1 s1 = 0}
K_{-x_{1}} = \frac{s_{2}}{2\pi i} (H_{1,11} \partial_{s_{k}}H_{1,21}-H_{1,21}\partial_{s_{k}}H_{1,11}).
\end{equation}
Using \eqref{E -x1 at -x1 s1 = 0}, \eqref{eq: asymp R inf}-\eqref{eq: asymp der beta R inf s1 = 0}, \eqref{lol15 s1 = 0}-\eqref{simplification K -x1 in terms of H-x1 s1 = 0}, and the fact that $R^{(1)}$ is traceless, after a careful calculation we obtain 
\begin{equation}\label{K x1 part 1 asymp s1 = 0}
K_{-x_{1}} = \frac{1}{2} \partial_{s_{k}}d_{1} \sqrt{r} -\frac{1}{2} \Big( d_{1} \partial_{s_{k}} (R_{11}^{(1)}(-x_{1})-R_{22}^{(1)}(-x_{1})) +id_{1}^{2} \partial_{s_{k}} R_{12}^{(1)}(-x_{1})+i \partial_{s_{k}} R_{21}^{(1)}(-x_{1}) \Big) + \bigO\bigg( \frac{\log r}{\sqrt{r}} \bigg)
\end{equation}
as $r \to + \infty$. The subleading term in \eqref{K x1 part 1 asymp s1 = 0} can be computed more explicit using the expression for $R^{(1)}(-x_{1})$ given by \eqref{expression for R^1 at -x1 for s1 = 0}:
\begin{multline}\label{K x1 part 2 asymp s1 = 0}
-\frac{1}{2} \Big( d_{1} \partial_{s_{k}} (R_{11}^{(1)}(-x_{1})-R_{22}^{(1)}(-x_{1})) +id_{1}^{2} \partial_{s_{k}} R_{12}^{(1)}(-x_{1})+i \partial_{s_{k}} R_{21}^{(1)}(-x_{1}) \Big) = \frac{d_{0}\partial_{s_{k}}d_{1}}{2}\\+ \sum_{j=2}^{m} \frac{1}{4 i c_{-x_{j}}\sqrt{x_{j}-x_{1}}}\partial_{s_{k}}\big( \beta_{j}^{2}(\widetilde{\Lambda}_{j,1}-\widetilde{\Lambda}_{j,2}-2i)\big)-\partial_{s_{k}}d_{1}\sum_{j=2}^{m} \frac{\beta_{j}^{2}(\widetilde{\Lambda}_{j,1}+\widetilde{\Lambda}_{j,2})}{2c_{-x_{j}}(x_{j}-x_{1})}.
\end{multline}
\paragraph{Asymptotics for $K_{0}$.} From \eqref{def of S}, \eqref{def of Pinf} and \eqref{def of R}, for $z$ in a neighborhood of $0$, we have
\begin{equation}\label{lol16 s1 = 0}
T(z) = R(z) \begin{pmatrix}
1 & 0 \\ id_{1} & 1
\end{pmatrix}(z+x_{1})^{-\frac{\sigma_{3}}{4}}ND(z)^{-\sigma_{3}}.
\end{equation}
On the other hand, from \eqref{def of G_0} and \eqref{def of T}, as $z \to 0$ we have
\begin{equation}\label{lol17 s1 = 0}
T(z) = \begin{pmatrix}
1 & 0 \\ i \frac{x_{1}}{2}\sqrt{r} & 1
\end{pmatrix}r^{\frac{\sigma_{3}}{4}} G_{0}(rz;r\vec{x},\vec{s})(rz)^{\frac{\alpha}{2}\sigma_{3}}e^{-\sqrt{r}g(z)\sigma_{3}}.
\end{equation}
Therefore, using \eqref{asymp D at 0 s1 = 0} and \eqref{lol16 s1 = 0}-\eqref{lol17 s1 = 0}, we obtain
\begin{equation}\label{lol18 s1 = 0}
G_{0}(0;r\vec{x},\vec{s}) = r^{-\frac{\sigma_{3}}{4}}\begin{pmatrix}
1 & 0 \\ - i\frac{x_{1}}{2}\sqrt{r} & 1
\end{pmatrix} R(0) \begin{pmatrix}
1 & 0 \\ id_{1} & 1
\end{pmatrix} x_{1}^{-\frac{\sigma_{3}}{4}}N \widehat{c}^{-\sigma_{3}},
\end{equation}
for a certain $\widehat{c} \in \mathbb{C}$ whose exact value is unimportant for us. Let us define
\begin{equation}\label{lol19 s1 = 0}
H_{0} = \begin{pmatrix}
1 & 0 \\ i\frac{x_{1}}{2}\sqrt{r} & 1
\end{pmatrix}r^{\frac{\sigma_{3}}{4}}G_{0}(0;r\vec{x},\vec{s})\widehat{c}^{\sigma_{3}} = R(0)\begin{pmatrix}
1 & 0 \\ id_{1} & 1
\end{pmatrix}x_{1}^{-\frac{\sigma_{3}}{4}}N.
\end{equation}
Since $s_{1} = 0$, by \eqref{K 0}, we have $K_{0} = 0$ if $\alpha = 0$. If $\alpha \neq 0$, by \eqref{K 0}, \eqref{eq: asymp R inf}-\eqref{eq: asymp der beta R inf s1 = 0} and \eqref{lol18 s1 = 0}-\eqref{lol19 s1 = 0}, we have
\begin{align}\label{K 0 asymp s1 = 0}
K_{0} & = \alpha \Big( H_{0,21} \partial_{s_{k}}H_{0,12} - H_{0,11} \partial_{s_{k}}H_{0,22} + (H_{0,12}H_{0,21}-H_{0,11}H_{0,22}) \partial_{s_{k}} \log D_{0} \Big) \nonumber \\
& = \frac{\alpha \partial_{s_{k}}d_{1}}{2 \sqrt{x_{1}}} - \alpha \partial_{s_{k}} \log D_{0} + \bigO \bigg( \frac{\log r}{\sqrt{r}} \bigg), \qquad \mbox{ as } r \to + \infty.
\end{align}

\paragraph{Asymptotics for the differential identity  \eqref{DIFF identity final form general case}.} Summing the contribution $K_{0}$, $K_{-x_{j}}$, $j=1,...,m$ and $K_{\infty}$  using \eqref{K inf asymp s1 = 0}, \eqref{K xj part 1 asymp s1 = 0}, \eqref{K xj part 2 asymp s1 = 0}, \eqref{K x1 part 1 asymp s1 = 0}, \eqref{K x1 part 2 asymp s1 = 0} and \eqref{K 0 asymp s1 = 0}, and substituting the expression for $c_{-x_{j}}$ given by \eqref{expansion conformal map}, and the expression for $d_{0}$ given by \eqref{d_ell in terms of beta_j}, after some calculations, we obtain
\begin{multline}\label{lol5}
\partial_{s_k}\log F_{\alpha}(r\vec{x},\vec{s}) = \partial_{s_{k}}d_{1}\sqrt{r} - \sum_{j=2}^{m} \bigg( 2\beta_{j} \partial_{s_{k}} \log \Lambda_{j} + \partial_{s_{k}} (\beta_{j}^{2}) - 2 i \alpha \arccos \bigg( \frac{\sqrt{x_{1}}}{\sqrt{x_{j}}}  \bigg)\partial_{s_{k}}\beta_{j} \bigg)  \\ + \sum_{j=2}^{m}  \frac{s_{j+1}-s_{j}}{2\pi i}\big( \Psi_{j,11}\partial_{s_{k}}\Psi_{j,21}-\Psi_{j,21}\partial_{s_{k}}\Psi_{j,11} \big) + \bigO\bigg( \frac{\log r}{\sqrt{r}} \bigg),
\end{multline}
as $r \to + \infty$. Using the explicit expressions for $\Psi_{j,11}$ and $\Psi_{j,21}$ (see \eqref{Psi j entries s1 = 0}) together with the relation \eqref{Psi_j first column connection formula s1 = 0}, we have
\begin{equation}\label{lol3}
\sum_{j=2}^{m} \frac{s_{j+1}-s_{j}}{2\pi i}\big( \Psi_{j,11}\partial_{s_{k}}\Psi_{j,21}-\Psi_{j,21}\partial_{s_{k}}\Psi_{j,11} \big) = \sum_{j=2}^{m}\beta_{j} \partial_{s_{k}} \log \frac{\Gamma(1+\beta_{j})}{\Gamma(1-\beta_{j})}.
\end{equation}
Also, using \eqref{def Lambda_j}, we have
\begin{equation}\label{lol4}
\sum_{j=2}^{m} - 2\beta_{j} \partial_{s_{k}} \log \Lambda_{j} = -2 \sum_{j=2}^{m} \beta_{j} \partial_{s_{k}}(\beta_{j}) \log \big(4\sqrt{r(x_{\smash{j}}-x_{1})}\big) -2\sum_{j=2}^{m} \beta_{j} \sum_{\substack{\ell = 2 \\ \ell \neq j}}^{m} \partial_{s_{k}}(\beta_{\ell})\log(\widetilde{T}_{\ell,j}).
\end{equation}
It will more convenient to integrate with respect to $\beta_{2},...,\beta_{m}$ instead of $s_{2},...,s_{m}$. Therefore, we define
\begin{equation}\label{def of F tilde s1 = 0}
\widetilde{F}_{\alpha}(r \vec{x}, \vec{\beta}) = F_{\alpha}(r \vec{x},\vec{s}),
\end{equation}
where $\vec{\beta} = (\beta_{2},...,\beta_{m})$ and $\vec{s} = (s_{2},...,s_{m})$ are related via the relations \eqref{def of beta_j}.
By substituting \eqref{lol3} and \eqref{lol4} into \eqref{lol5}, and by writing the derivative with respect to $\beta_{k}$ instead of $s_{k}$, we obtain
\begin{multline}\label{lol5 part 2}
\partial_{\beta_k}\log \widetilde{F}_{\alpha}(r\vec{x},\vec{\beta}) = \partial_{\beta_{k}}d_{1} \sqrt{r} -2 \sum_{j=2}^{m} \beta_{j} \partial_{\beta_{k}}(\beta_{j}) \log \big(4\sqrt{r(x_{\smash{j}}-x_{1})}\big)+ \sum_{j=2}^{m} 2 i \alpha \arccos \bigg( \frac{\sqrt{x_{1}}}{\sqrt{x_{j}}}  \bigg)\partial_{\beta_{k}}(\beta_{j}) \\ -2\sum_{j=2}^{m} \beta_{j} \sum_{\substack{\ell = 2 \\ \ell \neq j}}^{m} \partial_{\beta_{k}}(\beta_{\ell})\log(\widetilde{T}_{\ell,j}) - \sum_{j=2}^{m} \partial_{\beta_{k}}(\beta_{j}^{2}) + \sum_{j=2}^{m} \beta_{j} \partial_{\beta_{k}} \log \frac{\Gamma(1+\beta_{j})}{\Gamma(1-\beta_{j})} + \bigO\Big(\frac{\log r}{\sqrt{r}}\Big),
\end{multline}
as $r \to + \infty$. Using the value of $d_{1}$ in \eqref{d_ell in terms of beta_j} and the value of $c_{-x_{j}}$ in \eqref{expansion conformal map}, the above asymptotics can be rewritten more explicitly as follows
\begin{multline}\label{DIFF IDENTITY}
\partial_{\beta_k}\log \widetilde{F}_{\alpha}(r\vec{x},\vec{\beta}) = -2i \sqrt{r(x_{k}-x_{1})} -2 \beta_{k} \log\big(4\sqrt{r(x_{k}-x_{1})}\big) + 2 i \alpha \arccos \bigg( \frac{\sqrt{x_{1}}}{\sqrt{x_{k}}}  \bigg) \\ - 2 \sum_{\substack{j=2 \\ j \neq k}}^{m}\beta_{j} \log(\widetilde{T}_{k,j}) - 2\beta_{k} + \beta_{k} \partial_{\beta_{k}} \log \frac{\Gamma(1+\beta_{k})}{\Gamma(1-\beta_{k})} + \bigO\Big( \frac{\log r}{\sqrt{r}} \Big).
\end{multline}

\subsection{Integration of the differential identity}
By the steepest descent of Section \ref{Section: Steepest descent with s1=0} (see in particular the discussion in Section \ref{subsection Small norm s1 = 0}), the asymptotics \eqref{DIFF IDENTITY} are valid uniformly for $\beta_{2},...,\beta_{m}$ in compact subsets of $i \mathbb{R}$. First, we use \eqref{DIFF IDENTITY} with $\beta_{3} = 0 = \beta_{4} = ... = \beta_{m}$, and we integrate in $\beta_{2}$ from $\beta_{2} = 0$ to an arbitrary $\beta_{2} \in i \mathbb{R}$. Let us use the notations $\vec{\beta}_{2} = (\beta_{2},0,...,0)$ and $\vec{0} = (0,0,...,0)$. After integration (using \eqref{integral of Gamma with Barnes}), we obtain
\begin{multline}
\log \frac{\widetilde{F}_{\alpha}(r\vec{x},\vec{\beta}_{2})}{\widetilde{F}_{\alpha}(r\vec{x},\vec{0})} = -2i \beta_{2} \sqrt{r(x_{2}-x_{1})} - \beta_{2}^{2} \log\big(4\sqrt{r(x_{2}-x_{1})}\big) + 2 i \alpha \beta_{2} \arccos \bigg( \frac{\sqrt{x_{1}}}{\sqrt{x_{2}}}  \bigg) \\ +\log(G(1+\beta_{2})G(1-\beta_{2})) + \bigO \bigg( \frac{\log r}{\sqrt{r}} \bigg),
\end{multline}
as $r \to + \infty$. Now, we use \eqref{DIFF IDENTITY} with $\beta_{4}=...=\beta_{m} = 0$, $\beta_{2}$ fixed but not necessarily $0$, and we integrate in $\beta_{3}$. With the notation $\vec{\beta}_{3} = (\beta_{2},\beta_{3},0,...,0)$, as $r \to + \infty$ we obtain
\begin{multline}
\log \frac{\widetilde{F}_{\alpha}(r\vec{x},\vec{\beta}_{3})}{\widetilde{F}_{\alpha}(r\vec{x},\vec{\beta}_{2})} = -2i \beta_{3} \sqrt{r(x_{3}-x_{1})} - \beta_{3}^{2} \log\big(4\sqrt{r(x_{3}-x_{1})}\big) + 2 i \alpha \beta_{3} \arccos \bigg( \frac{\sqrt{x_{1}}}{\sqrt{x_{3}}}  \bigg) \\ - 2 \beta_{2}\beta_{3} \log(\widetilde{T}_{3,2}) +\log(G(1+\beta_{3})G(1-\beta_{3})) + \bigO \bigg( \frac{\log r}{\sqrt{r}} \bigg).
\end{multline}
By integrating successively in $\beta_{4},...,\beta_{m}$, and then by summing the expressions, we obtain
\begin{multline}
\log \frac{\widetilde{F}_{\alpha}(r\vec{x},\vec{\beta})}{\widetilde{F}_{\alpha}(r\vec{x},\vec{0})} = - \sum_{j=2}^{m} 2i \beta_{j} \sqrt{{r(x_{\smash{j}}-x_{1})}} - \sum_{j=2}^{m} \beta_{j}^{2} \log \big(4\sqrt{{r(x_{\smash{j}}-x_{1})}}\big) + \sum_{j=2}^{m} 2 i \alpha \beta_{j} \arccos \bigg( \frac{\sqrt{x_{1}}}{\sqrt{x_{j}}}  \bigg) \\ - 2 \sum_{2 \leq j < k \leq m} \beta_{j}\beta_{k} \log(\widetilde{T}_{j,k}) + \sum_{j=2}^{m} \log (G(1+\beta_{j})G(1-\beta_{j})) + \bigO \bigg( \frac{\log r}{\sqrt{r}} \bigg),
\end{multline}
as $r \to + \infty$. By adding the above asymptotics to \eqref{asymp m=1 with s1 = 0}, this finishes the proof of Theorem \ref{thm:s1=0} (after identifying $u_{j}=-2\pi i \beta_{j}$).

\subsection{Heuristic discussion of the asymptotics as $r \to + \infty$ and $s_{1} \to 0$}\label{subsection: heuristic discussion}
Uniform asymptotics for $F_{\alpha}(r\vec{x},\vec{s})$ as $r\to + \infty$ and simultaneously $s_{1} \to 0$ should describe the transition between the asymptotics of Theorems \ref{thm:s1 neq 0} and \ref{thm:s1=0}. As can be seen from \eqref{thm product s1 neq 0}, \eqref{def of Ec} and \eqref{F asymptotics thm s1 = 0} (and also from \eqref{asymp m=1 with s1 = 0} and \eqref{asymp m=1 with s1 > 0}), the leading term of $\log F_{\alpha}(r\vec{x},\vec{s})$ is of order $\bigO(\sqrt{r})$ if $s_{1}$ is bounded away from $0$, while it is of order $\bigO(r)$ if $s_{1} = 0$. On the level of the RH analysis, this indicates that a new $g$-function is needed (which should interpolate between \eqref{def of g s1 neq 0} and \eqref{def of g s1=0}). A similar transition has been studied for the sine point process in the series of papers \cite{BotDeiItsKra1, BDIK2017,BDIK2019}, and (with less depth) for the Circular Unitary Ensemble in \cite{ChCl,ChCl2}. In these works, the global parametrix $P^{(\infty)}$ is constructed in terms of elliptic $\theta$-functions and the asymptotics that describe the transition are oscillatory. By analogy, we expect that asymptotics for $F_{\alpha}(r\vec{x},\vec{s})$ as $r\to + \infty$ and simultaneously $s_{1} \to 0$ will also be described in terms of elliptic $\theta$-functions. A rigorous analysis is however very delicate, and is not addressed in this paper.

\section{Appendix}\label{Section:Appendix}
In this section, we recall two well-known RH problems: 1) the Bessel model RH problem, which depends on a parameter $\alpha > -1$ and whose solution is denoted by $\Phi_{\mathrm{Be}}(\cdot) = \Phi_{\mathrm{Be}}(\cdot;\alpha)$, and 2) the confluent hypergeometric model RH problem, which depends on a parameter $\beta \in i \mathbb{R}$ and whose solution is denoted by $\Phi_{\mathrm{HG}}(\cdot)=\Phi_{\mathrm{HG}}(\cdot;\beta)$.
\subsection{Bessel model RH problem}\label{subsection:Model Bessel}
\begin{itemize}
\item[(a)] $\Phi_{\mathrm{Be}} : \mathbb{C} \setminus \Sigma_{\mathrm{Be}} \to \mathbb{C}^{2\times 2}$ is analytic, where
$\Sigma_{\mathrm{Be}}$ is shown in Figure \ref{figBessel}.
\item[(b)] $\Phi_{\mathrm{Be}}$ satisfies the jump conditions
\begin{equation}\label{Jump for P_Be}
\begin{array}{l l} 
\Phi_{\mathrm{Be},+}(z) = \Phi_{\mathrm{Be},-}(z) \begin{pmatrix}
0 & 1 \\ -1 & 0
\end{pmatrix}, & z \in \mathbb{R}^{-}, \\

\Phi_{\mathrm{Be},+}(z) = \Phi_{\mathrm{Be},-}(z) \begin{pmatrix}
1 & 0 \\ e^{\pi i \alpha} & 1
\end{pmatrix}, & z \in e^{ \frac{2\pi i}{3} }  \mathbb{R}^{+}, \\

\Phi_{\mathrm{Be},+}(z) = \Phi_{\mathrm{Be},-}(z) \begin{pmatrix}
1 & 0 \\ e^{-\pi i \alpha} & 1
\end{pmatrix}, & z \in e^{ -\frac{2\pi i}{3} }  \mathbb{R}^{+}. \\
\end{array}
\end{equation}
\item[(c)] As $z \to \infty$, $z \notin \Sigma_{\mathrm{Be}}$, we have
\begin{equation}\label{large z asymptotics Bessel}
\Phi_{\mathrm{Be}}(z) = ( 2\pi z^{\frac{1}{2}} )^{-\frac{\sigma_{3}}{2}}N
\left(I+\frac{ \Phi_{\mathrm{Be},1}(\alpha)}{z^{\frac{1}{2}}} + \bigO(z^{-1})\right) e^{2 z^{\frac{1}{2}}\sigma_{3}},
\end{equation}
where $\ds \Phi_{\mathrm{Be},1}(\alpha) = \frac{1}{16}\begin{pmatrix}
-(1+4\alpha^{2}) & -2i \\ -2i & 1+4\alpha^{2}
\end{pmatrix}$.
\item[(d)] As $z$ tends to 0, the behavior of $\Phi_{\mathrm{Be}}(z)$ is
\begin{equation}\label{local behaviour near 0 of P_Be}
\begin{array}{l l}
\displaystyle \Phi_{\mathrm{Be}}(z) = \left\{ \begin{array}{l l}
\begin{pmatrix}
\bigO(1) & \bigO(\log z) \\
\bigO(1) & \bigO(\log z) 
\end{pmatrix}, & |\arg z| < \frac{2\pi}{3}, \\
\begin{pmatrix}
\bigO(\log z) & \bigO(\log z) \\
\bigO(\log z) & \bigO(\log z) 
\end{pmatrix}, & \frac{2\pi}{3}< |\arg z| < \pi,
\end{array}  \right., & \displaystyle \mbox{ if } \alpha = 0, \\[0.8cm]
\displaystyle \Phi_{\mathrm{Be}}(z) = \left\{ \begin{array}{l l}
\begin{pmatrix}
\bigO(1) & \bigO(1) \\
\bigO(1) & \bigO(1) 
\end{pmatrix}z^{\frac{\alpha}{2}\sigma_{3}}, & |\arg z | < \frac{2\pi}{3}, \\
\begin{pmatrix}
\bigO(z^{-\frac{\alpha}{2}}) & \bigO(z^{-\frac{\alpha}{2}}) \\
\bigO(z^{-\frac{\alpha}{2}}) & \bigO(z^{-\frac{\alpha}{2}}) 
\end{pmatrix}, & \frac{2\pi}{3}<|\arg z | < \pi,
\end{array} \right. , & \displaystyle \mbox{ if } \alpha > 0, \\[0.8cm]
\displaystyle \Phi_{\mathrm{Be}}(z) = \begin{pmatrix}
\bigO(z^{\frac{\alpha}{2}}) & \bigO(z^{\frac{\alpha}{2}}) \\
\bigO(z^{\frac{\alpha}{2}}) & \bigO(z^{\frac{\alpha}{2}}) 
\end{pmatrix}, & \displaystyle \mbox{ if } \alpha < 0.
\end{array}
\end{equation}
\end{itemize}
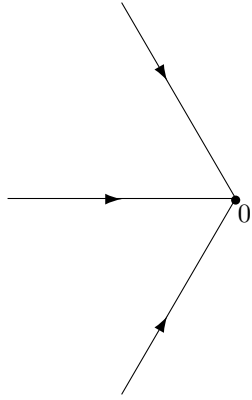
\begin{figure}[t]
    \begin{center}
    \setlength{\unitlength}{1truemm}
    \begin{picture}(100,55)(-5,10)
        \put(50,40){\line(-1,0){30}}
        \put(50,39.8){\thicklines\circle*{1.2}}
        \put(50,40){\line(-0.5,0.866){15}}
        \put(50,40){\line(-0.5,-0.866){15}}
        \put(50.3,36.8){$0$}
        \put(35,39.9){\thicklines\vector(1,0){.0001}}
        \put(41,55.588){\thicklines\vector(0.5,-0.866){.0001}}
        \put(41,24.412){\thicklines\vector(0.5,0.866){.0001}}
    \end{picture}
    \caption{\label{figBessel}The jump contour $\Sigma_{\mathrm{Be}}$ for $\Phi_{\mathrm{Be}}$.}
\end{center}
\end{figure}
This RH problem was introduced and solved in \cite{KMcLVAV}. Its unique solution is given by 
\begin{equation}\label{Psi explicit}
\Phi_{\mathrm{Be}}(z)=
\begin{cases}
\begin{pmatrix}
I_{\alpha}(2 z^{\frac{1}{2}}) & \frac{ i}{\pi} K_{\alpha}(2 z^{\frac{1}{2}}) \\
2\pi i z^{\frac{1}{2}} I_{\alpha}^{\prime}(2 z^{\frac{1}{2}}) & -2 z^{\frac{1}{2}} K_{\alpha}^{\prime}(2 z^{\frac{1}{2}})
\end{pmatrix}, & |\arg z | < \frac{2\pi}{3}, \\

\begin{pmatrix}
\frac{1}{2} H_{\alpha}^{(1)}(2(-z)^{\frac{1}{2}}) & \frac{1}{2} H_{\alpha}^{(2)}(2(-z)^{\frac{1}{2}}) \\
\pi z^{\frac{1}{2}} \left( H_{\alpha}^{(1)} \right)^{\prime} (2(-z)^{\frac{1}{2}}) & \pi z^{\frac{1}{2}} \left( H_{\alpha}^{(2)} \right)^{\prime} (2(-z)^{\frac{1}{2}})
\end{pmatrix}e^{\frac{\pi i \alpha}{2}\sigma_{3}}, & \frac{2\pi}{3} < \arg z < \pi, \\

\begin{pmatrix}
\frac{1}{2} H_{\alpha}^{(2)}(2(-z)^{\frac{1}{2}}) & -\frac{1}{2} H_{\alpha}^{(1)}(2(-z)^{\frac{1}{2}}) \\
-\pi z^{\frac{1}{2}} \left( H_{\alpha}^{(2)} \right)^{\prime} (2(-z)^{\frac{1}{2}}) & \pi z^{\frac{1}{2}} \left( H_{\alpha}^{(1)} \right)^{\prime} (2(-z)^{\frac{1}{2}})
\end{pmatrix}e^{-\frac{\pi i \alpha}{2}\sigma_{3}}, & -\pi < \arg z < -\frac{2\pi}{3},
\end{cases}
\end{equation}
where $H_{\alpha}^{(1)}$ and $H_{\alpha}^{(2)}$ are the Hankel functions of the first and second kind, and $I_\alpha$ and $K_\alpha$ are the modified Bessel functions of the first and second kind.

\vspace{0.2cm}\hspace{-0.55cm}A direct analysis of the RH problem for $\Phi_{\mathrm{Be}}$ shows that in a neighborhood of $z$ we have
\begin{equation}\label{precise asymptotics of Phi Bessel near 0}
\Phi_{\mathrm{Be}}(z;\alpha) = \Phi_{\mathrm{Be},0}(z;\alpha)z^{\frac{\alpha}{2}\sigma_{3}}\begin{pmatrix}
1 & h(z) \\ 0 & 1
\end{pmatrix} H_{0}(z),
\end{equation}
where $H_{0}$ is given by \eqref{def of H}, $h$ by \eqref{def of h}, and $\Phi_{\mathrm{Be},0}$ is analytic in a neighborhood of $0$. After some computation using asymptotics of Bessel functions near the origin (see \cite[Chapter 10.30(i)]{NIST}), we obtain
\begin{equation}\label{precise matrix at 0 in asymptotics of Phi Bessel near 0}
\Phi_{\mathrm{Be},0}(0;\alpha) = \left\{ 
\begin{array}{l l}
\begin{pmatrix}
\frac{1}{\Gamma(1+\alpha)} & \frac{i \Gamma(\alpha)}{2\pi} \\
\frac{i \pi}{\Gamma(\alpha)} & \frac{\Gamma(1+\alpha)}{2}
\end{pmatrix}, & \mbox{if } \alpha \neq 0, \\[0.5cm]
\begin{pmatrix}
1 & \frac{\gamma_{\mathrm{E}}}{\pi i} \\ 0 & 1
\end{pmatrix}, & \mbox{if } \alpha = 0,
\end{array} \right.
\end{equation}
where $\gamma_{\mathrm{E}}$ is Euler's gamma constant.
\subsection{Confluent hypergeometric model RH problem}\label{subsection: model RHP with HG functions}
\begin{itemize}
\item[(a)] $\Phi_{\mathrm{HG}} : \mathbb{C} \setminus \Sigma_{\mathrm{HG}} \rightarrow \mathbb{C}^{2 \times 2}$ is analytic, where $\Sigma_{\mathrm{HG}}$ is shown in Figure \ref{Fig:HG}.
\item[(b)] For $z \in \Gamma_{k}$ (see Figure \ref{Fig:HG}), $k = 1,...,6$, $\Phi_{\mathrm{HG}}$ has the jump relations
\begin{equation}\label{jumps PHG3}
\Phi_{\mathrm{HG},+}(z) = \Phi_{\mathrm{HG},-}(z)J_{k},
\end{equation}
where
\begin{align*}
& J_{1} = \begin{pmatrix}
0 & e^{-i\pi \beta} \\ -e^{i\pi\beta} & 0
\end{pmatrix}, \quad J_{4} = \begin{pmatrix}
0 & e^{i\pi\beta} \\ -e^{-i\pi\beta} & 0
\end{pmatrix}, \\
& J_{2} = \begin{pmatrix}
1 & 0 \\ e^{i\pi\beta} & 1
\end{pmatrix}\hspace{-0.1cm}, \hspace{-0.3cm} \quad J_{3} = \begin{pmatrix}
1 & 0 \\ e^{-i\pi\beta} & 1
\end{pmatrix}\hspace{-0.1cm}, \hspace{-0.3cm} \quad J_{5} = \begin{pmatrix}
1 & 0 \\ e^{-i\pi\beta} & 1
\end{pmatrix}\hspace{-0.1cm}, \hspace{-0.3cm} \quad J_{6} = \begin{pmatrix}
1 & 0 \\ e^{i\pi\beta} & 1
\end{pmatrix}.
\end{align*}
\item[(c)] As $z \to \infty$, $z \notin \Sigma_{\mathrm{HG}}$, we have
\begin{equation}\label{Asymptotics HG}
\Phi_{\mathrm{HG}}(z) = \left( I +  \frac{\Phi_{\mathrm{HG},1}(\beta)}{z} + \bigO(z^{-2}) \right) z^{-\beta\sigma_{3}}e^{-\frac{z}{2}\sigma_{3}}\left\{ \begin{array}{l l}
\displaystyle e^{i\pi\beta  \sigma_{3}}, & \displaystyle \frac{\pi}{2} < \arg z <  \frac{3\pi}{2}, \\
\begin{pmatrix}
0 & -1 \\ 1 & 0
\end{pmatrix}, & \displaystyle -\frac{\pi}{2} < \arg z < \frac{\pi}{2},
\end{array} \right.
\end{equation}
where 
\begin{equation}\label{def of tau}
\Phi_{\mathrm{HG},1}(\beta) = \beta^{2} \begin{pmatrix}
-1 & \tau(\beta) \\ - \tau(-\beta) & 1
\end{pmatrix}, \qquad \tau(\beta) = \frac{- \Gamma\left( -\beta \right)}{\Gamma\left( \beta + 1 \right)}.
\end{equation}
In \eqref{Asymptotics HG}, the root is defined by $z^{\beta} = |z|^{\beta}e^{i\beta \arg z}$ with $\arg z \in (-\frac{\pi}{2},\frac{3\pi}{2})$.

As $z \to 0$, we have
\begin{equation}\label{lol 35}
\Phi_{\mathrm{HG}}(z) = \left\{ \begin{array}{l l}
\begin{pmatrix}
\bigO(1) & \bigO(\log z) \\
\bigO(1) & \bigO(\log z)
\end{pmatrix}, & \mbox{if } z \in II \cup V, \\
\begin{pmatrix}
\bigO(\log z) & \bigO(\log z) \\
\bigO(\log z) & \bigO(\log z)
\end{pmatrix}, & \mbox{if } z \in I\cup III \cup IV \cup VI.
\end{array} \right.
\end{equation}
\end{itemize}
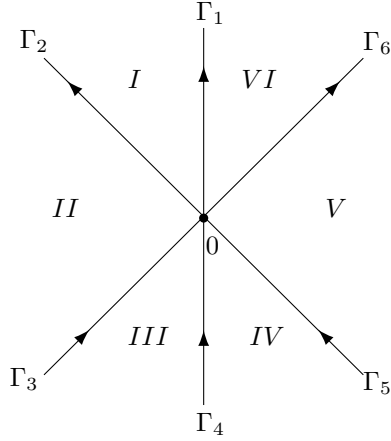
\begin{figure}[t!]
    \begin{center}
    \setlength{\unitlength}{1truemm}
    \begin{picture}(100,55)(-5,10)
             
        \put(50,39.8){\thicklines\circle*{1.2}}
        \put(50,40){\line(-0.5,0.5){21}}
        \put(50,40){\line(-0.5,-0.5){21}}
        \put(50,40){\line(0.5,0.5){21}}
        \put(50,40){\line(0.5,-0.5){21}}
        \put(50,40){\line(0,1){25}}
        \put(50,40){\line(0,-1){25}}
        
        \put(50.3,35){$0$}
        \put(71,62){$\Gamma_6$}        
        \put(49,66){$\Gamma_1$}        
        \put(25.8,62.3){$\Gamma_2$}        
        \put(24.5,17.5){$\Gamma_3$}
        \put(49,11.5){$\Gamma_4$}
        \put(71,17){$\Gamma_5$}        
        
        \put(32,58){\thicklines\vector(-0.5,0.5){.0001}}
        \put(35,25){\thicklines\vector(0.5,0.5){.0001}}
        \put(68,58){\thicklines\vector(0.5,0.5){.0001}}
        \put(65,25){\thicklines\vector(-0.5,0.5){.0001}}
        \put(50,60){\thicklines\vector(0,1){.0001}}
        \put(50,25){\thicklines\vector(0,1){.0001}}
        \put(40,57){$I$}
        \put(30,40){$II$}
        \put(40,23){$III$}
        \put(56,23){$IV$}
        \put(66,40){$V$}
        \put(55,57){$VI$}
    \end{picture}
    \caption{\label{Fig:HG}The jump contour $\Sigma_{\mathrm{HG}}$ for $\Phi_{\mathrm{HG}}$. The ray $\Gamma_{k}$ is oriented from $0$ to $\infty$, and forms an angle with $\mathbb{R}^{+}$ which is a multiple of $\frac{\pi}{4}$.}
\end{center}
\end{figure}
This model RH problem was first introduced and solved explicitly in \cite{ItsKrasovsky}. Consider the matrix
\begin{equation}\label{phi_HG}
\widehat{\Phi}_{\mathrm{HG}}(z) = \begin{pmatrix}
\Gamma(1 -\beta)G(\beta; z) & -\frac{\Gamma(1 -\beta)}{\Gamma(\beta)}H(1-\beta;ze^{-i\pi }) \\
\Gamma(1 +\beta)G(1+\beta;z) & H(-\beta;ze^{-i\pi })
\end{pmatrix},
\end{equation}
where $G$ and $H$ are related to the Whittaker functions:
\begin{equation}\label{relation between G and H and Whittaker}
G(a;z) = \frac{M_{\kappa,\mu}(z)}{\sqrt{z}}, \quad H(a;z) = \frac{W_{\kappa,\mu}(z)}{\sqrt{z}}, \quad \mu = 0, \quad \kappa = \frac{1}{2}-a.
\end{equation}
The solution $\Phi_{\mathrm{HG}}$ is given by
\begin{equation}\label{model RHP HG in different sector}
\Phi_{\mathrm{HG}}(z) = \left\{ \begin{array}{l l}
\widehat{\Phi}_{\mathrm{HG}}(z)J_{2}^{-1}, & \mbox{ for } z \in I, \\
\widehat{\Phi}_{\mathrm{HG}}(z), & \mbox{ for } z \in II, \\
\widehat{\Phi}_{\mathrm{HG}}(z)J_{3}^{-1}, & \mbox{ for } z \in III, \\
\widehat{\Phi}_{\mathrm{HG}}(z)J_{2}^{-1}J_{1}^{-1}J_{6}^{-1}J_{5}, & \mbox{ for } z \in IV, \\
\widehat{\Phi}_{\mathrm{HG}}(z)J_{2}^{-1}J_{1}^{-1}J_{6}^{-1}, & \mbox{ for } z \in V, \\
\widehat{\Phi}_{\mathrm{HG}}(z)J_{2}^{-1}J_{1}^{-1}, & \mbox{ for } z \in VI. \\
\end{array} \right.
\end{equation}
We need in the present paper a better knowledge than \eqref{lol 35}. From \cite[Section 13.14 (iii)]{NIST}, as $z \to 0$ we have
\begin{align*}
& G(\beta;z) = 1+\bigO(z), \qquad G(1+\beta;z) = 1+ \bigO(z), \\
& H(1-\beta;z) = \frac{-1}{\Gamma(1-\beta)} \left( \log z + \frac{\Gamma^{\prime}(1-\beta)}{\Gamma(1-\beta)} + 2\gamma_{\mathrm{E}} \right) + \bigO(z \log z), \\
& H(-\beta;z) = \frac{-1}{\Gamma(-\beta)} \left( \log z + \frac{\Gamma^{\prime}(-\beta)}{\Gamma(-\beta)} + 2\gamma_{\mathrm{E}} \right) + \bigO(z \log z),
\end{align*}
where $\gamma_{\mathrm{E}}$ is Euler's gamma constant. Using the connection formula $\Gamma(z)\Gamma(1-z) = \frac{\pi}{\sin (\pi z)} = -\Gamma(-z)\Gamma(1+z)$, as $z \to 0$, $z \in II$, we have
\begin{equation}\label{precise asymptotics of Phi HG near 0}
\widehat{\Phi}_{\mathrm{HG}}(z) = \begin{pmatrix}
\Psi_{11} & \Psi_{12} \\ \Psi_{21} & \Psi_{22}
\end{pmatrix} (I + \bigO(z)) \begin{pmatrix}
1 & \frac{\sin (\pi \beta)}{\pi} \log z \\
0 & 1
\end{pmatrix},
\end{equation}
where in the above expression
\begin{equation}
\log z = \log |z| + i \arg z, \qquad \arg z \in \Big(-\frac{\pi}{2},\frac{3\pi}{2}\Big),
\end{equation}
and
\begin{align*}
& \Psi_{11} = \Gamma(1-\beta), \qquad \Psi_{12} = \frac{1}{\Gamma(\beta)} \left( \frac{\Gamma^{\prime}(1-\beta)}{\Gamma(1-\beta)}+2\gamma_{\mathrm{E}} - i \pi \right), \\
& \Psi_{21} = \Gamma(1+\beta), \qquad \Psi_{22} = \frac{-1}{\Gamma(-\beta)} \left( \frac{\Gamma^{\prime}(-\beta)}{\Gamma(-\beta)} + 2\gamma_{\mathrm{E}} - i \pi \right).
\end{align*}

\section*{Acknowledgements}
The author is grateful to the three anonymous referees for their careful reading, and for their suggestions for improvement. This work was supported by the Swedish Research Council, Grant No. 2015-05430.

\end{document}